\documentclass{article}

\newcommand{\helper}{h}
\newcommand{\model}{f}
\newcommand{\DT}{D_f}
\newcommand{\DG}{D_G}
\newcommand{\Din}{D_{\textrm{in}}}
\newcommand{\Dout}{D_{\textrm{out}}}
\newcommand{\Dintr}{D^{\textrm{tr}}_{\textrm{in}}}
\newcommand{\Douttr}{D^{\textrm{tr}}_{\textrm{out}}}
\newcommand{\Dinte}{D^{\textrm{te}}_{\textrm{in}}}
\newcommand{\Doutte}{D^{\textrm{te}}_{\textrm{out}}}
\newcommand{\Dinouttrte}{D^{\{\textrm{tr},\textrm{te}\}}_{\{\textrm{in},\textrm{out}\}}}
\usepackage[textsize=tiny]{todonotes}
\newcommand{\tp}{\textrm{tp}}
\newcommand{\pr}{\textrm{prec}}
\newcommand{\clb}{c_{\textrm{lb}}}
\newcommand{\cpepslb}{\{c\!+\!\epsilon\}_{\textrm{lb}}}
\def\acronym{\texttt{\textbf{PANORAMIA}}\xspace}
\usepackage[final]{neurips_2024}

\usepackage[utf8]{inputenc} % allow utf-8 input
\usepackage[T1]{fontenc}    % use 8-bit T1 fonts
\usepackage{hyperref}       % hyperlinks
\usepackage{url}            % simple URL typesetting
\usepackage{booktabs}       % professional-quality tables
\usepackage{amsfonts}       % blackboard math symbols
\usepackage{nicefrac}       % compact symbols for 1/2, etc.
\usepackage{microtype}      % microtypography
\usepackage{xcolor}         % colors

\usepackage{microtype}
\usepackage{graphicx}
\usepackage{subfigure}
\usepackage{booktabs} % for professional tables
\usepackage{subcaption}
\usepackage{enumitem}
\usepackage{multirow}

\usepackage{xspace}
\usepackage{adjustbox}

\usepackage{hyperref}
 \definecolor{myblue}{HTML}{FF6633}
\definecolor{mygreen}{HTML}{3399FF}
\definecolor{mypurple}{HTML}{9467BD}
\definecolor{mycyan}{HTML}{17BECF}

\hypersetup{
    colorlinks=true,
    linkcolor=myblue,      % color of internal links (sections)
    citecolor=mygreen,     % color of links to bibliography
    filecolor=mypurple,    % color of file links
    urlcolor=mycyan        % color of external links
}

\usepackage{algorithm}
\usepackage{algorithmicx}
\usepackage{algpseudocode}
\usepackage{subcaption}
\usepackage{booktabs}
\usepackage{xcolor}

\usepackage{fixltx2e}
\usepackage{stfloats}
\usepackage{placeins}

\usepackage{float}
\usepackage{afterpage}
\usepackage{amsmath}
\usepackage{amssymb}
\usepackage{mathtools}
\usepackage{amsthm}
\usepackage{amsmath,amsfonts,bm}
\usepackage{dsfont}
\usepackage{amsthm}

\newtheorem{definition}{Definition}
\newtheorem{proposition}{Proposition}
\newtheorem{corollary}{Corollary}
\newtheorem{lemma}{Lemma}

\newtheoremstyle{TheoremNum}
{\topsep}{\topsep}              %%% space between body and thm
{\itshape}                      %%% Thm body font
{}                              %%% Indent amount (empty = no indent)
{\bfseries}                     %%% Thm head font
{.}                             %%% Punctuation after thm head
{ }                             %%% Space after thm head
{\thmname{#1}\thmnote{ \bfseries #3}}%%% Thm head spec
\theoremstyle{TheoremNum}
\newtheorem{numberedprop}{Proposition}

\newcommand{\xin}{x^{\textrm{in}}}
\newcommand{\xgen}{x^{\textrm{gen}}}
\newcommand{\1}{\mathds{1}}
\newcommand{\ind}{\perp\!\!\!\!\perp} 
 
\def\eqref#1{equation~\ref{#1}}
\DeclareMathAlphabet{\mathsfit}{\encodingdefault}{\sfdefault}{m}{sl}
\SetMathAlphabet{\mathsfit}{bold}{\encodingdefault}{\sfdefault}{bx}{n}

\def\gD{{\mathcal{D}}}

\def\gF{{\mathcal{F}}}
\def\gG{{\mathcal{G}}}
\def\gH{{\mathcal{H}}}

\def\gO{{\mathcal{O}}}

\def\gX{{\mathcal{X}}}
\def\gY{{\mathcal{Y}}}

\def\sP{{\mathbb{P}}}

\def\sR{{\mathbb{R}}}

\newcommand{\E}{\mathbb{E}}

\newcommand{\indep}{\perp\!\!\!\!\perp}

\usepackage[capitalize,noabbrev]{cleveref}%\baselineskip 
\usepackage{etoolbox}

\NewEnviron{showIfProofs}{%
  \iftoggle{showProofs}{\BODY}{}%
}
\NewEnviron{hideIfProofs}{%
  \iftoggle{showProofs}{}{\BODY}%
}

\newtoggle{showProofs}
\title{\acronym: Privacy Auditing of Machine Learning Models without Retraining}

\author{%
  Mishaal Kazmi\thanks{equal contribution} \\
  University of British Columbia\\
  \And
  Hadrien Lautraite$^*$ \\
 University du Québec à Montréal \\
  \And
  Alireza Akbari$^*$ \\
  Simon Fraser University \\
  \And
  Qiaoyue Tang$^*$ \\
  University of British Columbia \\
  \And
  Mauricio Soroco \\
  University of British Columbia \\
  \And
  Tao Wang \\
  Simon Fraser University \\
  \And
  Sébastien Gambs \\
  University du Québec à Montréal \\
  \And
  Mathias L\'ecuyer \\
  University of British Columbia \\
}
\usepackage{wrapfig}
\usepackage{afterpage}

\begin{document}

\maketitle
\begin{abstract}
%% Old version
% We introduce a privacy auditing scheme for ML models that relies on membership inference attacks using generated data as "non-members".
% This scheme, \acronym, quantifies the privacy leakage for large-scale ML models without control of the training process or model re-training and only requires access to a subset of the training data. \acronym also does not depend on the withdrawal of in-distribution data from the ML model to construct non-member data. To demonstrate its applicability, we evaluate our auditing scheme across multiple ML domains, ranging from image and tabular data classification to large-scale language models.
%
% Neurips submission version
% We introduce \acronym, a privacy leakage measurement scheme for ML models that relies on membership inference attacks using generated data as ``non-members''. 
% \acronym does not modify the training data or training process, and only requires access to a subset of the training data.
% We evaluate \acronym on ML models for image and tabular data classification, as well as on large-scale language models.
%
%% New version for arxiv v2
We present \acronym, a privacy leakage measurement framework for machine learning models that relies on membership inference attacks using generated data as non-members.
By relying on generated non-member data, \acronym eliminates the common dependency of privacy measurement tools on in-distribution non-member data.
As a result, \acronym does not modify the model, training data, or training process, and only requires access to a subset of the training data.
We evaluate \acronym on ML models for image and tabular data classification, as well as on large-scale language models.
\end{abstract}

\section{Introduction}

Training Machine Learning (ML) models with Differential Privacy (DP) ~\cite{dwork2006calibrating}, such as with DP-SGD~\cite{abadi2016deep}, upper-bounds the worst-case privacy loss incurred by the training data. 
In contrast, privacy auditing aims to empirically lower-bound the privacy loss of a target ML model or algorithm. 
In practice, privacy audits usually rely on the link between DP and the performance of membership inference attacks (MIA) ~\cite{wasserman2010statistical,kairouz2015composition,dong2019gaussian}. 
At a high level, DP implies an upper-bound on the performance of MIAs, thus creating a high-performance MIA implies a lower-bound on the privacy loss.
Auditing schemes have proven valuable in many settings, such as to audit DP implementations~\cite{nasr2023tight}, or to study the tightness of DP algorithms~\cite{inst,lu2023general,steinke2023privacy}.
% It therefore becomes important to equip ML model owners with a privacy metric that is easy to interpret and one that can help them assess the training data privacy leakage of their ML model to ensure no sensitive information is leaked.
Typical privacy audits rely on retraining the model several times, each time guessing the membership of one sample~\cite{jagielski2020auditing,LIRA,zanellabéguelin2022bayesian}, which is computationally prohibitive, requires access to the target model (entire) training data as well as control over the training pipeline.

To circumvent these concerns, \cite{steinke2023privacy} proposed an auditing recipe (called O(1)) requiring only one training run (which could be the same as the actual training) by randomly including/excluding several samples (called auditing examples) into the training dataset of the target model. 
Later, the membership of the auditing examples are guessed for privacy audit. 
However, O(1) faces a few challenges in certain setups. 
First, canaries, which are datapoints specially crafted to be easy to detect when added to the training set \cite{nasr2023tight,lu2023general,steinke2023privacy}, cannot be employed as auditing examples when measuring the privacy leakage for data that a contributor actually puts into the model, and not a \emph{worst case data point}. 

This matches a setting in which individual data contributors (\emph{e.g.}, a hospital in a cross-site Federated Learning (FL) setting or a user of a service that trains ML models on users’ data) measure the leakage of their own (\emph{i.e.}, known) partial training data in the final trained model. 
%Seb: I have put in comments to save on space
%Further, users may not want to change the data they contribute but still might want to measure the privacy leakage for their data, more than that of the algorithm. 
% Generally, users of ML-driven applications in even the centralized settings cannot always measure privacy leakage {\em a posteriori}, as they cannot retroactively alter data and cannot retrain the model.
Second, O(1) also relies on the withdrawal of real data from the model to construct non-member in-distribution data. 
This is problematic in situations in which ML model owners need to conduct post-hoc audits, in which case it is too late for removal ~\cite{negoescu2023epsilon}. 
Moreover, in-distribution audits require much more data, thus withholding many data points (typically more than the test set size) and reducing model utility. 
%Seb: I have put in comments to save on space
This brings us to the question: \emph{Given an instance of a machine learning model as a target, can we perform post-hoc estimation of the privacy loss with regards to a known member subset of the target model training dataset?}

% as well as internal teams owning embeddings and predictive models, or advertising pipelines, split between advertisers and publishers.

% Adding canaries requires control of the training data, the training pipeline and requires retraining of the model. 
% This may not always be possible in large-scale ML pipelines, especially in Federated Learning (FL) settings where each client has to cooperate to retrain the full global model

% Indeed, in these settings, a single actor cannot control the whole training set or retrain the whole pipeline several times. 

% These shortcomings prevent the application of privacy audits of complex ML pipelines, especially when the model is split among different actors. 
% Second, participants co-training a model with Federated Learning (FL) cannot easily audit the privacy leak of the trained model, released to all participants, for their own data as this would require either poisoning the final model with canaries or collaborating honestly with all clients to train shadow models. 
% This prevents asking questions such as how much a model trained on one's data reveals this information. \emph{Can we measure the privacy leakage of a final trained ML model without the need to retrain, add canaries, or withdraw training data to create non-member data?}
% -------

\textbf{Our contributions.}
We propose \acronym, a new scheme for \emph{Privacy Auditing with NO Retraining by using Artificial data for Membership Inference Attacks}.
% We build on the key insight that the bottleneck limiting privacy audits' performance and applicability is the access to examples of both training set data (\emph{i.e.}, members) and comparable data that is not in the training set (\emph{i.e.}, non-members).
% Indeed, shadow models which are often used to train a strong MIA, are models retrained on different subsets of the training data. 
% %Seb: I put in comments to save on space
% %, to collect examples of models with and without each data point.
% On the other hand, canaries introduce new member data points, typically from a different distribution, that are easier to distinguish from non-members and reduce the requirement for a large number of members and non-members. \\
More precisely, we consider an auditor with access to a subset of the training data 
% (\emph{e.g.}, the training set of a participant in FL or an ML model user's own data)
and introduce a new alternative for accessing non-members: using synthetic datapoints from a generative model trained on the member data, unlocking the limit on non-member data.
\acronym uses this generated data, together with known members, to train and evaluate a MIA attack on the target model to audit (\S\ref{sec:approach}).
We also adapt the theory of privacy audits, and show how \acronym can estimate the privacy loss (though not a lower-bound) of the target model with regards to the known member subset (\S\ref{sec:theory}).
An important benefit of \acronym is to perform privacy loss measurements with (1) no retraining the target ML model (\emph{i.e.}, we audit the end-model, not the training algorithm), (2) no alteration of the model, dataset, or training procedure, and (3) only partial knowledge of the training set.
We evaluate \acronym on CIFAR10 models and observe that overfitted models, larger models, and models with larger DP parameters have higher measured privacy leakage. We also demonstrate the applicability of our approach on the GPT-2 based model (\emph{i.e.}, WikiText dataset) and CelebA models.

\label{sec:intro}

\section{Background and Related Work}
\label{sec:background}

% In this section, we provide a brief background on differential privacy (DP), private machine learning, and techniques to audit the privacy guarantees claimed under DP. 

% \paragraph*{\textbf{Black-box privacy auditing}}  In our setting, we define black-box access for the adversary as only partial access to the model $f$ trained on a dataset $X$. 
% In our auditing system, the auditor is assumed to have access only to the loss values of $f$ on the data (\emph{i.e.}, he does not exploit the structure of the model itself). 
% In addition, the auditor knows a subset of the data the ML model $f$ was trained on.
% For instance, this would to the setting in which each client in a Federated Learning (FL) system would have knowledge of its own data but not the entire dataset the global model is trained on.
% Information is also provided to the adversary with respect to labels (\emph{i.e.}, member vs non-member) for real versus synthetic data.
% The objective here will be to see how well a classifier can distinguish between the two based on knowledge of the scores of loss values of $f$ vs lack thereof.
% %Seb: je pense qu'il faudrait donner plus de contexte sur cela

DP is the established privacy definition in the context of ML models, as well as for data analysis in general.
We focus on the pure DP definition to quantify privacy loss with well-understood semantics.
In a nutshell, DP is a property of a randomized mechanism (or computation) from datasets to an output space $\gO$, noted $M: \gD \rightarrow \gO$.
It is defined over neighboring datasets $D, D'$, differing by one element $x \in \gX$ (we use the add/remove neighboring definition), which is $D' = D \cup \{x\}$. 
Formally:
\begin{definition}[Differential Privacy~\cite{dwork2006calibrating}]\label{def:dp}
A mechanism $M: \gD \rightarrow \gO$ is $\epsilon$-DP if for any two neighbouring datasets $D$, $D' \in \gD$, and for any measurable output subset $O \subseteq \gO$ it holds that:
\begin{equation*}
P[M(D) \in O] \leq e^\varepsilon P[M(D') \in O] .
\end{equation*}
\end{definition}%
Since the neighbouring definition is symmetric, so is the DP definition, and we also have that $P[M(D') \in O] \leq e^\varepsilon P[M(D) \in O]$.
Intuitively, $\epsilon$ upper-bounds the worst-case contribution of any individual example to the distribution over outputs of the computation (\emph{i.e.}, the ML model learned). 
More formally, $\epsilon$ is an upper-bound on the {\em privacy loss} incurred by observing an output $o$, defined as $\big|\ln\left(\frac{\sP[M(D) = o]}{\sP[M(D') = o]}\right)\big|$, which quantifies how much an adversary can learn to distinguish $D$ and $D'$ based on observing output $o$ from $M$.
A smaller $\epsilon$ hence means higher privacy.
% A smaller $\epsilon$ means that adding or removing a data point cannot drastically change the output distribution, and hence provides higher privacy.
% Due to this, $\epsilon$ is said to be an upper bound on the {\em privacy loss} incurred by observing an output $o$, defined as $\big|\ln\left(\frac{\sP[M(D) = o]}{\sP[M(D') = o]}\right)\big|$, which upper-bounds how much an adversary can learn to distinguish $D$ and $D'$ based on observing output $o$ from $M$.

% \ml{true and nice but we don't really need and don't have space}
% DP exhibits several important properties, including composability, which can be used to bound the overall cost of a DP algorithm invoking several DP mechanisms; as well as resilience to post-processing, which  states that the output of an $\epsilon$-DP algorithm remains $\epsilon$-DP whatever computations are further performed on it, as long as they are data-independent.

{\bf DP,  MIA and privacy audits.} To audit a DP training algorithm $M$ that outputs a model $f$, one can perform a MIA on datapoint $x$, trying to distinguish between a neighboring training sets $D$ and $D'=D\cup \{x\}$. 
The MIA can be formalized as a hypothesis test to distinguish between $\gH_0 = D$ and $\gH_1 = D'$ using the output of the computation $f$. 
\citet{wasserman2010statistical,kairouz2015composition,dong2019gaussian} show that any such test at significance level $\alpha$ (False Positive Rate or FPR) 
% or probability of rejecting $\gH_0$ when it is true
 has power (True Positive Rate or TPR) 
%or the probability of rejecting $\gH_0$ when $\gH_1$ is indeed true
 bounded by $e^\epsilon \alpha$.
In practice, one repeats the process of training model $f$ with and without $x$ in the training set, and uses a MIA to guess whether $x$ was included. 
If the MIA has $\textrm{TPR} > e^\epsilon \textrm{FPR}$, the training procedure that outputs $f$ is not $\epsilon$-DP. 
This is the building block of most privacy audits~\cite{jagielski2020auditing, inst, zanellabéguelin2022bayesian, lu2023general, nasr2023tight}. \cref{appendix:rel_work} provides further context, and discusses other related work that is not directly relevant to our contribution.

{\bf Averaging over data instead of models with O(1).} 
The above result bounds the success rate of MIAs when performed over several {\em retrained models}, on two alternative datasets $D$ and $D'$. 
\citet{steinke2023privacy} show that it is possible to average {\em over data} when several data points independently differ between $D$ and $D'$. 
Let $x_{1, m}$ be the $m$ data points independently included in the training set, and $s_{1, m} \in \{0, 1\}^m$ be the vector encoding inclusion. 
$T_{0, m} \in \sR^m$ represents any vector of guesses, with positive values for inclusion in the training set (member), negative values for non-member, and zero for abstaining. 
Then, if the training procedure is $\epsilon$-DP, Proposition 5.1 in~\citet{steinke2023privacy} bounds the performance of guesses from $T$ with:
{\small\begin{align*}
\sP\Big[ \sum_{i=1}^m &\max\{0, T_i \cdot S_i\} \geq v \ | \ T = t \Big]
\leq \underset{{S' \sim \textrm{Bernoulli}(\frac{e^\epsilon}{1+e^\epsilon})^m}}\sP\Big[ \sum_{i=1}^m |t_i| \cdot S'_i \geq v \Big] .
\end{align*}
}%
In other words, an audit (MIA) $T$ that can guess membership better than a Bernoulli random variable with probability $\frac{e^\epsilon}{1+e^\epsilon}$ refutes an $\epsilon$-DP claim.
In this work we build on this result, extending the algorithm (\S\ref{sec:approach}) and theoretical analysis (\S\ref{sec:theory}) to enable the use of generated data for non-members, which breaks the independence between data points.

\section{\acronym} 
\label{sec:approach}

% The goal here is to introduce the end-to-end approach, with both the notations, the mechanics of it (splitting the sets, training a gen, baseline etc..), and the intuition (why that's an audit, why we need a baseline), but without trying to formalize too much (currently there's still a bunch of stuff with in and out, which are not very precise and subsumed by the theory).

%In this section, we describe in detail our approach to  \acronym: \emph{Privacy Auditing with NO Retraining using Artificial data for Membership Inference Attacks}.

% \input{src/notations_table}
Figure~\ref{fig:sysdesign} summarizes the end-to-end \acronym privacy measurement.
The measurement starts with a target model $f$, and a subset of its training data $\DT$ from distribution $\gD$. 
For instance, $\gD$ and $\DT$ could be the distribution and dataset of one participant in an FL training procedure that outputs a final model $f$. 
The privacy measurement then proceeds in two phases.
\begin{figure*}[t]
\centering  
\includegraphics[width=.70\textwidth]{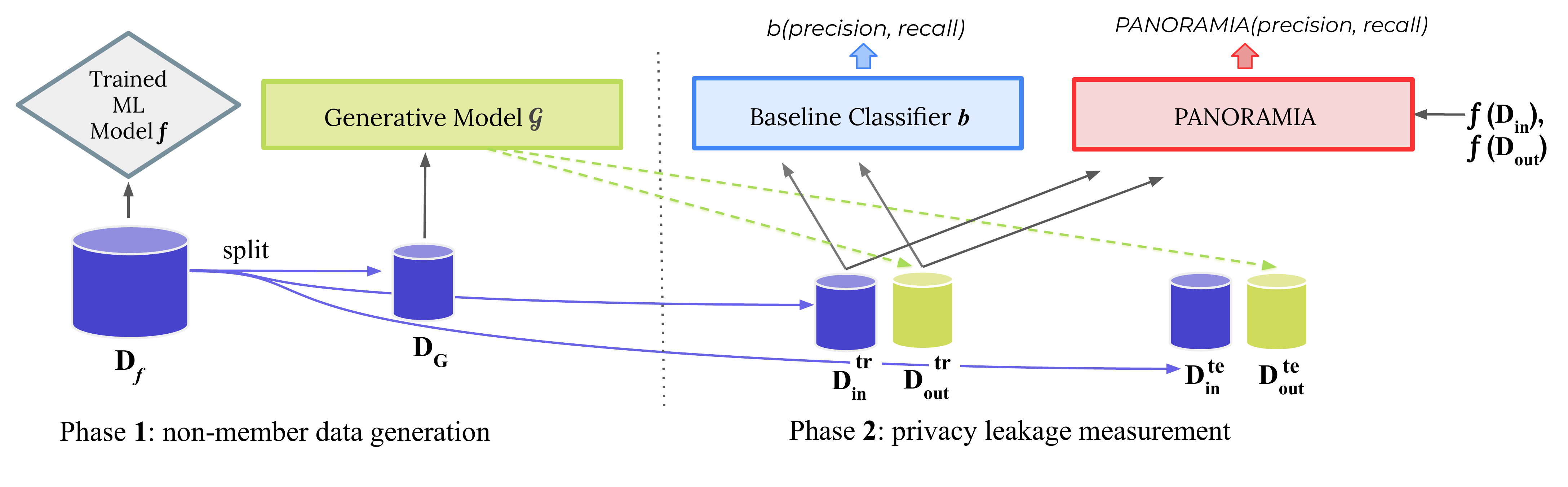}
% \vspace{-.5cm}
  \caption{\acronym's two-phase privacy audit. 
  Phase 1 trains generative model $\gG$ on member data. 
  Phase 2 trains a MIA on a subset of member data and generated non-member data, using the loss of $f$ on these data points. 
  The performance of the MIA is compared to a baseline classifier that does not have access to $f$. Notations are summarized in Table~\ref{table:notations} in \cref{appendix:notations}.}
  \label{fig:sysdesign} 
  % \vspace{-.5cm}
\end{figure*}

{\bf Phase 1:} In the first phase, \acronym uses a subset of the known training data $\DG \subset \DT$ to train a generative model $\gG$. 
The goal of the generative model $\gG$ is to match the training data distribution $\gD$ as closely as possible, which is formalized in Definition~\ref{c-closeness-def} (\S\ref{sec:theory}). 
% At a high level, we require the generative model to put a high likelihood on data that is also likely under $\gD$,
% This means that the generator needs to generate a large number of data points that are good (likely under $\gD$),
% though it is allowed to occasionally generate bad samples (unlikely under $\gD$).
% In \S\ref{sec:theory}, we demonstrate that this one-sided requirement is sufficient due to our focus on finding members (\emph{i.e.}, true positives) as opposed to also finding non-members (\emph{i.e.}, true negatives) in our audit game.
%
Using the generative model $\gG$, we can synthesize non-member data, which corresponds to data that was not used in the training of target model $f$. 
% Hence, we now have access to an independent dataset of member data $\Din = \DT \backslash \DG$, and a synthesized dataset of non-member data $\Dout$ from $\gG$, that we choose to be of the same size $m$ as $\Din$.
Hence, we now have access to an independent dataset of member data $\Din = \DT \backslash \DG$, and a synthesized dataset of non-member data $\Dout \sim \gG$, of size $m = |\Din|$.
%
% \acronym is agnostic to the generative modeling approach, as the only primitive we need is that of data synthesis, with % (and not likelihood evaluation). 
% potential approaches including GANs, VAEs or Diffusion Models. 
%Seb: if possible we could put a reference for each of these models
% We can also fine-tune existing generators and filter generated samples as a post-processing phase. 
% We leverage several such approaches in practice (\S\ref{sec:eval}, \cref{appendix:exp_details}).
%Seb: the above link does not seem to work

{\bf Phase 2:} In the second phase, we leverage $\Din$ and $\Dout$ to audit the privacy leakage of $f$ using a MIA. 
To this end, we split $\Din, \Dout$ into training and testing sets, respectively called $\Dintr, \Douttr$ and $\Dinte, \Doutte$.
We use the training set to train a MIA (called \acronym in Figure~\ref{fig:sysdesign}), a binary classifier that predicts whether a given datapoint is a member of $\DT$, the training set of the target model $f$.
This MIA classifier makes its prediction based on both a training example $x$, as well as information from applying the target model $f$ to the input, such as the loss of the target model when applied to this example $\textrm{loss}(f(x))$
% , or the likelihood of the data $f(x)$ for language models
(see \S\ref{sec:eval}, \cref{appendix:exp_details} for details).
% As discussed in \S\ref{sec:background}, MIA performance is directly related to a lower-bound on possible $\epsilon$-DP values for the target model $f$.
We use the test set to measure the MIA performance, using the precision at different recall values.

Previous results linking the performance of a MIA on several data-points to $\epsilon$-DP bounds rely on independence between members and non-members. 
This intuitively means that there is no information about membership in $x$ itself.
When the auditor controls the training process this independence is enforced by construction, by adding data points to the training set based on an independent coin flip. 
%As a result of independence, MIA performance can be compared to that of a random guess to assess privacy leakage.
In \acronym, we do not have independence between membership and $x$, as all non-members come from the generator $\gG \neq \gD$.
% Intuitively, this implies that there are two ways to guess membership and have high MIA precision: either using $f$ to detect membership (symptomatic of privacy leakage) or detecting generated data (not a symptom of privacy leakage).
% Since we are only interested in measuring privacy leakage, we cannot compare our MIA results to that of a random guess but have to account for the detection of generated data based on $x$ only. 
% To do this, we compare the results of the MIA to that of a baseline classifier $b$ that guesses membership based exclusively on $x$, without access to $f$ or its predictions on $x$. 
% The stronger this baseline, the better the removal of the effect of synthesized data detection. 
% \S\ref{subsec:baseline} details how to construct a strong baseline while
% In practice, we leverage features from pre-trained models to strengthen the baseline (not on Figure~\ref{fig:sysdesign} for simplicity, see \S\ref{sect_exp} for details).
As a result, there are two ways to guess membership and have high MIA precision: either by using $f$ to detect membership (\emph{i.e.}, symptomatic of privacy leakage) or by detecting generated data (\emph{i.e.}, not a symptom of privacy leakage).
To measure the privacy leakage, we compare the results of the MIA to that of a baseline classifier $b$ that guesses membership based exclusively on $x$, without access to $f$. % or its predictions on $x$. 
The stronger this baseline, the better the removal of the effect of synthesized data detection.
Algorithm \ref{algo:panoramia} summarizes the entire procedure.
In the next section, we demonstrate how to relate the difference between the baseline $b$ and the MIA performance to the privacy loss $\epsilon$.

\begin{figure}
\begin{algorithm}[H]
\caption{\acronym}\label{algo:panoramia} 
\begin{algorithmic}[1]

    \Statex{\hspace{-.4cm}{\bf Input:} Target ML model $\model$, audit set size $m$, confidence $1-\beta$}
    % \Ensure Privacy Leakage Assessment
    \Statex{\hspace{-.4cm}\bf Phase 1:}
    \State{Split $\DT$ in $\DG, \Dintr, \Dinte$, with $|\Dinte|=m$}
    \State{Train generator $\gG$ on $\DG$}
    \State{Generate $\Douttr, \Doutte$ of size $|\Dintr|, |\Dinte|$}   
    % \State{\quad $D_{\text{in}}$: Set of in data points used to train $f$.}
    % \State \quad $D_{\text{out}}$: Set of out data points representing synthetic data generated by $\mathcal{G}$.
    % \State \quad $m$: Number of in data points, (\emph{i.e.}, $m = |D_{\text{in}}|$).
    % \State \quad For each data point $x_i$, $1 \leq i \leq n$, the label $y_i$ is such that $y_i = 1$ if $x_i \in D_{\text{in}}$, and $y_i = 0$ if $x_i \in D_{\text{out}}$.
    \Statex{\hspace{-.4cm}\bf Phase 2:}   
    \Statex{\hspace{-.4cm}Train the baseline and MIA:}
    \setcounter{ALG@line}{0}
    \State Label $\Dintr$ as members, and $\Douttr$ as non-members
    \State Train $b$ to predict labels using $x \in \Dintr \cup \Douttr$
    \State Train MIA to predict labels using $x \in \Dintr \cup \Douttr$ and $f(x)$

    \Statex{\hspace{-.4cm}Measure privacy leakage (see \S\ref{sec:theory}):}
    \setcounter{ALG@line}{0}
    \State Sample $s \sim \textrm{Bernoulli}(\frac{1}{2})^m$ \Comment{Def.\ref{def:auditing-game}}
    \State Create audit set $X = s \cdot \Dinte + (1-s)
    \Doutte$
    \State Score each audit point for membership, creating $t^b \triangleq b(X) \in \sR_+^m$ and $t^a \triangleq \textrm{MIA}(X) \in \sR_+^m$
    \State Set $v_{\textrm{ub}}^b(c, t) \triangleq \sup \ \{v : \beta^b(m, c, v, t) \leq \frac{\beta}{2}\}$ \Comment{Prop.\ref{prop:gen-test}} \label{algo:max-t-start}
    % \State Compute the precision of $b$ as $\textrm{pec}^b \triangleq t^b \cdot s / \sum t^b$
    \State $\clb = \max_{t,c} \1\{t^b \geq t\} \cdot s \leq v_{\textrm{ub}}^b(c, \1\{t^b \geq t\})$
    \State Set $v_{\textrm{ub}}^a(c, \epsilon, t) \triangleq \sup \ \{v : \beta^a(m, c, \epsilon, v, t) \leq \frac{\beta}{2}\}$ \Comment{Prop.\ref{prop:dp-test}}
    % \State Compute the MIA precision as $\textrm{pec}^b \triangleq t^a \cdot s / \sum t^a$
    \State $\cpepslb = \max_{t, c, \epsilon} \1\{t^a \geq t\} \cdot s \leq v_{\textrm{ub}}^a(c, \epsilon, \1\{t^a \geq t\})$ \label{algo:max-t-end}

    \Statex{\hspace{-.4cm}{\bf Return} $\tilde\epsilon \triangleq \cpepslb - \clb$}
\end{algorithmic}
\end{algorithm}
\vspace{-20pt}
\end{figure}

\section{Quantifying Privacy Leakage with \acronym} 
\label{sec:theory}
% The first step to quantifying privacy leakage is to formalize an auditing game. 
% \acronym starts with $\xin \in \gX^m$ training points from $\gD$ (\emph{e.g.}, the data distribution of one participant in an FL setting), as well as $\xgen \in \gX^m$ generated points from the generator distribution $\gG$ ($\xgen \sim \gG$).
% % Say if we use discrete or density or anything... Ideally we'd do a measure theory version, but it seems painful, maybe let's just sweep all that under the rug for now, I think it's intuitive enough...
% We construct a sequence auditing samples $x \in \gX^m$ as follows:
We formalize a privacy game as follows: construct a sequence of auditing samples, in which $s_i$ is sampled independently to choose either the real ($\xin \in \gX^m$ training points from $\gD$ if $s_i=1$) or generated ($\xgen \in \gX^m$ generated points from $\gG$ if $s_1=0$) data point at each index $i$.

\begin{definition}[Privacy game]\label{def:auditing-game}
    \begin{align*}
    & s \sim \textrm{Bernoulli}(\frac{1}{2})^m , \ \textrm{with} \ s_i \in \{0, 1\} , x_i = (1-s_i) \xgen_i + s_i \xin_i, \ \forall i \in \{1, \ldots, m\} .
    \end{align*}
\end{definition}

The key difference between \cref{def:auditing-game} and the privacy game of \citet{steinke2023privacy} lies in how we create the audit set. Instead of sampling non-members independently and removing them from the training set (which requires modifyig the training set and the model), \cref{def:auditing-game} starts from a set of known members (after the fact), and pairs each point with a non-member (generated i.i.d. from the generator distribution). For each pair, we then flip a coin to decide which point in the pair will be shown to the auditor for testing, thereby creating the test task of our privacy measurement.

% This means that $s$ is sampled independently to choose either the real ($s_i=1$) or generated ($s_1=0$) data point at each index $i$. 
% This creates a sequence of $m$ examples that \acronym will try to tell apart (\emph{i.e.}, guess $s$). 
The level of achievable success in this game quantifies the privacy leakage of the target model $f$.
More precisely, we follow an analysis inspired by that of \cite{steinke2023privacy}, but require several key technical changes to support auditing with no retraining using generated non-member data.
% The first major change is the introduction of a new game based on generated data, which requires accounting for the quality of the generator in the hypothesis test and the analysis, and interpreting results accordingly.
% The second adaptation is to focus on member detection, ignoring non-members. 
% This lets us soften the requirements put on the data generator, which only needs to assign high likelihood to real data, but can (and does) occasionally generate poor samples that are easy to detect.
% To realize this, we first formalize \acronym's audit procedure as a hypothesis test on our auditing game, for which we construct a statistical test (\S\ref{subsec:hyp-test}). 
% Then, we show how to use this hypothesis test to quantify privacy leakage as a lower confidence interval, and interpret the semantics of our privacy leakage measurement (\S\ref{subsec:dp-measurement}).

\subsection{Formalizing the Privacy Measurements as a Hypothesis Test}
\label{subsec:hyp-test}

We first need a notion of quality for our generator:
\begin{definition}[$c$-closeness]\label{c-closeness-def} For all $c > 0$, we say that a generative model $\gG$ is $c$-close for data distribution $\gD$ if: $$ \forall x \in \gX, \ e^{-c} \sP_\gD\big[x\big] \leq \sP_\gG\big[x\big] . $$
\end{definition}
The smaller $c$, the better the generator, as $\gG$ assigns a probability to real data that cannot be much smaller than that of the real data distribution. We make three important remarks.\\
{\bf Remark 1:} $c$-closeness is a strong requirement but it is achievable, at least for a large $c$. 
For instance, a uniform random generator is $c$-close to any distribution $\gD$,  with $c < \infty$. 
Of course $c$ would be very large, and the challenge in \acronym is to create a generator that empirically has a small enough $c$.\\
{\bf Remark 2:} We show in \cref{appendix:delta_relaxation} that we can relax this definition to hold only with high probability (following the $(\epsilon, \delta)$-DP relaxation). 
Empirical results under this relaxation are very similar.
We also show how to measure $(\epsilon, \delta)$-DP, instead of $\epsilon$-DP on which we focus here. 
% \qt{Make sure there is a direct side-by-side comparison of how the relaxation impacts results, or remove the second sentence (but ideally keep it :)).}
%\qt{Added eval in E.3, I think it roughly follows the same trend.}
\\
{\bf Remark 3:} Our definition of closeness is very similar to that of DP.
This is not a coincidence as we will use this definition to be able to reject claims of both $c$-closeness for the generator and $\epsilon$-DP for the algorithm that gave the target model.
We further note that, contrary to the DP definition, $c$-closeness is one-sided as we only bound $\sP_\gG$ from below. 
Intuitively, this means that the generator has to produce high-quality samples (\emph{i.e.}, samples likely under $\gD$) with high enough probability but it does not require that all samples are good though.
%Seb: I have put in comments to save on space
%It does not require that all samples are good though, and the generator can sometimes generate samples that are unlikely under $\gD$.
This one-sided measure of closeness is enabled by our focus on detecting members (\emph{i.e.}, true positives) as opposed to members and non-members, and puts less stringent requirements on the generator.

Using Definition~\ref{c-closeness-def}, we can formulate the hypothesis test on the privacy game of Definition~\ref{def:auditing-game} that underpins our approach:
\begin{equation*}\label{eq:hypothesis}
\gH: \textrm{generator} \ \gG \ \textrm{is} \ c\textrm{-close, and target model} \ f \ \textrm{is} \ \epsilon\textrm{-DP}.
\end{equation*}
Note that, in a small abuse of language,
% \mishaal{very minor thing: would "terminology" sound better than "language"?} \ml{I don't think so}
we will often refer to $f$ as $\epsilon$-DP to say that $f$ is the output of an $\epsilon$-DP training mechanism.
% To construct a statistical test allowing us to reject $\gH$ based on evidence, we define two key mechanisms (corresponding to \acronym's privacy measurement scheme). 
% First, the (potentially randomized) baseline guessing mechanism $B(s, x): \{0, 1\}^m \times \gX^m \rightarrow \sR_+^m$, which outputs a (non-negative) score for the membership of each datapoint $x_i$, based on this datapoint only. 
That is,
%\ml{[I think I need $B(s, x) = \{b_1(x_1), b_2(x_2), \ldots, b_m(x_m)\}$ to make it work, which should be fine, but may require explicitly adding side info for the baseline (and MIA?) training set(s). Maybe we can relax this later...]  \sout{$B(s, x) = \{b_1(x_1), b_2(x_{\leq 2}), \ldots, b_m(x_{\leq m})\}$  predicts guess $i$ based on $x_{\leq i}$, the sequence of $\{x_1, \ldots, x_i\}$} $B(s, x) = \{b(x_1), b(x_2), \ldots, b(x_m)\}$.}
$B(s, x) = \{b(x_1), b(x_2), \ldots, b(x_m)\}$.

% Second, we define $A(s, x, f): \{0, 1\}^m \times \gX^m \times \gF \rightarrow \sR_+^m$, which outputs a (potentially randomized) non-negative score for the membership of each datapoint, with the guess for index $i$ depending on $x_{\leq i}$ and target model $f$. 
% Note that if the target model $f$ is DP, then $A$ is DP w.r.t. inclusion in the auditing game $s$, outside of what is revealed by $x$.
% We are now ready to construct a hypothesis test for $\gH$. 
% First, we construct tests for each part of the hypothesis separately.
We define two key mechanisms: $B(s, x): \{0, 1\}^m \times \gX^m \rightarrow \sR_+^m$ outputs a (potentially randomized) non-negative score for the membership of each datapoint $x_i$, based on $x_i$ only; $A(s, x, f): \{0, 1\}^m \times \gX^m \times \gF \rightarrow \sR_+^m$ outputs a (potentially randomized) non-negative score for the membership of each datapoint, with the guess for index $i$ depending on $x_{\leq i}$ and target model $f$. 
We can construct a statistical test for each part of the hypothesis separately.

%Seb: je suggère de mettre un titre pour cette proposition
\begin{proposition}\label{prop:gen-test}
Let $\gG$ be $c$-close, $S$ and $X$ be the random variables for $s$ and $x$ from \cref{def:auditing-game}, and $T^b \triangleq B(S, X)$ be the vector of guesses from the baseline. Then, for all $v \in \sR$ and all $t$ in the support of $T$:
\begin{align*}
& \sP_{S, X, T^b}\Big[ \sum_{i=1}^m T^b_i \cdot S_i \geq v \ | \ T^b = t^b \Big] \leq \underset{{S' \sim \textrm{Bernoulli}(\frac{e^c}{1+e^c})^m}}\sP\Big[ \sum_{i=1}^m t^b_i \cdot S'_i \geq v \Big] \triangleq \beta^b(m, c, v, t^b)
\end{align*}
\end{proposition}
\begin{showIfProofs}
\begin{proof}
Notice that under our baseline model $B(s, x) = \{b(x_1), b(x_2), \ldots, b(x_m)\}$, and given that the $X_i$ are i.i.d., we have that: $S_{<i} \ind T^b_{<i} \ | \ X_{<i}$, since $T^b_i = B(S, X)_i$'s distribution is entirely determined by $X_i$;
and $S_{\leq i} \ind T^b_{> i} \ | \ X_{<i}$ since the $X_{i}$ are sampled independently from the past.% Our generative process for $X_i$ also implies that $X_i \ind S_{<i}, X_{<i} \ | \ S_i$.

We study the distribution of $S$ given a fixed prediction vector $t^b$, one element $i \in [m]$ at a time:
\begin{align*}
& \sP\big[ S_i = 1 \ | \ T^b = t^b, S_{<i} = s_{<i}, X_{\leq i} = x_{\leq i}\big] \\
&= \sP\big[ S_i = 1 \ | \ S_{<i} = s_{<i}, X_{\leq i} = x_{\leq i}\big] \\
&= \sP\big[ X_i \ | \ S_i = 1, S_{<i} = s_{<i}, X_{< i} = x_{< i}\big] \\
& \ \ \ \ \frac{\sP\big[ S_i = 1 \ | \ S_{<i} = s_{<i}, X_{< i} = x_{< i} \big]}{\sP\big[ X_i \ | \ S_{<i} = s_{<i}, X_{< i} = x_{< i} \big]} \\
&= \frac{\sP\big[ X_i \ | \ S_i = 1, S_{<i} = s_{<i}, X_{< i} = x_{< i}\big]\sP\big[ S_i = 1 \big]}{\sP\big[ X_i \ | \ S_{<i} = s_{<i}, X_{< i} = x_{< i} \big]} \\
&= \frac{\sP\big[ X_i \ | \ S_i = 1 \big]\frac{1}{2}}{\sP\big[ X_i \ | \ S_i=1 \big]\frac{1}{2} + \sP\big[ X_i \ | \ S_i=0 \big]\frac{1}{2}} \\
&= \frac{1}{1 + \frac{\sP\big[ X_i \ | \ S_i=0 \big]}{\sP\big[ X_i \ | \ S_i=1 \big]}} = \frac{1}{1 + \frac{\sP_\gG\big[ X_i \big]}{\sP_\gD\big[ X_i \big]}} \leq \frac{1}{1 + e^{-c}} = \frac{e^c}{1 + e^{c}}
\end{align*}
The first equality uses the independence remarks at the beginning of the proof, the second relies on Bayes' rule, while the third and fourth that $S_i$ is sampled i.i.d. from a Bernoulli with probability half, and $X_i$ i.i.d. conditioned on $S_i$. 
The last inequality uses Definition~\ref{c-closeness-def} for $c$-closeness.

Using this result and the law of total probability to introduce conditioning on $X_{\leq i}$, we get that:
\begin{align*}
& \sP\big[ S_i = 1 \ | \ T^b = t^b, S_{<i} = s_{<i}\big] \\
& = \sum_{x_{\leq i}} \sP\big[ S_i = 1 \ | \ T^b = t^b, S_{<i} = s_{<i},  X_{\leq i} = x_{\leq i} \big] \\
& \ \ \ \ \sP\big[ X_{\leq i} = x_{\leq i} \ | \ T^b = t^b, S_{<i} = s_{<i} \big] \\
& \leq \sum_{x_{\leq i}} \frac{e^c}{1 + e^{c}} \sP\big[ X_{\leq i} = x_{\leq i} \ | \ T^b = t^b, S_{<i} = s_{<i} \big] ,
\end{align*}
and hence that:
\begin{equation}\label{eq:stoc-dom-gen}
\sP\big[ S_i = 1 \ | \ T^b = t^b, S_{<i} = s_{<i}\big] \leq \frac{e^c}{1 + e^{c}}
\end{equation}

We can now proceed by induction: assume inductively that $W_{m-1} \triangleq \sum_{i=1}^{m-1} T^b_i \cdot S_i$ is stochastically dominated (see Definition 4.8 in \cite{steinke2023privacy}) by $W'_{m-1} \triangleq \sum_{i=1}^{m-1} T^b_i \cdot S'_i$, in which $S' \sim \textrm{Bernoulli}(\frac{e^c}{1 + e^{c}})^{m-1}$. 
Setting $W_1 = W'_1 = 0$ makes it true for $m=1$. Then, conditioned on $W_{m-1}$ and using Eq. \ref{eq:stoc-dom-gen}, $T^b_m \cdot S_m = T_m \cdot \1\{S_m = 1\}$ is stochastically dominated by $T^b_m \cdot \textrm{Bernoulli}(\frac{e^c}{1 + e^{c}})$.
Applying Lemma 4.9 from \cite{steinke2023privacy} shows that $W_{m}$ is stochastically dominated by $W'_{m}$, which proves the induction and implies the proposition's statement.
\end{proof}
\end{showIfProofs}
\begin{hideIfProofs}
\begin{proof}
In Appendix \ref{appendix:proof-prop1}.
\end{proof}
\end{hideIfProofs}

Now that we have a test to reject a claim that the generator $\gG$ is $c$-close for the data distribution $\gD$, we turn our attention to the second part of $\gH$ which claims that the target model $f$ is the result of an $\epsilon$-DP mechanism.

\begin{proposition}\label{prop:dp-test}
Let $\gG$ be $c$-close, $S$ and $X$ be the random variables for $s$ and $x$ from \cref{def:auditing-game}, $f$ be $\epsilon$-DP, and $T^a \triangleq A(S, X, f)$ be the guesses from the membership audit. Then, for all $v \in \sR$ and all $t$ in the support of $T$:
\begin{align*}
& \sP_{S, X, T^a}\Big[ \sum_{i=1}^m T^a_i \cdot S_i \geq v \ | \ T^a = t^a \Big] \leq \underset{{S' \sim \textrm{Bernoulli}(\frac{e^{c+\epsilon}}{1+e^{c + \epsilon}})^m}}\sP\Big[ \sum_{i=1}^m t^a_i \cdot S'_i \geq v \Big] \triangleq \beta^a(m, c, \epsilon, v, t^a)
\end{align*}
\end{proposition}
\begin{showIfProofs}
\begin{proof}
Fix some $t^a \in \sR_+^m$. 
We study the distribution of $S$ one element $i \in [m]$ at a time:
\begin{align*}
& \sP\big[ S_i = 1 \ | \ T^a = t^a, S_{<i} = s_{<i}, X_{\leq i} = x_{\leq i}\big] \\
&= \sP\big[ T^a = t^a \ | \ S_i = 1, S_{<i} = s_{<i}, X_{\leq i} = x_{\leq i}\big] \\
& \ \ \ \ \frac{\sP\big[ S_i = 1 \ | \ S_{<i} = s_{<i}, X_{\leq i} = x_{\leq i} \big]}{\sP\big[ T^a = t^a \ | \ S_{<i} = s_{<i}, X_{\leq i} = x_{\leq i} \big]} \\
&\leq \frac{1}{1 + e^{-\epsilon}\frac{\sP\big[ S_i = 0 \ | \ S_{<i} = s_{<i}, X_{\leq i} = x_{\leq i} \big]}{\sP\big[ S_i = 1 \ | \ S_{<i} = s_{<i}, X_{\leq i} = x_{\leq i} \big]}} \\
&\leq \frac{1}{1 + e^{-\epsilon}e^{-c}} = \frac{e^{c+\epsilon}}{1 + e^{c+\epsilon}}
\end{align*}
The first equality uses Bayes' rule. 
The first inequality uses the decomposition:
\begin{align*}
& \sP\big[ T^a = t^a \ | \ S_{<i} = s_{<i}, X_{\leq i} = x_{\leq i} \big] = \\
&= \sP\big[ T^a = t^a \ | \ S_i = 1, S_{<i} = s_{<i}, X_{\leq i} = x_{\leq i} \big] \\
& \ \ \cdot \sP\big[ S_i = 1 \ | \ S_{<i} = s_{<i}, X_{\leq i} = x_{\leq i} \big] \\
&+ \sP\big[ T^a = t^a \ | \ S_i = 0, S_{<i} = s_{<i}, X_{\leq i} = x_{\leq i} \big] \\
& \ \ \cdot \sP\big[ S_i = 0 \ | \ S_{<i} = s_{<i}, X_{\leq i} = x_{\leq i} \big] ,
\end{align*}
and the fact that $A(s, x, f)$ is $\epsilon$-DP w.r.t. $s$ and hence that:
\[
\frac{\sP\big[ T^a = t^a \ | \ S_i = 0, S_{<i} = s_{<i}, X_{\leq i} = x_{\leq i}\big]}{\sP\big[ T^a = t^a \ | \ S_i = 1, S_{<i} = s_{<i}, X_{\leq i} = x_{\leq i}\big]} \geq e^{-\epsilon} .
\]

The second inequality uses that:
\begin{align*}
& \frac{\sP\big[ S_i = 0 \ | \ S_{<i} = s_{<i}, X_{\leq i} = x_{\leq i} \big]}{\sP\big[ S_i = 1 \ | \ S_{<i} = s_{<i}, X_{\leq i} = x_{\leq i} \big]} \\
&= \frac{\sP\big[ X_i = x_i \ | \ S_i = 0, S_{<i} = s_{<i}, X_{< i} = x_{< i} \big]}{\sP\big[ X_i = x_i \ | \ S_i = 1, S_{<i} = s_{<i}, X_{< i} = x_{< i} \big]} \\
& \ \ \cdot \frac{\sP\big[ S_i = 0 \ | \  S_{<i} = s_{<i}, X_{< i} = x_{< i} \big]}{\sP\big[ S_i = 1 \ | \ S_{<i} = s_{<i}, X_{< i} = x_{< i} \big]} \\
&= \frac{\sP\big[ X_i = x_i \ | \ S_i = 0, S_{<i} = s_{<i}, X_{< i} = x_{< i} \big]}{\sP\big[ X_i = x_i \ | \ S_i = 1, S_{<i} = s_{<i}, X_{< i} = x_{< i} \big]} \cdot \frac{1/2}{1/2} \\
& = \frac{\sP_\gG\big[ X_i \big]}{\sP_\gD\big[ X_i \big]} \geq e^{-c}
\end{align*}

As in Proposition~\ref{prop:gen-test}, applying the law of total probability to introduce conditioning on $X_{\leq i}$ yields:
\begin{equation}\label{eq:stoc-dom-mia}
\sP\big[ S_i = 1 \ | \ T^a = t^a, S_{<i} = s_{<i} \big] \leq \frac{e^{c+\epsilon}}{1 + e^{c+\epsilon}} ,
\end{equation}
and we can proceed by induction.
Assume inductively that $W_{m-1} \triangleq \sum_{i=1}^{m-1} T^a_i \cdot S_i$ is stochastically dominated (see Definition 4.8 in \cite{steinke2023privacy}) by $W'_{m-1} \triangleq \sum_{i=1}^{m-1} T^a_i \cdot S'_i$, in which $S' \sim \textrm{Bernoulli}(\frac{e^{c+\epsilon}}{1 + e^{c+\epsilon}})^{m-1}$. 
Setting $W_1 = W'_1 = 0$ makes it true for $m=1$. 
Then, conditioned on $W_{m-1}$ and using Equation~\ref{eq:stoc-dom-mia}, $T^a_m \cdot S_m = T^a_m \cdot \1\{S_m = 1\}$ is stochastically dominated by $T^a_m \cdot \textrm{Bernoulli}(\frac{e^{c+\epsilon}}{1 + e^{c+\epsilon}})$. 
Applying Lemma 4.9 from \cite{steinke2023privacy} shows that $W_{m}$ is stochastically dominated by $W'_{m}$, which proves the induction and implies the proposition's statement.
\end{proof}
\end{showIfProofs}
\begin{hideIfProofs}
\begin{proof}
In Appendix \ref{appendix:proof-prop2}.
\end{proof}
\end{hideIfProofs}

We can now provide a test for hypothesis $\gH$, by applying a union bound over \cref{prop:gen-test,prop:dp-test}:

%Seb: suggestion : put a title for the proposition
\begin{corollary}\label{corollary:hyp-test}
Let $\gH$ be true, $T^b \triangleq B(S, X)$, and $T^a \triangleq A(S, X, f)$.
Then:
\begin{align*}
& \sP\Big[ \sum_{i=1}^m T^a_i \cdot S_i \geq v^a, \ \sum_{i=1}^m T^b_i \cdot S_i \geq v^b\ | \ T^a = t^a, T^b = t^b \Big] \leq \beta^a(m, c, \epsilon, v^a, t^a) + \beta^b(m, c, v^b, t^b)
\end{align*}
\end{corollary}
To make things more concrete, let us instantiate Corollary~\ref{corollary:hyp-test} as we do in \acronym. 
Our baseline ($B$ above) and MIA ($A$ above) classifiers return a membership guesses $T^{a,b} \in \{0, 1\}^m$, with $1$ corresponding to membership. 
Let us call $r^{a,b} \triangleq \sum_i t_i^{a,b}$ the total number of predictions, and $\tp^{a,b} \triangleq \sum_i t_i^{a,b} \cdot s_i$ the number of correct membership guesses (true positives).
We also call the precision $\pr^{a,b} \triangleq \frac{\tp^{a,b}}{r^{a,b}}$.
Using the following tail bound on the sum of Bernoulli random variables for simplicity and clarity (we use a tighter bound in practice, but this one is easier to read),
\begin{equation}\label{eq:bound-example}
    \underset{{S' \sim \textrm{Bernoulli}(p)^r}}\sP\Big[ \sum_{i=1}^r \frac{S'_i}{r} \geq p + \sqrt{\frac{\log(1/\beta)}{2r}} \Big] \leq \beta ,
\end{equation}
we can reject $\gH$ at confidence level $\beta$ by setting $\beta^a = \beta^b = \frac{\beta}{2}$ and if either $\pr^b \geq \frac{e^c}{1+e^c} + \sqrt{\frac{\log(2/\beta)}{2r^b}}$ or $\pr^a \geq \frac{e^{c+\epsilon}}{1+e^{c+\epsilon}} + \sqrt{\frac{\log(2/\beta)}{2r^a}}$.

\subsection{Quantifying Privacy Leakage and Interpretation}
\label{subsec:dp-measurement}

In an privacy measurement, we would like to quantify $\epsilon$, not just reject a given $\epsilon$ claim. 
We use the hypothesis test from Corollary \ref{corollary:hyp-test} to compute a confidence interval on $c$ and $\epsilon$. 
To do this, we first define an ordering between $(c, \epsilon)$ pairs, such that if  $(c_1, \epsilon_1) \leq (c_2, \epsilon_2)$, the event (\emph{i.e.}, set of observations for $T^{a,b}, S$) for which we can reject $\gH(c_2, \epsilon_2)$ is included in the event for which we can reject $\gH(c_1, \epsilon_1)$. 
That is, if we reject $\gH$ for values $(c_2, \epsilon_2)$ based on audit observations, we also reject values $(c_1, \epsilon_1)$ based on the same observations.

We define the following order to fit this assumption, based on the hypothesis test from Corollary \ref{corollary:hyp-test}:
\begin{equation}\label{eq:ordering}
(c_1, \epsilon_1) \leq (c_2, \epsilon_2) \ \textrm{if either} \ 
\begin{cases}
      c_1 < c_2, \ \textrm{or} \\
      c_1 = c_2 \ \textrm{and} \ \epsilon_1 \leq \epsilon_2
\end{cases}
\end{equation}
% \qt{The lexicographic order in Equation~\ref{eq:ordering} is not to ensure $\tilde{\epsilon} > c$ but is needed to construct the confidence interval from Corollary~\ref{corollary:hyp-test}.}
This lexicographic order over hypotheses ensures that, when using the joint hypothesis test from Corollary~\ref{corollary:hyp-test} to construct confidence intervals, we always reject hypotheses with a low value of $c$ incompatible with the baseline performance. Formally, this yields the following confidence intervals:
\begin{corollary}
\label{result:confidence-interval}
For all $\beta \in ]0, 1]$, $m$, and observed $t^a, t^b$, call $v_{\textrm{ub}}^a(c, \epsilon) \triangleq \sup \ \{v : \beta^a(m, c, \epsilon, v, t^a) \leq \frac{\beta}{2}\}$ and $v_{\textrm{ub}}^b(c) \triangleq \sup \ \{v : \beta^b(m, c, v, t^b) \leq \frac{\beta}{2}\}$. Then:
{\footnotesize
\begin{align*}
& \sP\Big[ (c, \epsilon) \geq \sup \big\{(c', \epsilon') : t^a \cdot s \leq v_{\textrm{ub}}^a(c', \epsilon') \ \textrm{and} \ t^b \cdot s \leq v_{\textrm{ub}}^b(c') \big\} \Big] \geq 1 - \beta
\end{align*}
}
\end{corollary}
\begin{proof}
%\ml{maybe develop and make this a proposition...}
Apply Lemma 4.7 from \cite{steinke2023privacy} with the ordering from Equation~\ref{eq:ordering} and the test from Corollary \ref{corollary:hyp-test}.
\end{proof}
The lower confidence interval for $(c, \epsilon)$ at confidence $1-\beta$ is the largest $(c, \epsilon)$ pair that cannot be rejected using Corollary~\ref{corollary:hyp-test} with false rejection probability at most $\beta$.
Hence, for a given confidence level $1-\beta$, \acronym computes $(c_{\textrm{lb}}, \tilde\epsilon)$, the largest value for $(c, \epsilon)$ that it cannot reject. 
%$(c_{\textrm{lb}}, \tilde\epsilon)$ lower-bounds the true value for $(c, \epsilon)$ with probability at least $1-\beta$.
Note that Corollaries~\ref{corollary:hyp-test} and \ref{result:confidence-interval} rely on a union bound between two tests, one for $c$ and one for $c+\epsilon$.
We can thus consider each test separately (at level $\beta/2$).
In practice,
% we follow previous approaches~\cite{maddock2023canife,steinke2023privacy} and
we use the whole precision/recall curve to achieve tighter bounds:
each level of recall (\emph{i.e.}, threshold on $t^{a,b}$ to predict membership) corresponds to a bound on the precision, which we can compare to the empirical value. 
Separately for each test, we pick the level of recall yielding the highest lower-bound (Lines~\ref{algo:max-t-start}-\ref{algo:max-t-end} in the last section of Algorithm \ref{algo:panoramia}).
For all bounds to hold together, we apply a union bound over all tests at each level of recall. Specially, if we have $m$ test datapoints, we have at most $m$ different values of recall, and the union bound corresponds to $\beta \leftarrow \beta/m$ (notice in \cref{eq:bound-example} that the bounds depend on $\sqrt{\log(1/\beta)}$, so this is not too costly).
We next discuss the semantics of returned measurement values $(c_{\textrm{lb}}, \tilde\epsilon)$.

{\bf Measurement semantics.} \cref{result:confidence-interval} gives us a lower-bound for $(c, \epsilon)$, based on the ordering from \cref{eq:ordering}. 
To understand the value $\tilde\epsilon$ returned by \acronym, we need to understand what the hypothesis test rejects. 
Rejecting $\gH$ means either rejecting the claim about $c$, or the claim about $c + \epsilon$ (which is the reason for the ordering in Equation~\ref{eq:ordering}). 
With Corollary~\ref{result:confidence-interval}, we hence get both a lower-bound $\clb$ on $c$, and $\cpepslb$ on $c + \epsilon$. 
Unfortunately, $\tilde\epsilon \triangleq \max\{0, \cpepslb - \clb \}$, which is the value \acronym returns, does not imply a lower-bound on $\epsilon$. 
Instead, we can claim that ``\acronym could not reject a claim of $c$-closeness for $\gG$, and if this claim is tight, then $f$ cannot be the output of a mechanism that is $\tilde\epsilon$-DP''.
%Seb: I was wondering if we could provide one or two illustrative examples with particular values for c or epsilon

%\ml{Should we use a variation on $\epsilon$ for \acronym's measurement, to show that it's not a true lower bound? $\tilde\epsilon$?}
While this is not as strong a claim as typical lower-bounds on $\epsilon$-DP from prior privacy auditing works, we believe that this measure is useful in practice. 
Indeed, the $\tilde\epsilon$ measured by \acronym~is a quantitative privacy measurement, that is accurate (close to a lower-bound on $\epsilon$-DP) when the baseline performs well (and hence $c_{\textrm{lb}}$ is tight). 
%TW: Baseline performs well -> does this mean the generator performs well, or does it mean the baseline classifier performs well (thus the generator performs poorly)? It seems to me the former is intuitive but the text implies the latter. 
Thus, when the baseline is good, $\tilde\epsilon$ can be interpreted as (close to) a lower bound on (pure) DP. 
% In addition, since models on the same dataset and task share the same baseline, which does not depend on the audited model, \acronym's measurement can be used to directly compare privacy leakage between models. 
\acronym~opens a new capability, measuring privacy leakage of a trained model $f$ without access or control of the training pipeline, with a meaningful and practically useful measurement, as we empirically show next.

\section{Evaluation}
\label{sec:eval}
    \begin{table*} [t]
    \centering
    \small
     \resizebox{1\textwidth}{!}{
    \begin{tabular}{|c|c|c|c|c|}
        \hline
        ML Model & Dataset & Training Epoch  & Test Accuracy &  Model Variants Names* \\
        \hline
        ResNet101 & CIFAR10 & 20, 50, 100  & 91.61\%, 90.18\%, 87.93\%& ResNet101\_E20, ResNet101\_E50, ResNet101\_E100,\\
        WideResNet-28-10 & CIFAR10 & 150, 300  & 95.67\%, 94.23\% & WRN-28-10\_E150, WRN-28-10\_E300\\
        ViT-small (pretrained) & CIFAR10 & 35  & 96.38\% & ViT\_E35\\
        % Wide ResNet-16-4 & CIFAR10 & 20, 50, 70  & \%, \%, \% & W-RN\_E20, W-RN\_E70, W-RN\_E100 \\
        Multi-Label CNN & CelebA & 50, 100  & 81.77\%, 78.12\% & CNN\_E50, CNN\_E100 \\
        GPT-2 & WikiText & 
        % 12, 23, 46, 92
        37, 75, 150
        & Refer to Appendix \ref{subsec:llmfullexpdetails} & 
        % GPT-2\_E12, GPT-2\_E23, GPT-2\_E46, GPT-2\_E92 
        GPT-2\_E37, GPT-2\_E75, GPT-2\_E150 
        \\
        % MLP Tabular Classification  & Adult & 10, 100 & 86\%, 82\%  & MLP\_E10, MLP\_E100\\
        \hline
    \end{tabular}}
    \caption{Train and Test Metrics for ML Models Audited. *"Model Variants" trained for different numbers of epochs $E$. For ViT-small, we use a model pre-trained on imagenet and tune it on CIFAR10. \\
}
    \label{table_metrics}
    \vspace{-0.52cm}
\end{table*}

We instantiate \acronym on target models for four tasks from three data modalities\footnote[1]{Code available here: \url{https://github.com/ubc-systopia/panoramia-privacy-measurement}.}. 
For {\bf image classification}, we consider the CIFAR10~\cite{Krizhevsky2009LearningML}, and CelebA~\cite{liu2015faceattributes} datasets, with varied target models: a four-layers CNN~\cite{oshea2015introduction}, a ResNet101~\cite{he2015deep}, a Wide ResNet-28-10 ~\cite{zagoruyko2017wideresidualnetworks}, a Vision Transformer (ViT)-small ~\cite{dosovitskiy2021imageworth16x16words}, and a DP ResNet18~\cite{he2015deep} trained with DP-SGD~\citep{abadi2016deep} using Opacus~\citep{opacus} at different values of $\epsilon$.
We use StyleGAN2~\cite{karras2020training} for $\gG$.
For {\bf language models}, we fine-tune small GPT-2~\cite{radford2019language} on the WikiText train dataset~\cite{merity2016pointer} (we take a subset of the documents in WikiText-103). 
$\gG$ is again based on small GPT-2, and then fine-tuned on $\DG$. 
We generate samples using top-$k$ sampling~\cite{holtzman2019curious} and a held-out prompt dataset $\DG^{
\textnormal{prompt}} \subset \DG$.
Finally, we conduct experiments on {\bf classification on tabular data}. However for this data modality, \acronym did not detect any meaningful leakage. 
More precisely, we have observed significant variance in the audit values returned by \acronym which make the results inconclusive. 
Nonetheless, we present the results obtained for this task and discuss this issue further in \cref{appendix:tabular-results}. 
\cref{table_metrics} summarizes the tasks, models, and performance. 
More details on the settings, the models used, and implementation details are provided in \cref{appendix:exp_details}. 
 
% \ml{On tabular data, \acronym is not able to detect privacy leakage.}
 %, we fit a Multi-Layer Perceptron (MLP) with 4 hidden layers trained on the Adult dataset \cite{misc_adult_2}, on a binary classification task predicting income $>\$50k$. 
% We use the MST algorithm~\cite{mckenna2021winning} for $\gG$.
%Seb: MST is a differentially-private data synthesis method, what was the value chosen for epsilon

% Our results are organized as follows. 
% First, we evaluate the strength of our baseline (\S\ref{subsec:baseline}), on which the semantics of \acronym's audit rely. Second, we show \acronym detects meaningful privacy leakage in our settings, comparable to the lower-bounds provided by the O(1) approach \citep{steinke2023privacy} (though under weaker requirements) (\S\ref{subsec:main-auditing-results}).
% %Seb: I suggest to give the name to the approach 0(1) the first time we mention it
% Finally, we show that \acronym can detect varying amounts of leakage from models with controlled data leakage using model size and DP (\S\ref{subsec:variations-auditing-results}).

\subsection{Baseline Design and Evaluation}
\label{subsec:baseline}
Our MIA is a loss-based attack, which uses an ML model taking as input both a datapoint $x$ and the value of the loss of target model $f$ on point $x$. 
\cref{appendix:exp_details} details the architectures used for the attack model for each data modality.
\begin{wrapfigure}{r}{0.55\textwidth}
    \begin{minipage}[h!]{0.28\textwidth}  % Align at the top
        \centering
        \begin{adjustbox}{valign=t}  % Vertical alignment at the top
            \resizebox{1\textwidth}{!}{
                \begin{tabular}{|l|l|}
                    \hline
                    Baseline model                               & \(c_{lb}\) \\
                    \hline
                    CIFAR-10 Baseline\textsubscript{\(D_{h}^{\text{tr}}=\text{gen}\)}       & $\mathbf{2.44 \pm 0.19 }$  \\
                    CIFAR-10 Baseline\textsubscript{\(D_{h}^{\text{tr}}=\text{real}\)}      & $2.21 \pm$ 0.17 \\
                    CIFAR-10 Baseline\textsubscript{no helper}                                    & $1.25 \pm 0.24$ \\
                    \hline
                    CelebA Baseline\textsubscript{\(D_{h}^{\text{tr}}=\text{gen}\)}      & $\mathbf{2.05 \pm 0.21}$ \\
                    CelebA Baseline\textsubscript{\(D_{h}^{\text{tr}}=\text{real}\)}      & $1.97\pm 0.22$ \\
                    CelebA Baseline\textsubscript{no helper}                                    & $0.947\pm 0.25$ \\
                    \hline
                    WikiText Baseline\textsubscript{\(D_{h}^{\text{tr}}=\text{gen}\)}      & \bm{$3.31 \pm 0.15$} \\
                    WikiText Baseline\textsubscript{\(D_{h}^{\text{tr}}=\text{real}\)}      & $3.26 \pm 0.14$ \\
                    WikiText Baseline\textsubscript{no helper}                                    & $3.11 \pm 0.15$ \\
                    \hline
                    % Adult Baseline\textsubscript{\(D_{h}^{\text{tr}}=\text{gen}\)}      & \textbf{2.34} \\
                    % Adult Baseline\textsubscript{\(D_{h}^{\text{tr}}=\text{real}\)}      & 2.18 \\
                    % Adult Baseline\textsubscript{no helper}                                    & 2.01 \\
                    % \hline
                \end{tabular}
            }
        \end{adjustbox}
        \captionof{table}{Baseline evaluation with different helper model scenarios.}
        \label{table:baseline_eval-new}
    \end{minipage}%
    \hfill
    \begin{minipage}[h!]{0.25\textwidth}  % Align at the top
        \centering
        \begin{adjustbox}{valign=t}  % Vertical alignment at the top
            \includegraphics[width=\textwidth]{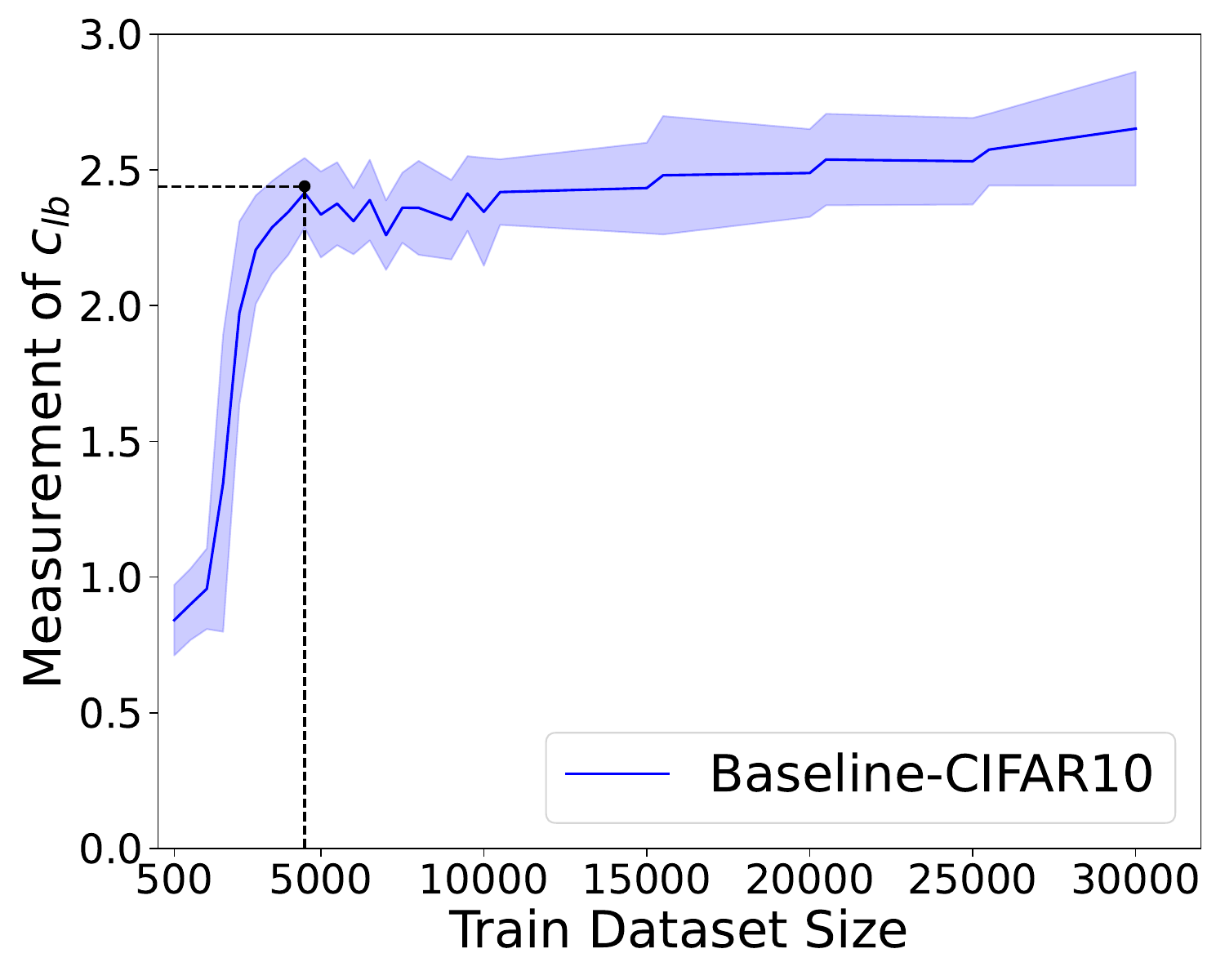}
        \end{adjustbox}
        \captionof{figure}{CIFAR-10 baseline on increasing training size. }
        \label{fig:baseline_eval_figure}
    \end{minipage}
    % \caption{Evaluating the strength of the baseline classifier at distinguishing between real member and synthetic non-member data, to determine a tight bound on $c$-closeness.\ml{can you put that next to one of the baseline perf with more and more data?} \mishaal{done}}
    \label{baseline-evals}
    % \ml{Mishaal, if it's easy can you add a dashed vertical line on the figure that corresponds to the training size you actually use?} 
    % \mishaal{updated with new optimization}
\end{wrapfigure}
Recall from \S\ref{subsec:dp-measurement} the importance of having a tight $\clb$ for our measure $\tilde\epsilon$ to be close to a lower-bound on $\epsilon$-DP.
%, which also requires a strong baseline. 
To increase the performance of our baseline $b$, we mimic the role of the target model $f$'s loss in the MIA using a helper model $\helper$ trained on synthetic data, which adds a loss-based feature to $b$. 

This new feature can be viewed as side information about the data distribution. 
\cref{table:baseline_eval-new} shows the $\clb$ value for different designs for $\helper$. 
The best performance is consistently when $\helper$ is trained on synthetic data before being used as a feature to train the $b$. 
Indeed, such a design reaches a $\clb$ up to $1.19$ larger that without any helper (CIFAR10) and $0.23$ higher than when training on real non-member data, {\em without requiring access to real non-member data}, a key requirement in \acronym. 
We adopt this design in all the following experiments. 
In \cref{fig:baseline_eval_figure} we show that the baseline has enough training data (vertical dashed line) to reach its best performance. \cref{subsec:baseline_moreeval}, shows similar results on GPT-2 as well as different architectures that we tried for the helper model. 
All these pieces of evidence confirm the strength of our baseline.
\vspace{-1.5cm}
\subsection{Main Privacy Measurement Results}
\label{subsec:main-auditing-results}

\vspace{-0.22cm}
We run \acronym on models with different degrees of over-fitting by varying the number of epochs, (see the final accuracy on \cref{table_metrics}) for each data modality. 
More over-fitted models are known to leak more information about their training data due to memorization~\citep{yeom2018privacy, carlini2019secret,LIRA}. 
To show the privacy loss measurement power of \acronym, we compare it with two strong approaches to lower-bounding privacy loss.
\vspace{-0.5cm}
\begin{figure}[htb]
    \centering
    \subfigure[ResNet-101, CIFAR10]{
    \includegraphics[width=0.31\textwidth]{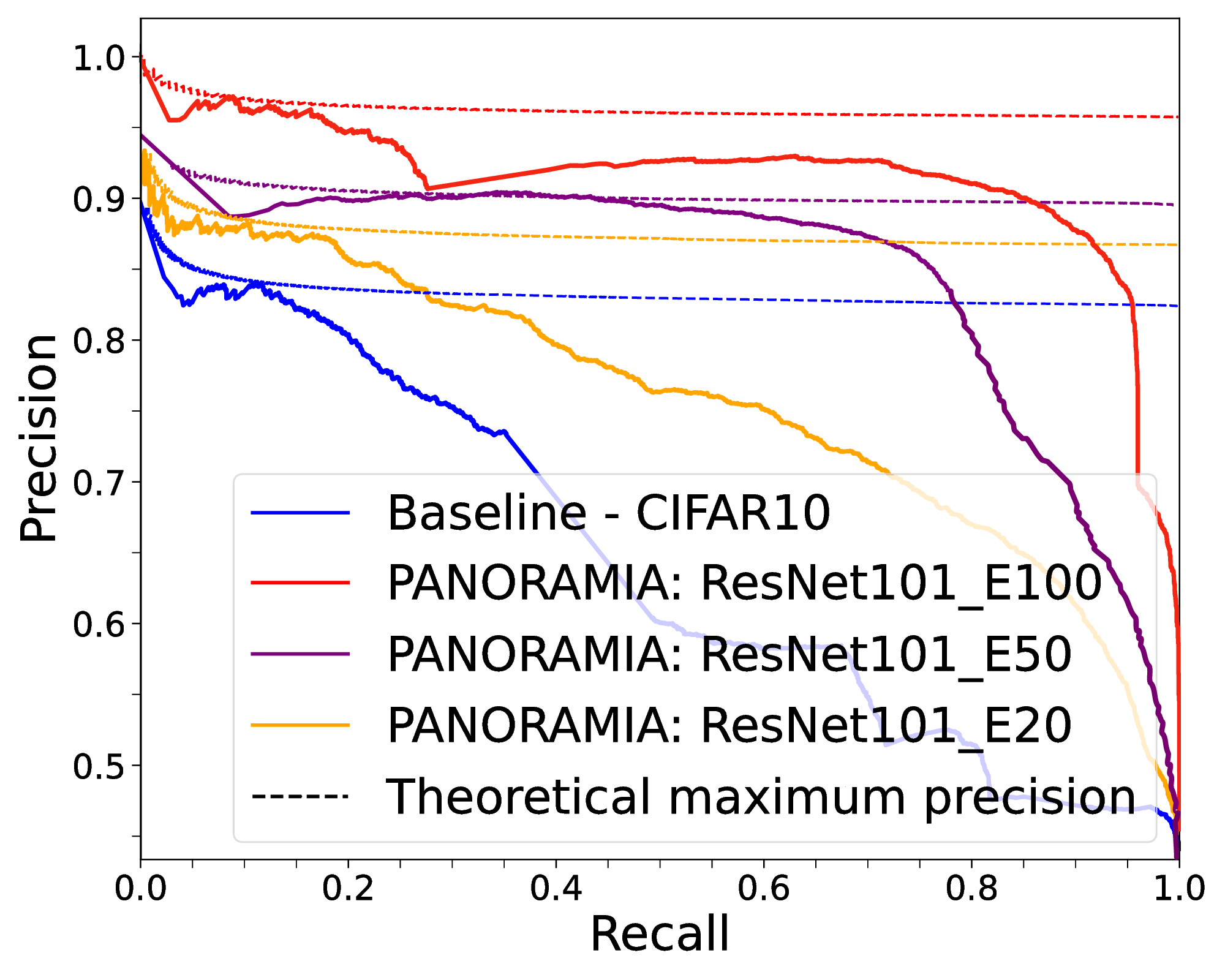}
    \label{fig:image2}
    }
    \subfigure[CNN, CelebA]{
    \includegraphics[width=0.31\textwidth]{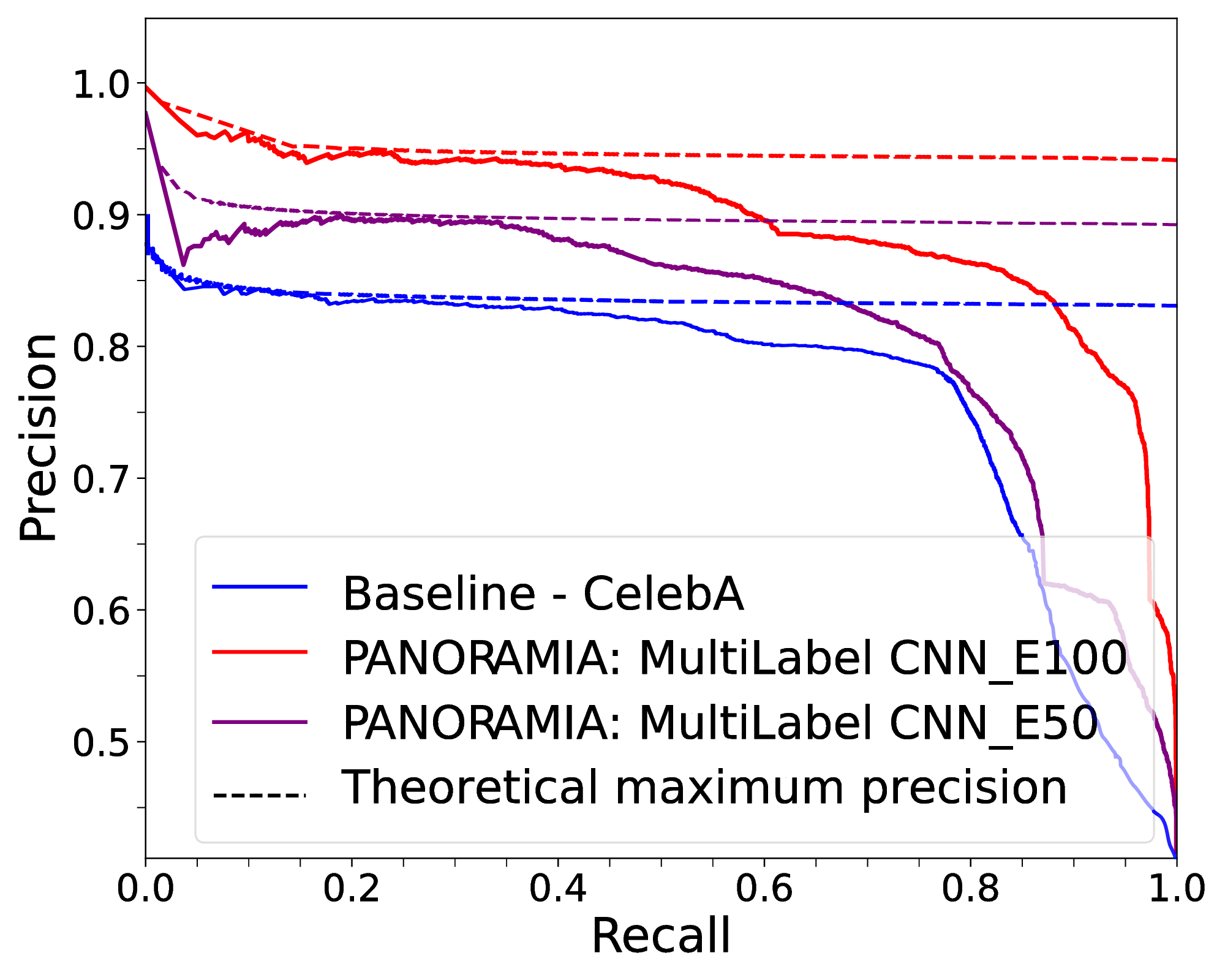}
    \label{fig:image3}
    }
    \subfigure[GPT-2, WikiText]{
    \includegraphics[width=0.31\textwidth]{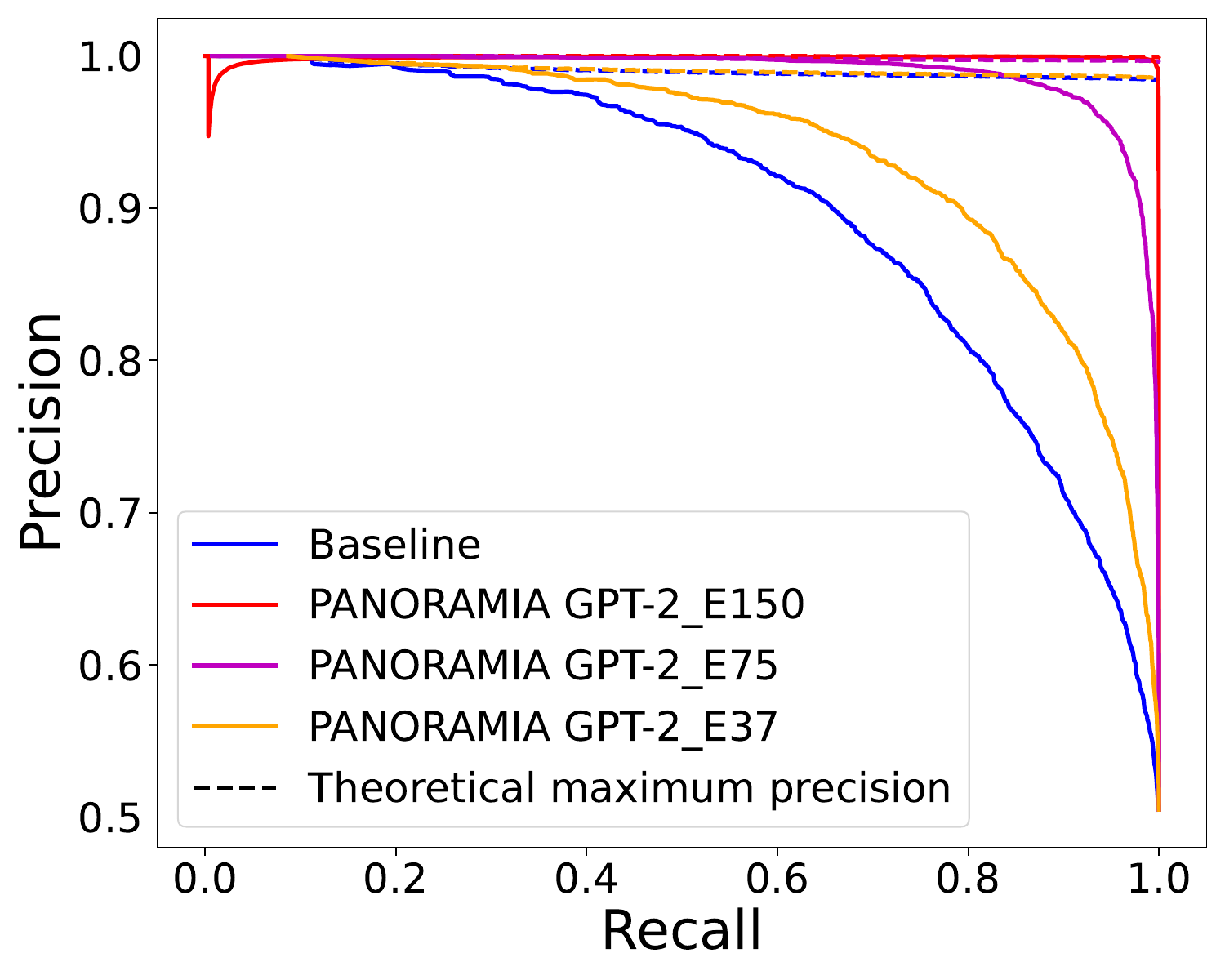}
    \label{fig:gpt2_precision_recall}
    }
    \vspace{-0.1cm}
  \caption{Precision vs. recall comparison between \acronym and the baseline $b$ for our target models, based on one experiment run. 
  The maximum $\clb$ or $\cpepslb$ values set an upper bound on the empirical precision values across different recall levels, indicated by the dashed line.}
  \label{fig:eval:comparison_precision}
 \vspace{-1.73cm}
\end{figure}
\vspace{-1.2cm}
\begin{figure}[htb]
    \centering
    %\vspace{-1.3cm}
    \subfigure[ResNet101, CIFAR10]{
    \includegraphics[width=0.31\textwidth]{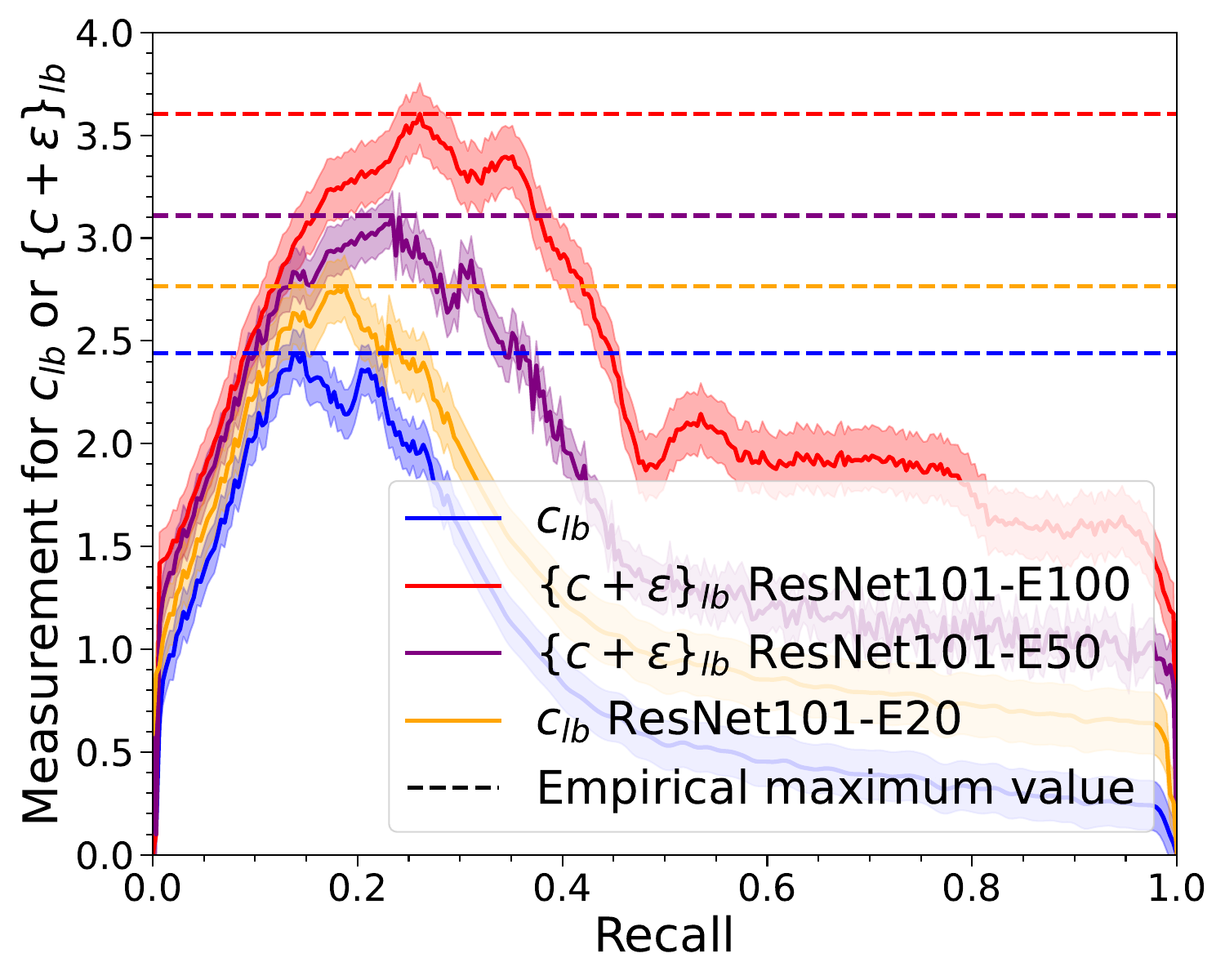}
    \label{fig:image3} 
    }
    \subfigure[CNN, CelebA]{
    \includegraphics[width=0.31\textwidth]{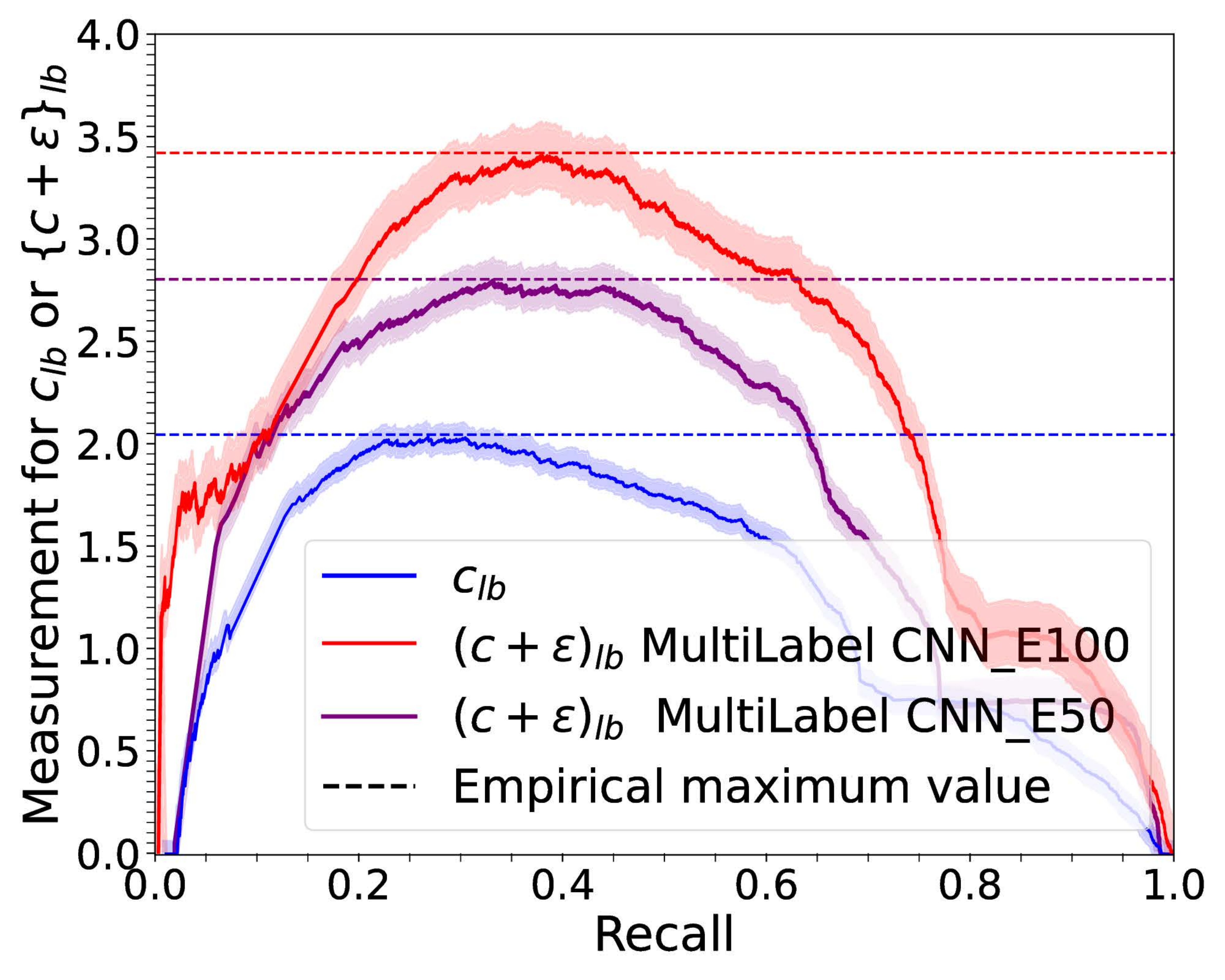}
    \label{fig:image3}
    }
   \subfigure[GPT-2, WikiText]{
    \includegraphics[width=0.31\textwidth]{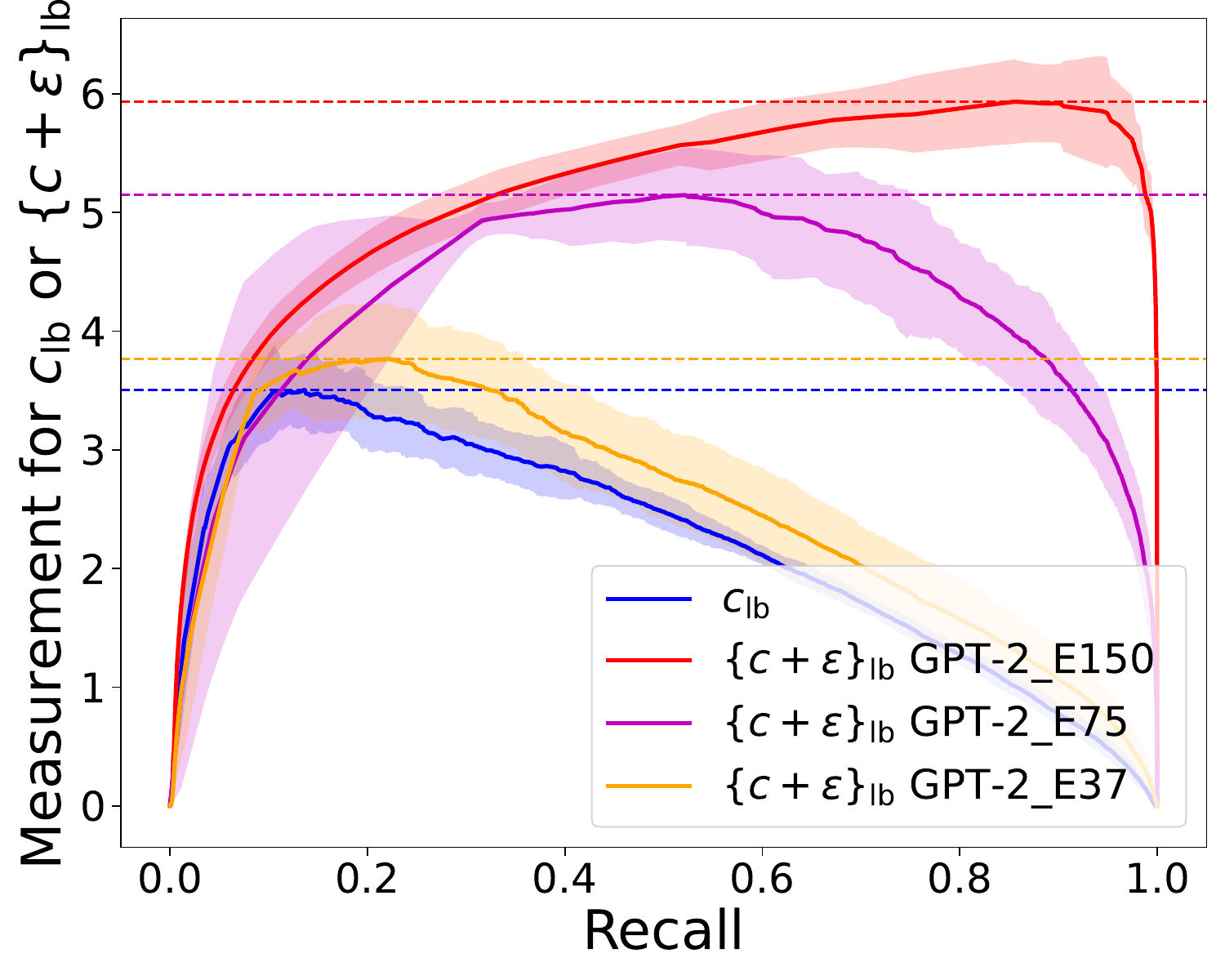}
    \label{fig:gpt2_privacy_recall}
    }
    \vspace{-0.1cm}
    \caption{$\{c+\epsilon\}_{\textnormal{lb}}$ (or $\clb$) vs recall, for our target models, reported over $5$ independent experiment runs.}
    % \mishaal{fig 4a and 4b: new with eps/clb optimization}
    \label{fig:eval:comparison_privacy}
\end{figure}
\vspace{-0.98cm}First, we compare \acronym with the O(1) audit \citep{steinke2023privacy} (in input space, without canaries in a black box setting), which uses a loss threshold MIA 
% (ML based attacks in \cref{appendix:O(1)-ML-based})
% \ml{commented out ``(ML based attacks in \cref{appendix:O(1)-ML-based})'', the ref was broken and it's not crucial (especially since the new phasing makes it clear that RMRN is an ML attack), but feel free to put it back with proper ref}
to detect both members and non-members.
Second, we use a hybrid between O(1) and \acronym, in which we use real non-member data instead of generated data with our ML-based MIA 
%on members only
(called RM;RN for Real Members; Real Non-members).
In both cases $\clb = 0$, but it requires control over training data inclusion and a large set of non-member data.
The privacy loss measured by these techniques gives a target that we hope \acronym to detect.

% First, we use a variation of our approach using real non-member data instead of generated data (called RM;RN for Real Members; Real Non-members). 
% While this basically removes the role of the baseline ($\clb = 0$), it requires access to a large sample of non-member data from the same distributions as members (hard requirement) or the possibility of training costly shadow models to create such non-members.
% %  \mishaal{not sure it is clear why shadow models are needed to create non-members} \ml{if this is not clear at this point, our intro/background is not doing its job, but it's not the place to explain it anyways}
% Second, we rely on the O(1) audit from~\citet{steinke2023privacy}, which is similar to the previous approach but predicts membership based on a loss threshold while also leveraging guesses on non-members in its statistical test. 
% Note that this technique requires control of the training process.
% The privacy loss measured by these techniques gives a target that we hope \acronym to detect.

\cref{fig:eval:comparison_precision} shows the precision of $b$ and \acronym at different levels of recall, and \cref{fig:eval:comparison_privacy} the corresponding value of $\{c+\epsilon\}_{\textnormal{lb}}$ (or $\clb$ for $b$). 
Dashed lines show the maximum value of $\{c+\epsilon\}_{\textnormal{lb}}$ (or $\clb$) achieved (Fig. \ref{fig:eval:comparison_privacy}) (returned by \acronym), and the precision implying these values at different recalls (Fig. \ref{fig:eval:comparison_precision}). \cref{table:audit_values} summarizes those $\{c+\epsilon\}_{\textnormal{lb}}$ (or $\clb$) values, as well as the $\epsilon$ measured by existing approaches. 
We make two key observations.

 \vspace{10cm}
\begin{wraptable}{r}{0.57\textwidth}
% \vspace{-0.2cm}
\centering
\resizebox{0.55\textwidth}{!}{
\begin{tabular}{|c|c|c|c|c|c|}
\hline
\textbf{Target model} $f$    & Audit            & $\mathbf{c_\textrm{lb}}$ & $\mathbf{\{\epsilon + c\}_\textrm{lb}}$ & $\mathbf{\tilde\epsilon}$ & $\mathbf{\epsilon}$\\
\hline
               & \acronym & $2.44\pm 0.190$ & $2.76\pm 0.250$ & $0.33 \pm 0.21$ & - \\
ResNet101\_E20 & \acronym RM;RN & $0.00$ & $0.421 \pm 0.19$ & - &$0.421\pm 0.19$\\
             & O (1) RM;RN    & -   & -       & - & $0.450\pm 0.150$\\
\hline
               &  \acronym & $2.44\pm 0.190$ & $3.11\pm 0.18$ & $0.67 \pm 0.17$& -\\
ResNet101\_E50 & \acronym RM;RN & $0.00$ & $0.71\pm 0.21$  & - &$0.71\pm 0.21$\\
               & O (1) RM;RN    & -   & -     & - &$0.79\pm 0.23$\\
\hline

               & \acronym & $2.44\pm 0.190$ & $3.60\pm 0.21$ & $1.16 \pm 0.19$ & - \\
 ResNet101\_E100 & \acronym RM;RN& $0.00$ & $1.22\pm 0.190$  & - &$1.22\pm 0.19$\\
 & O (1) RM;RN    & -   & -     & - &$1.41\pm 0.180$\\
\hline 
& \acronym & $2.44\pm 0.190$ & $5.01\pm 0.25$ & $2.57 \pm 0.23$ &-\\
 WRN-28-10\_E150& \acronym RM;RN& $0.00$ & $2.87\pm 0.20$  & - &$2.87\pm 0.20$\\
 & O (1) RM;RN    & -   & -     & -& $2.96\pm 0.19$\\
\hline
 & \acronym & $2.44\pm 0.190$ & $5.98\pm 0.23$ & $3.54 \pm 0.23$ &-\\
WRN-28-10\_E300 & \acronym RM;RN& $0.00$ & $3.65\pm 0.22$  & - &$3.65\pm 0.22$\\
 & O (1) RM;RN    & -   & -     & - &$3.73\pm 0.19$\\
\hline
& \acronym & $2.44\pm 0.190$ & $5.03\pm 0.24$ & $2.59 \pm 0.22$ &-\\
ViT\_E35 & \acronym RM;RN& $0.00$ & $2.77\pm 0.21$  & - &$2.77\pm 0.21$\\
 & O (1) RM;RN    & -   & -     & - &$2.86\pm 0.19$\\
\hline
& \acronym & $2.05\pm 0.21$ & $2.81\pm 0.24$ & $ 0.76 \pm 0.23$ &-\\
 CNN\_E50& \acronym RM;RN& $0.00$ & $0.91\pm 0.22$  & - &$0.91\pm 0.22$\\
 & O (1) RM;RN    & -   & -     & -& $1.01\pm 0.18$\\
\hline
 & \acronym & $2.05\pm 0.21$ & $3.38\pm 0.24$ & $1.33 \pm 0.22$ &-\\
CNN\_E100 & \acronym RM;RN& $0.00$ & $1.38\pm 0.21$  & - &$1.38\pm 0.21$\\
 & O (1) RM;RN    & -   & -     & - &$1.46\pm 0.16$\\
\hline
 & \acronym & $3.62 \pm 0.32$ & $3.85 \pm 0.46$ & $0.22 \pm 0.37$&- \\
 GPT2\_E37 & \acronym RM;RN & $0.00$ & $1.04 \pm 0.42$  & - &$1.04 \pm 0.42$\\
 & O (1) RM;RN    & -   & -     & - &$2.82 \pm 0.31$\\
\hline
%  & \acronym & $3.78$ & $3.96$ & $0.18$ &-\\
%  GPT2\_E23 & \acronym{RM;RN}& $0$ & $1.15$  & - &$1.15$\\
%  & O (1) RM;RN    & -   & -     & - &$2.54$\\
% \hline
% Ali: Commennted this row for saving space
 & \acronym & $3.62 \pm 0.32$ & $5.20 \pm 0.34$ & $1.57 \pm 0.45$ &-\\
GPT2\_E75 & \acronym RM;RN& $0.00$ & $2.52 \pm 0.36$  & - &$2.52 \pm 0.36$\\
 & O (1) RM;RN    & -   & -     & - &$4.71 \pm 0.34$\\
\hline
 & \acronym & $3.62 \pm 0.32$ & $6.02 \pm 0.34$ & $2.40 \pm 0.53$ &-\\
 GPT2\_E150 & \acronym RM;RN& $0.00$ & $3.60 \pm 0.41$  & - &$3.60 \pm 0.41$\\
 & O (1) RM;RN    & -   & -     & - &$5.73 \pm 0.08$\\
\hline
ResNet18 $\epsilon=\infty$ & \acronym  & $2.44\pm 0.190$ & $3.87\pm 0.20$ & $1.43\pm 0.21$ & - \\
                             & O (1) RM;RN & - & - & - & $1.471\pm 0.13$ \\
        \hline
        ResNet18 $\epsilon=15$ & \acronym  & $2.44\pm 0.190$ & $3.57\pm 0.19$ & $1.13\pm 0.19$ & - \\
                           & O (1) RM;RN & - & - & - & $1.20\pm 0.18$ \\
        \hline
        ResNet18 $\epsilon=10$ & \acronym & $2.44\pm 0.190$ & $2.70\pm 0.25$ & $0.26 \pm 0.22$ & - \\
                           & O (1) RM;RN & - & - & - & $0.28\pm0.14$ \\
        \hline
        ResNet18 $\epsilon=2$ & \acronym  & $2.44\pm 0.190$& $1.709\pm 0.23$ & $0.00$ & - \\
                          & O (1) RM;RN & - & - & - & $0.05\pm0.12$ \\
        \hline
        ResNet18 $\epsilon=1$ & \acronym  & $2.44\pm 0.190$ & $1.412\pm 0.12$ & $0.00$ & - \\
                          & O (1) RM;RN & - & - & - & $0.00$ \\
        \hline
\end{tabular}}
\caption{Privacy measurements on different target models ($m=10k$). O(1) RM;RN with our ML-based attack for $\epsilon$-DP models shown in \cref{subfig:panoramia_vs_O1_m_dp_increase}.
% \mishaal{updated with clb/epslb optimization.}
}
\vspace{-0.2cm}
\label{table:audit_values}
\end{wraptable}
First, the best prior method (whether RM;RN or O(1)) measures a larger privacy loss ($\tilde\epsilon \leq \epsilon$). 
Overall, these results empirically confirm the strength of $b$, as we do not seem to spuriously assign differences between $\gG$ and $\gD$ to our privacy loss proxy $\tilde\epsilon$. 
We also note that O(1) tends to perform better, due to its ability to rely on non-member detection, which improves the power of the statistical test at equal data sizes. 
Such tests are not available in \acronym given our one-sided closeness definition for $\gG$ (see \S\ref{sec:theory}), and we keep that same one-sided design for RM;RN for the sake of comparison.
% \vspace{5cm}

Second, the values of $\tilde\epsilon$ measured by \acronym are close to those of the methods against which we compared.
In particular, despite a more restrictive adversary model (\emph{i.e.}, no non-member data, no control over the training process, and no shadow model training), \acronym is able to detect meaningful amounts of privacy loss, comparable to that of state-of-the-art methods.

For instance, on a non-overfitted CIFAR-10 ResNet101 model (E20), \acronym detects a privacy loss of $0.33$, while using real non-member (RM;RN) data yields $0.42$, and controlling the training process $O(1)$ gets $0.45$.
The relative gap gets even closer on models that reveal more about their training data. 
Indeed, for the most over-fitted model (E100), $\tilde\epsilon=1.16$ is very close to RM;RN ($\epsilon=1.22$) and O(1) ($\epsilon=1.41$).
This also confirms that the leakage detected by \acronym on increasingly overfitted models does augment, which is confirmed by prior state-of-the-art methods. 
For instance, NLP models' $\tilde\epsilon$ goes from $0.22$ to $2.40$ ($2.82$ to $5.73$ for O(1)). 
\cref{appendix:tabular-results} details results on tabular data, in which \acronym is not able to detect significant privacy loss.

\textbf{Privacy Measurement of $\epsilon$-DP Models:} 

\begin{wrapfigure}{r}{0.35\textwidth}
\begin{center}
    \includegraphics[width=\linewidth]{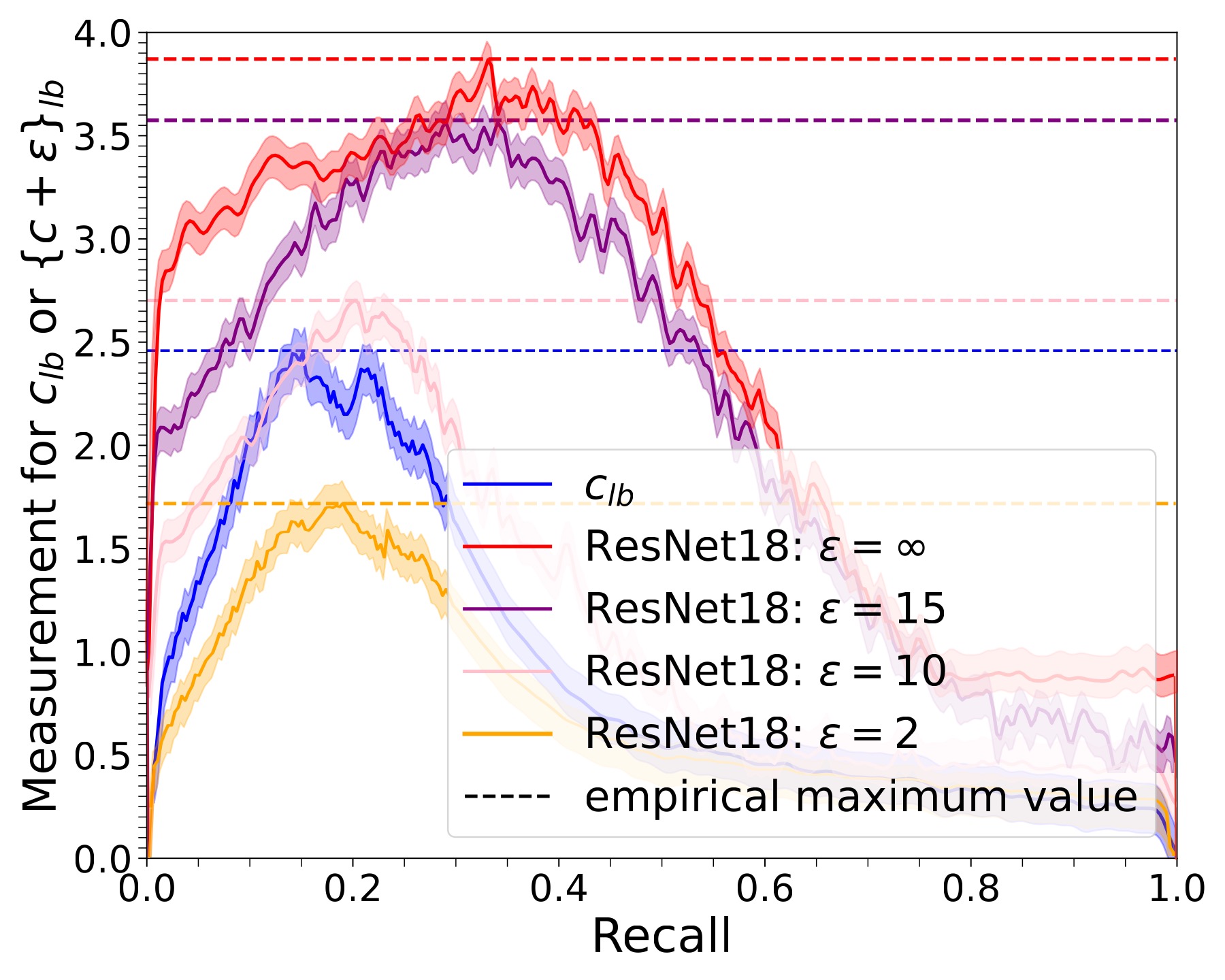}
    \caption{ResNet18, CIFAR-10, $\epsilon$-DP for various $\epsilon$ values. 
    % \acronym detects higher privacy leakage as the $\epsilon$ value increases.
    }
    \label{fig:dp-res18-plot}
\end{center}
    \vspace{-0.758cm}

\end{wrapfigure}
Another avenue to varying privacy leakage is by training Differentially-Private (DP) models with different values of $\epsilon$. 
We evaluate \acronym on DP ResNet-18 models on CIFAR10, with $\epsilon$ values shown on \cref{table:audit_values} and \cref{fig:dp-res18-plot}. 
The hyper-parameters were tuned independently for the highest train accuracy (see Appendix \ref{appendix:dp_sup_results} for more results and discussion).
As the results show, neither \acronym nor O(1) detect privacy loss on the most private models ($\epsilon=1, 2$). 
At higher values of $\epsilon=10,15$ (\emph{i.e.}, less private models) and $\epsilon=\infty$ (\emph{i.e.}, non-private model) \acronym does detect an increasing level of privacy leakage with $\tilde{\epsilon}_{\epsilon=10} < \tilde{\epsilon}_{\epsilon=15} < \tilde{\epsilon}_{\epsilon=\infty}$. 
In this regime, the O(1) approach detects a larger, though comparable, amount of privacy loss. 

We also evaluate \acronym on DP (fine-tuned) large language models (see Appendix \ref{appendix:dp_sup_results} and \cref{table:audit-nlp-dp-target-models} for details). We fine-tune GPT2-Small target models with DP-SGD on the WikiText dataset, with various values of $\epsilon$. Neither \acronym nor O(1) are able to measure a positive privacy loss. 

% \begin{table}[H]
%     \centering
%     \resizebox{0.6\textwidth}{!}{
%     \begin{tabular}{|l|l|l|l|l|l|}
%         \hline
%         \textbf{Target model} & \textbf{Audit} & $\mathbf{c_{lb}}$ & $\mathbf{\varepsilon + c_{lb}}$ & $\mathbf{\tilde{\varepsilon}}$ & $\mathbf{\varepsilon}$ \\
%         \hline
%         ResNet18 $\epsilon=\infty$ & \acronym RM;GN & 2.323 & 3.425 & 1.102 & - \\
%                              & O (1) RM;RN & - & - & - & 1.471 \\
%         \hline
%         ResNet18 $\epsilon=15$& \acronym RM;GN & 2.323 & 3.31 & 1.017 & - \\
%                            & O (1) RM;RN & - & - & - & 1.20 \\
%         \hline
%         ResNet18 $\epsilon=10$& \acronym RM;GN & 2.323 & 2.563 & 0.24 & - \\
%                            & O (1) RM;RN & - & - & - & 0.28 \\
%         \hline
%         ResNet18 $\epsilon=2$ & \acronym RM;GN & 2.323 & 1.65 & 0 & - \\
%                           & O (1) RM;RN & - & - & - & 0.05 \\
%         \hline
%         ResNet18 $\epsilon=1$ & \acronym RM;GN & 2.323 & 1.492 & 0 & - \\
%                           & O (1) RM;RN & - & - & - & 0 \\
%         \hline
%     \end{tabular}}
%     \vspace{0.3cm}
%     \caption{$\epsilon$-DP ResNet18 on CIFAR-10.}
%     \label{tab:dpResNet18-table-main}
% \end{table}

\subsection{Leveraging More Data to Improve Privacy Measurements}

\acronym can leverage much more data for its measurement (up to the whole training set size for training and testing the MIA), while in our setting O(1) is limited by the size of the test set (for non-members). 
As we have seen in \cref{subfig:main-test-inc-resnet100}, \acronym can leverage larger test set size to measure higher privacy loss values than O(1), up to $2.62$ on the ResNet101 trained for 100 epochs, versus $2.23$ for O(1). 
It is also important to note that at small amounts of data O(1) measures a larger privacy loss ($\tilde\epsilon \leq \epsilon$). 
These findings provide additional empirical evidence for the robustness of $b$, indicating that we are not erroneously attributing differences between $\gG$ and $\gD$ to our privacy loss proxy $\tilde\epsilon$.
However, in the case of text data (on the GPT-2 trained for 150 epochs), \cref{subfig:main-test-inc-gpt2},  we hit an upper-bound on the power of the hypothesis test which estimates $\cpepslb$  using $20k$ test samples. The maximum measurement for $\tilde\epsilon$ we can achieve is around $2.59$ with a $20k$ test set, as opposed to O(1), which reaches $5.73$ with $10k$ test set size.

\begin{figure}[htb]
    \centering
    \subfigure[WideResNet-28-10, CIFAR10]
    {
        \includegraphics[width=0.22\textwidth]{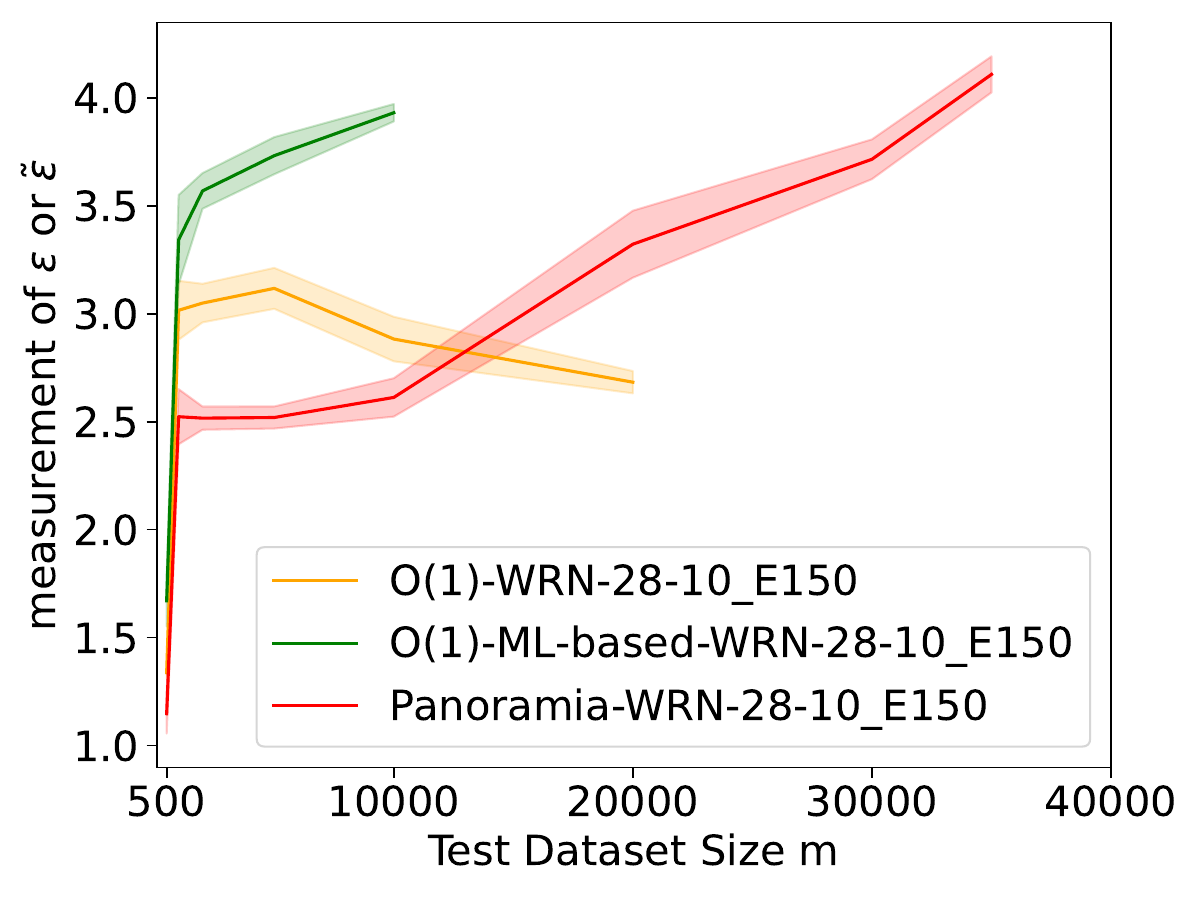}
        \label{subfig:main-test-inc-wrn150}
    } \hfill
    \subfigure[ResNet-101, CIFAR10]
    {
        \includegraphics[width=0.22\textwidth]{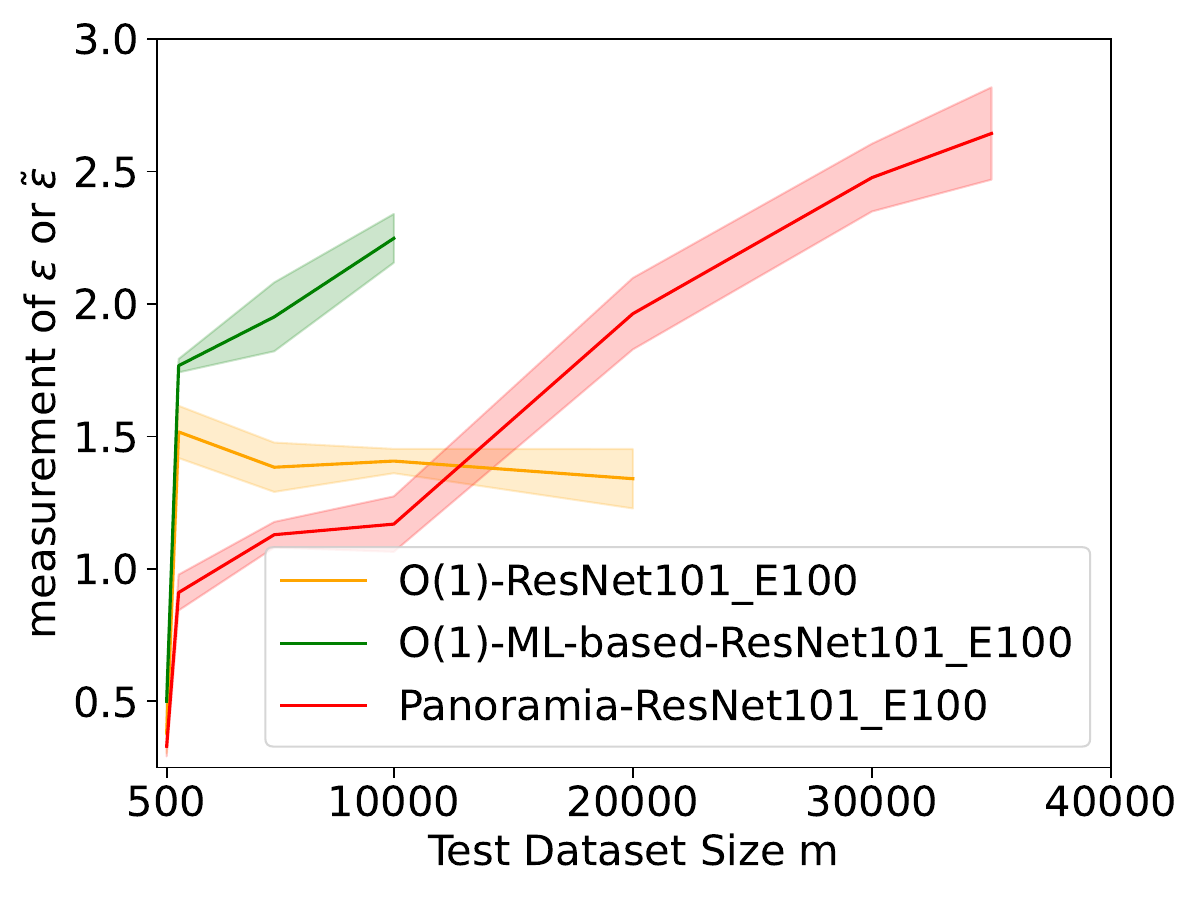}
        \label{subfig:main-test-inc-resnet100}
    } \hfill
    \subfigure[$\epsilon$-DP ResNet18, CIFAR10]
    {
        \includegraphics[width=0.22\textwidth]{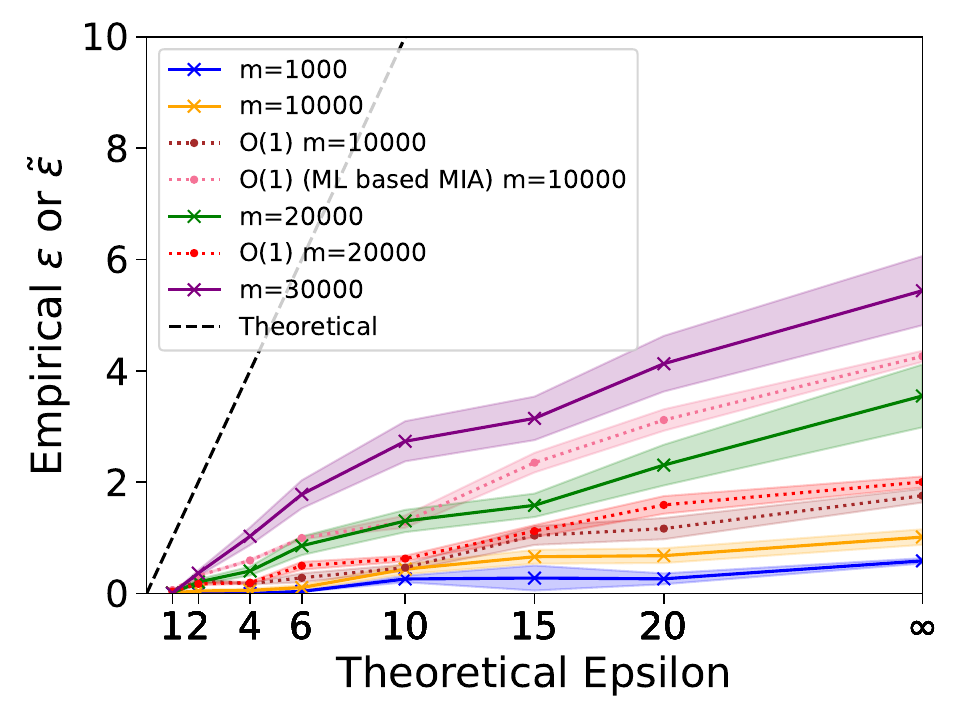}
        \label{subfig:panoramia_vs_O1_m_dp_increase}
    } \hfill
    \subfigure[GPT-2, WikiText]
    {
        \includegraphics[width=0.22\textwidth]{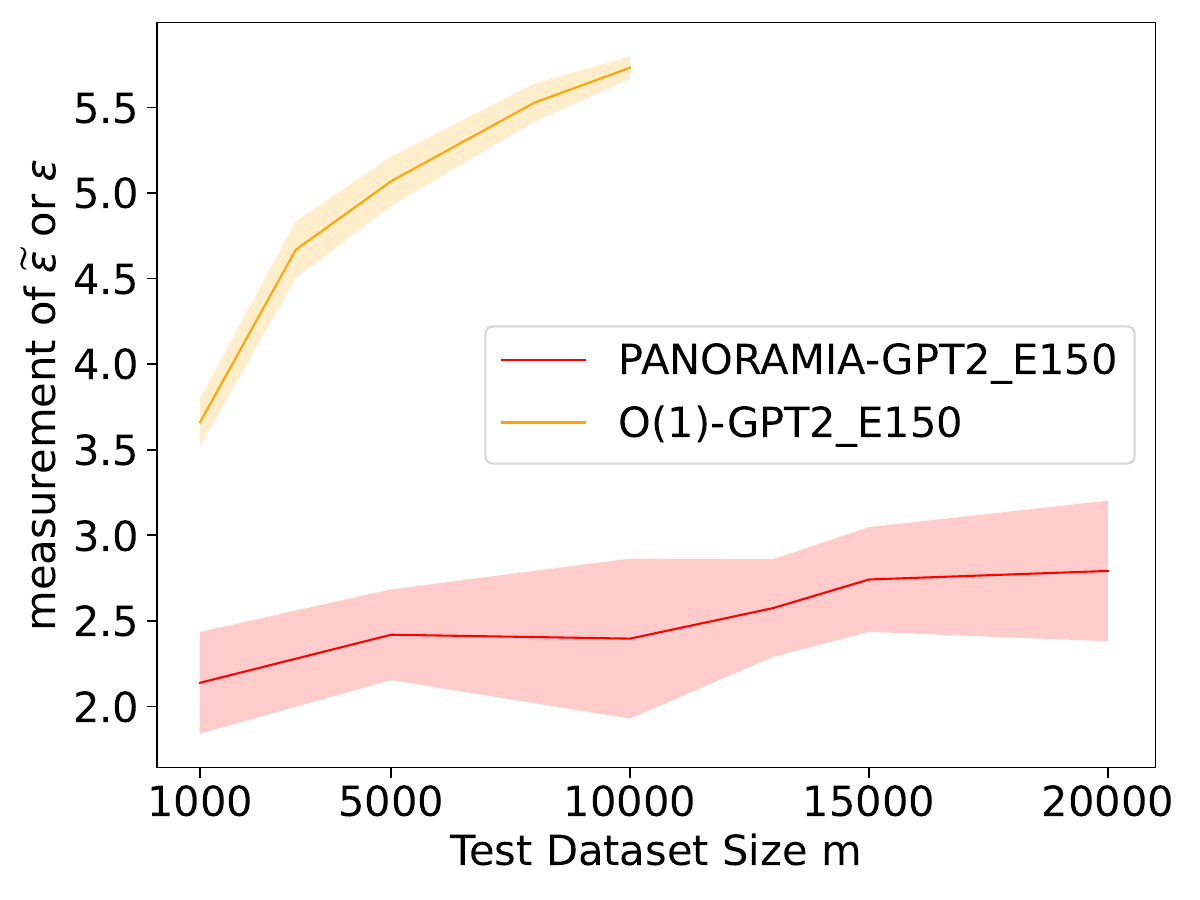}
        \label{subfig:main-test-inc-gpt2}
    }
    
    \caption{Effect of increasing test set size on \acronym's privacy measurement for our target models. In the case of the image dataset, increasing the number of auditing examples allows us to achieve tighter empirical measurements for privacy leakage despite a restricted adversary. For the case of the language modeling task, when increasing test set size $m$, we hit an upper bound on the power of the hypothesis test due to test dataset size, and do not see significant improvement in privacy measurement in this case.}
    \label{fig:eval:test-size-growth}
\end{figure}

% \begin{figure}[H]
%   \centering
%   {\includegraphics[width=0.5\linewidth]{images/panoramia-o1-m-dp.pdf} \label{fig:panoramia_vs_O1_m_dp_increase}
%    \caption{Effect of the number of auditing examples (m) on $\tilde\epsilon$ in the case of \acronym (solid lines) and $\epsilon$ in the case of O(1) (dotted lines) for $\epsilon$-DP ResNet18 model.}
% \end{figure}

\textbf{Increasing Test Set Size for DP models:} With an increase in the test set size $m$, we demonstrate that we can achieve tighter bounds on the privacy measurement $\tilde\epsilon$. \cref{subfig:panoramia_vs_O1_m_dp_increase} shows that as $m$ increases, \acronym (solid lines) can measure higher privacy leakage for various $\epsilon$-DP ResNet18 models. 
The reason for this is that \acronym can leverage a much higher test set size (up to the whole training set size for training and testing the MIA), in comparison to O(1) in our setting (due to being limited by the number of real non-member data points available).

\section{Impact, Limitations and Future Directions}
\label{limitations}
We have introduced a new approach to quantify the privacy leakage from ML models, in settings in which individual data contributors (such as a hospital in a cross-site FL setting or a user of a text auto-complete service) measure the leakage of their own, known partial training data in the final trained model.
Our approach does not introduce new attack capabilities (it would likely benefit from future progress on MIA performance though), but can help model providers and data owners assess the privacy leakage incurred by their data, due to inclusion in ML models training datasets.
Consequently, we believe that any impact from our proposed work is likely to be positive.

However, \acronym suffers from a major limitation: its privacy measurement is not a lower-bound on the privacy loss (see details in \cref{subsec:dp-measurement}). Providing a proper lower-bound $\epsilon_{lb}$ on privacy leakage in this setting is an important avenue for future work. A promising direction is to devise a way to measure or enforce an upper-bound on $c$, thereby yielding $\epsilon_{lb}$ via $\{c+\epsilon\}_{lb} - c_{ub}$.
Despite this shortcoming, we believe that the new measurement setting we introduce in this work is important. 
Indeed, recent work has shown that MIA evaluation on foundation models suffers from a similar lack of non-member data (since all in distribution data is included in the model) \cite{das2024blind,duan2024membership,meeus2024inherent}.
The theory we develop in this paper provides a meaningful step towards addressing privacy measurements in this setting, and provides a more rigorous approach to privacy benchmarks for such models.
We also demonstrate that \acronym's privacy measurements can also be empirically valuable for instance for providing improved measurements with more data (\cref{fig:eval:test-size-growth}).

\section*{Acknowledgments}
Mathias L\'ecuyer is grateful for the support of the Natural Sciences and Engineering Research Council of Canada (NSERC) [reference number RGPIN-2022-04469], as well as a Google Research Scholar award. Sébastien Gambs is supported by the Canada Research Chair program, a Discovery Grant from the Natural Sciences and Engineering Research Council of Canada (NSERC) as well as the DEEL Project CRDPJ 537462-18 funded by the NSERC and the Consortium for Research and Innovation in Aerospace in Québec (CRIAQ).
This research was enabled by computational support provided by the Digital Research Alliance of Canada (alliancecan.ca). Hadrien Lautraite is supported by National Bank of Canada and a bursary from the Canada Research Chair in Privacy-preserving and Ethical Analysis of Big Data..

\bibliography{panoramia_2024}
\bibliographystyle{icml2024}

%%%%%%%%%%%%%%%%%%%%%%%%%%%%%%%%%%%%%%%%%%%%%%%%%%%%%%%%%%%%%%%%%%%%%%%%%%%%%%%
%%%%%%%%%%%%%%%%%%%%%%%%%%%%%%%%%%%%%%%%%%%%%%%%%%%%%%%%%%%%%%%%%%%%%%%%%%%%%%%
% APPENDIX
%%%%%%%%%%%%%%%%%%%%%%%%%%%%%%%%%%%%%%%%%%%%%%%%%%%%%%%%%%%%%%%%%%%%%%%%%%%%%%%
%%%%%%%%%%%%%%%%%%%%%%%%%%%%%%%%%%%%%%%%%%%%%%%%%%%%%%%%%%%%%%%%%%%%%%%%%%%%%%%

% \onecolumn

\clearpage
\newpage
\appendix
\section{Notations}
\label{appendix:notations}
\cref{table:notations} summarizes the main notations used in the paper.
\begin{table}[h]
  
  \centering
  \resizebox{0.65\textwidth}{!}{
  \begin{tabular}{@{}l@{~~}l@{}}
  \toprule
  \bf Notation    & \bf Description \\
  \midrule
  $f$ & the target model to be audited. \\
  $\gG$ & the generative model \\
  %$\mathcal{M}$   & a stochastic training algorithm \\
  $\gD$   & distribution over auditor samples \\
  %$z \sim \mathcal{D}$   & draw a sample $z$ from $\mathcal{D}$ \\
  $\DT$ & (subset of the) training set of the target model $f$ from $\gD$ \\
  $\DG$ & training set of the generative model $\gG$, with $\DG \subset \DT$ \\
  $\Din$ & member auditing set, with $\Din \subset \DT$ and $\Din \cap \DG = \{\}$ \\
  $\Dout$ & non-member auditing set, with $\Dout \sim \gG$ \\
  $\Dinouttrte$ & training and testing splits of $\Din$ and $\Dout$ \\
  $m$ & $|\Din| = |\Dout| \triangleq m$ \\
  $b$ & baseline classifier for $\Din$ vs. $\Dout$ \\
%  $e$ & embedding model used in baseline classifier b to strengthen it.\\ % [ML] This is too long, and not discussed in those sections anyway so it doesn't belong here
  %\DMRT & the member part of the training dataset of the MIA and baseline attacks, \emph{i.e.} \DMRT $ \subset $ \DT and \DMRT $ \cap $ \DG $ = \emptyset$ \\
  \bottomrule

  \end{tabular}}

  \vspace{0.11cm}
    \caption{Summary of notations}
  \label{table:notations}  
 
\end{table}

\begin{hideIfProofs}
\section{Proofs}
\label{appendix:proofs}

For both Proposition~\ref{prop:gen-test} and Proposition~\ref{prop:dp-test}, we state the proposition again for convenience before proving it.

\subsection{Proof of Proposition~\ref{prop:gen-test}}
\label{appendix:proof-prop1}

\begin{numberedprop}[\ref{prop:gen-test}]
Let $\gG$ be $c$-close, and $T^b \triangleq B(S, X)$ be the guess from the baseline. 
Then, for all $v \in \sR$ and all $t$ in the support of $T$:
\begin{align*}
& \sP_{S, X, T^b}\Big[ \sum_{i=1}^m T^b_i \cdot S_i \geq v \ | \ T^b = t^b \Big] \leq \underset{{S' \sim \textrm{Bernoulli}(\frac{e^c}{1+e^c})^m}}\sP\Big[ \sum_{i=1}^m t^b_i \cdot S'_i \geq v \Big] \triangleq \beta^b(m, c, v, t^b)
\end{align*}
\end{numberedprop}
\begin{proof}
Notice that under our baseline model $B(s, x) = \{b(x_1), b(x_2), \ldots, b(x_m)\}$, and given that the $X_i$ are i.i.d., we have that: $S_{<i} \ind T^b_{<i} \ | \ X_{<i}$, since $T^b_i = B(S, X)_i$'s distribution is entirely determined by $X_i$;
and $S_{\leq i} \ind T^b_{> i} \ | \ X_{<i}$ since the $X_{i}$ are sampled independently from the past.
% Our generative process for $X_i$ also implies that $X_i \ind S_{<i}, X_{<i} \ | \ S_i$.

We study the distribution of $S$ given a fixed prediction vector $t^b$, one element $i \in [m]$ at a time:
\begin{align*}
& \sP\big[ S_i = 1 \ | \ T^b = t^b, S_{<i} = s_{<i}, X_{\leq i} = x_{\leq i}\big] \\
&= \sP\big[ S_i = 1 \ | \ S_{<i} = s_{<i}, X_{\leq i} = x_{\leq i}\big] \\
&= \sP\big[ X_i \ | \ S_i = 1, S_{<i} = s_{<i}, X_{< i} = x_{< i}\big] \frac{\sP\big[ S_i = 1 \ | \ S_{<i} = s_{<i}, X_{< i} = x_{< i} \big]}{\sP\big[ X_i \ | \ S_{<i} = s_{<i}, X_{< i} = x_{< i} \big]} \\
&= \frac{\sP\big[ X_i \ | \ S_i = 1, S_{<i} = s_{<i}, X_{< i} = x_{< i}\big]\sP\big[ S_i = 1 \big]}{\sP\big[ X_i \ | \ S_{<i} = s_{<i}, X_{< i} = x_{< i} \big]} \\
&= \frac{\sP\big[ X_i \ | \ S_i = 1 \big]\frac{1}{2}}{\sP\big[ X_i \ | \ S_i=1 \big]\frac{1}{2} + \sP\big[ X_i \ | \ S_i=0 \big]\frac{1}{2}} \\
&= \frac{1}{1 + \frac{\sP\big[ X_i \ | \ S_i=0 \big]}{\sP\big[ X_i \ | \ S_i=1 \big]}} = \frac{1}{1 + \frac{\sP_\gG\big[ X_i \big]}{\sP_\gD\big[ X_i \big]}} \leq \frac{1}{1 + e^{-c}} = \frac{e^c}{1 + e^{c}}
\end{align*}
The first equality uses the independence remarks at the beginning of the proof, the second relies Bayes' rule, while the third and fourth that $S_i$ is sampled i.i.d from a Bernoulli with probability half, and $X_i$ i.i.d. conditioned on $S_i$. 
The last inequality uses Definition \ref{c-closeness-def} for $c$-closeness.

Using this result and the law of total probability to introduce conditioning on $X_{\leq i}$, we get that:
\begin{align*}
& \sP\big[ S_i = 1 \ | \ T^b = t^b, S_{<i} = s_{<i}\big] \\
& = \sum_{x_{\leq i}} \sP\big[ S_i = 1 \ | \ T^b = t^b, S_{<i} = s_{<i},  X_{\leq i} = x_{\leq i} \big] \sP\big[ X_{\leq i} = x_{\leq i} \ | \ T^b = t^b, S_{<i} = s_{<i} \big] \\
& \leq \sum_{x_{\leq i}} \frac{e^c}{1 + e^{c}} \sP\big[ X_{\leq i} = x_{\leq i} \ | \ T^b = t^b, S_{<i} = s_{<i} \big] ,
\end{align*}
and hence that:
\begin{equation}\label{eq:stoc-dom-gen}
\sP\big[ S_i = 1 \ | \ T^b = t^b, S_{<i} = s_{<i}\big] \leq \frac{e^c}{1 + e^{c}}
\end{equation}

We can now proceed by induction: assume inductively that $W_{m-1} \triangleq \sum_{i=1}^{m-1} T^b_i \cdot S_i$ is stochastically dominated (see Definition 4.8 in \cite{steinke2023privacy}) by $W'_{m-1} \triangleq \sum_{i=1}^{m-1} T^b_i \cdot S'_i$, in which $S' \sim \textrm{Bernoulli}(\frac{e^c}{1 + e^{c}})^{m-1}$. 
Setting $W_1 = W'_1 = 0$ makes it true for $m=1$. Then, conditioned on $W_{m-1}$ and using Eq. \ref{eq:stoc-dom-gen}, $T^b_m \cdot S_m = T_m \cdot \1\{S_m = 1\}$ is stochastically dominated by $T^b_m \cdot \textrm{Bernoulli}(\frac{e^c}{1 + e^{c}})$.
Applying Lemma 4.9 from \cite{steinke2023privacy} shows that $W_{m}$ is stochastically dominated by $W'_{m}$, which proves the induction and implies the proposition's statement.
\end{proof}

\subsection{Proof of Proposition~\ref{prop:dp-test}}
\label{appendix:proof-prop2}

\begin{numberedprop}[\ref{prop:dp-test}]
Let $\gG$ be $c$-close, $f$ be $\epsilon$-DP, and $T^a \triangleq A(S, X, f)$ be the guess from the membership audit. Then, for all $v \in \sR$ and all $t$ in the support of $T$:
\begin{align*}
& \sP_{S, X, T^a}\Big[ \sum_{i=1}^m T^a_i \cdot S_i \geq v \ | \ T^a = t^a \Big] \leq \underset{{S' \sim \textrm{Bernoulli}(\frac{e^{c+\epsilon}}{1+e^{c + \epsilon}})^m}}\sP\Big[ \sum_{i=1}^m t^a_i \cdot S'_i \geq v \Big] \triangleq \beta^a(m, c, \epsilon, v, t^a)
\end{align*}
\end{numberedprop}
\begin{proof}
Fix some $t^a \in \sR_+^m$. 
We study the distribution of $S$ one element $i \in [m]$ at a time:
\begin{align*}
& \sP\big[ S_i = 1 \ | \ T^a = t^a, S_{<i} = s_{<i}, X_{\leq i} = x_{\leq i}\big] \\
&= \sP\big[ T^a = t^a \ | \ S_i = 1, S_{<i} = s_{<i}, X_{\leq i} = x_{\leq i}\big] \frac{\sP\big[ S_i = 1 \ | \ S_{<i} = s_{<i}, X_{\leq i} = x_{\leq i} \big]}{\sP\big[ T^a = t^a \ | \ S_{<i} = s_{<i}, X_{\leq i} = x_{\leq i} \big]} \\
&\leq \frac{1}{1 + e^{-\epsilon}\frac{\sP\big[ S_i = 0 \ | \ S_{<i} = s_{<i}, X_{\leq i} = x_{\leq i} \big]}{\sP\big[ S_i = 1 \ | \ S_{<i} = s_{<i}, X_{\leq i} = x_{\leq i} \big]}} \\
&\leq \frac{1}{1 + e^{-\epsilon}e^{-c}} = \frac{e^{c+\epsilon}}{1 + e^{c+\epsilon}}
\end{align*}
The first equality uses Bayes' rule. 
The first inequality uses the decomposition:
\begin{align*}
& \sP\big[ T^a = t^a \ | \ S_{<i} = s_{<i}, X_{\leq i} = x_{\leq i} \big]\\
&= \sP\big[ T^a = t^a \ | \ S_i = 1, S_{<i} = s_{<i}, X_{\leq i} = x_{\leq i} \big] \cdot \sP\big[ S_i = 1 \ | \ S_{<i} = s_{<i}, X_{\leq i} = x_{\leq i} \big] \\
&+ \sP\big[ T^a = t^a \ | \ S_i = 0, S_{<i} = s_{<i}, X_{\leq i} = x_{\leq i} \big] \cdot \sP\big[ S_i = 0 \ | \ S_{<i} = s_{<i}, X_{\leq i} = x_{\leq i} \big] ,
\end{align*}
and the fact that $A(s, x, f)$ is $\epsilon$-DP w.r.t. $s$ and hence that:

\begin{equation}\label{eq:dp-bound-f}
    \frac{\sP\big[ T^a = t^a \ | \ S_i = 0, S_{<i} = s_{<i}, X_{\leq i} = x_{\leq i}\big]}{\sP\big[ T^a = t^a \ | \ S_i = 1, S_{<i} = s_{<i}, X_{\leq i} = x_{\leq i}\big]} \geq e^{-\epsilon} .
\end{equation}

This inequality might seem surprising at first, as we know from the privacy game (\cref{def:auditing-game}) that $S$ is independent of $f$, that is $S \ind f$. However, these two quantities are not conditionally independent when conditioning on $X$ (where $X$ is the test set for the MIA/baseline), (i.e., $S \not\ind f | X$). This is because $X_i$ is either a member or non-member based on $S_i$. So when $S_i=1$, then $f$ was trained on $X_i$, which makes $f$ and $S_i$ dependent if membership can be detected through $f$. If $f$ is $\epsilon$-DP, we can bound this dependency using DP properties, which is what we do in \cref{eq:dp-bound-f}.

The second inequality uses that:
\begin{align*}
& \frac{\sP\big[ S_i = 0 \ | \ S_{<i} = s_{<i}, X_{\leq i} = x_{\leq i} \big]}{\sP\big[ S_i = 1 \ | \ S_{<i} = s_{<i}, X_{\leq i} = x_{\leq i} \big]} \\
&= \frac{\sP\big[ X_i = x_i \ | \ S_i = 0, S_{<i} = s_{<i}, X_{< i} = x_{< i} \big]}{\sP\big[ X_i = x_i \ | \ S_i = 1, S_{<i} = s_{<i}, X_{< i} = x_{< i} \big]} \cdot \frac{\sP\big[ S_i = 0 \ | \  S_{<i} = s_{<i}, X_{< i} = x_{< i} \big]}{\sP\big[ S_i = 1 \ | \ S_{<i} = s_{<i}, X_{< i} = x_{< i} \big]} \\
&= \frac{\sP\big[ X_i = x_i \ | \ S_i = 0, S_{<i} = s_{<i}, X_{< i} = x_{< i} \big]}{\sP\big[ X_i = x_i \ | \ S_i = 1, S_{<i} = s_{<i}, X_{< i} = x_{< i} \big]} \cdot \frac{1/2}{1/2} \\
& = \frac{\sP_\gG\big[ X_i \big]}{\sP_\gD\big[ X_i \big]} \geq e^{-c}
\end{align*}

As in Proposition~\ref{prop:gen-test}, applying the law of total probability to introduce conditioning on $X_{\leq i}$ yields:
\begin{equation}\label{eq:stoc-dom-mia}
\sP\big[ S_i = 1 \ | \ T^a = t^a, S_{<i} = s_{<i} \big] \leq \frac{e^{c+\epsilon}}{1 + e^{c+\epsilon}} ,
\end{equation}
and we can proceed by induction.
Assume inductively that $W_{m-1} \triangleq \sum_{i=1}^{m-1} T^a_i \cdot S_i$ is stochastically dominated (see Definition 4.8 in \cite{steinke2023privacy}) by $W'_{m-1} \triangleq \sum_{i=1}^{m-1} T^a_i \cdot S'_i$, in which $S' \sim \textrm{Bernoulli}(\frac{e^{c+\epsilon}}{1 + e^{c+\epsilon}})^{m-1}$. 
Setting $W_1 = W'_1 = 0$ makes it true for $m=1$. 
Then, conditioned on $W_{m-1}$ and using Eq. \ref{eq:stoc-dom-mia}, $T^a_m \cdot S_m = T^a_m \cdot \1\{S_m = 1\}$ is stochastically dominated by $T^a_m \cdot \textrm{Bernoulli}(\frac{e^{c+\epsilon}}{1 + e^{c+\epsilon}})$. 
Applying Lemma 4.9 from \cite{steinke2023privacy} shows that $W_{m}$ is stochastically dominated by $W'_{m}$, which proves the induction and implies the proposition's statement.
\end{proof}

\end{hideIfProofs}

% \begin{table} [htbp]
%     \centering
%     \small
%     \begin{tabular}{cccc}
%         \hline
%         ML Model & Training Epoch  & Test Accuracy \\
%         \hline
%         ResNet101 & 29, 99, 300  & 88\%, 90\%, 92\% \\
%         ResNet50 &  60  & 89\%\\
%         Wide ResNet-28-2 & 10, 20, 70  & 72\%, 80\%, 89\% \\
%         Multi-Label CNN & 30, 100  & 85\%, 93\%\\
%         GPT-2 & 25, 50, ..  & Row 5, Cell 4 \\
%         Tabular Classification Model  & 10, 100 & 86\%, 82\% \\
%         \hline
%     \end{tabular}
%     //
%     \caption{Train and Test Metrics for ML Models Audited}
%     \label{table_metrics}
% \end{table}

\section{Experimental Details}
\label{appendix:exp_details}

Hereafter, we provide the details about the datasets and models trained in our experiments.

\subsection{Image data}
\label{subsec:imgexpdetails}

{\bf Target models.} We audit target models with the following architectures: a Multi-Label Convolutional Neural Network (CNN) with four layers \cite{oshea2015introduction}, and the ResNet101 \cite{he2015deep} as well as a Wide ResNet-28-10 ~\cite{zagoruyko2017wideresidualnetworks}, and a Vision Transformer (ViT)-small ~\cite{dosovitskiy2021imageworth16x16words}. 
We also include in our analysis, differentially-private models for ResNet18~\cite{he2015deep} and WideResNet-16-4~\cite{zagoruyko2017wide} models as targets, with $\epsilon = 1, 2, 4, 6, 10, 15, 20$. 
The ResNet-based models are trained on CIFAR10 using $50k$ images \cite{Krizhevsky2009LearningML} of $32$x$32$ resolution. 
For ResNet101 CIFAR10-based classification models, we use a training batch size of 32; for WideResNet-28-10, we use batch size 128. The details for training DP models are given in \ref{appendix:dp_sup_results}. 
The associated test accuracies and epochs are mentioned in Table~\ref{table_metrics}. 
The Multi-Label CNN is trained on $200k$ images of CelebA \cite{liu2015faceattributes} of $128$x$128$ resolution, training batch-size 32, to predict $40$ attributes associated with each image.

{\bf Generator.} For both image datasets, we use StyleGAN2~\cite{karras2020training} to train the generative model $\gG$ from scratch on $\DG$, and produce non-member images. 
For CIFAR10 dataset, we use a $10,000$ out of $50,000$ images from the training data of the target model to train the generative model.
For the CelebA dataset, we select $35,000$ out of $200,000$ images from the training data of the target model to train the generative model.
Generated images will in turn serve as non-members for performing the MIAs. 
Figure~\ref{datasets} shows examples of member and non-member images used in our experiments. 
In the case of CelebA, we also introduce a vanilla CNN as a classifier or filter to distinguish between fake and real images, and remove any poor-quality images that the classifier detects with high confidence. 
The data used to train this classifier was the same data used to train StyleGAN2, which ensures that the generated high-resolution images are of high quality.
%However, our technique does not necessarily depend on this filtering.

\textbf{MIA and Baseline training.}
For the MIA, we follow a loss-based attack approach: \acronym takes as input raw member and non-member data points for training along with the loss values the target model $f$ attributes to these data points.
More precisely, the training set of \acronym is:
\[(D^{tr}_{in}, f(D^{tr}_{in})) \cup  (D^{tr}_{out}, f(D^{tr}_{out}))\]
% \tao{There are missing closed brackets here.} \mishaal{resolved thank you!}
In \S\ref{subsec:dp-measurement}, we discussed the importance of having a tight $\clb$ so that our measure, $\tilde\epsilon$, becomes close to a lower-bound on $\epsilon$-DP, which requires a strong baseline. 
To strengthen our baseline, we introduce the helper model $\helper$, which helps the baseline model $b$ by supplying additional features (\emph{i.e.}, embeddings) that can be viewed as side information about the data distribution.
The motivation is that $\helper$'s features might differ between samples from $\gD$ and $\gD'$, enhancing the performance of the baseline classifier.
This embedding model $\helper$ is similar in design to $f$ (same task and architecture) but is trained on synthetic data that is close in distribution to the real member data distribution.

The helper model has the same classification task and architecture as the target model. To train it, we generate separate sets of training and validation data using our generator. For image data, the synthetic samples do not have a label. We thus train a labeler model, also a classifier with the same task, on the same dataset used to train the generator. We use the labeler model to provide labels for our synthetic samples (training and validation sets above). We train the helper model on the resulting training set, and select hyperparameters on the validation set.

Whether for the baseline or MIA, we use side information models ($\helper$ and $f$, respectively) by concatenating the loss of $\helper(x)$ and $f(x)$ to the final feature representation (more details are provided later) before a last layer of the MIA/Baseline makes the membership the prediction. 
% \hadrien{is it still the case with the loss module?} \mishaal{yes it's the MIA here}
%(see Appendix \ref{app} 
Since we need labels to compute the loss, we label synthetic images with a Wide ResNet-28-2 in the case of CIFAR10, and a Multi-Label CNN of similar architecture as the target model in the case of CelebA labeling. 
For both instances, we used a subset of the data, that was used to train the respective generative models, to train the ``labeler'' classifiers as well. The labeler used to train the helper model for image data has 80.4\% test accuracy on the CIFAR10 real data test set. 
The rationale for this approach is to augment the baseline with a model providing good features (here for generated data) to balance the good features provided to the MIA by $f$ outside of the membership information. 
In practice, the labeled generated data seems enough to provide such good features, despite the fact that the labeler is not extremely accurate.
We studied alternative designs (e.g., a model trained on non-member data, no helper model) in Appendix D.1, Table 5, and the helper model trained on the synthetic data task performs best (while not requiring non-member data, which is a key point).

We use two different modules for each MIA and Baseline training.
More precisely, the first module optimizes image classification using a built-in Pytorch ResNet101 classifier. 
The second module, in the form of a multi-layer perceptron, focuses on classifying member and non-member labels via loss values attributed to these data points by $f$ as input for the loss module of MIA and losses of $e$ to the baseline $b$ respectively. 
We then stack the scores of both image and loss modules into a logistic regression task (as a form of meta-learning) to get the final outputs for member and non-member data points by MIA and baseline $b$. When training the baseline and MIA models, we use a validation set to select hyper-parameters and training stopping time that maximize the lower bounds (effectively maximizing the $c_{lb}$ or $\tilde\epsilon_{lb}$ for the Baseline and MIA respectively) on the validation set. 
The MIA and baseline are trained on $4500$ data samples (half members and half generated non-members). 
The test dataset consists of $10000$ samples, again half members and half generated non-members. 
The actual and final number of members and non-members that ended up in the test set depends on the Bernoulli samples in our auditing game.
We repeat the training process over 5 times, each time independently resampling the training dataset for the MIA and baseline (keeping the train dataset size fixed). We report all our results over a $95\%$ confidence interval and report the standard deviation for our results in \cref{table:audit_values}.
The \textbf{compute resources} used for MIA and Baseline training were mainly running all experiments on cloud-hosted VMs, using 1 v100l GPU, 4 nodes per task, and 32G memory requested for each job on the cloud cluster. The time to run the MIA and baseline attack pipeline was around 10 hours including hyperparameter tuning on $\epsilon_{lb}$ and $c_{lb}$ respectively. 
\begin{figure}[htbp]
  \centering
  \subfigure[Real]{\includegraphics[width=0.18\textwidth]{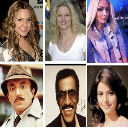}
  \label{subfig:celebAreal}}
  \hfill
  \subfigure[Synthetic]{\includegraphics[width=0.18\textwidth]{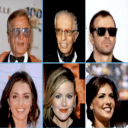}
  \label{subfig:celebFake}}
  \hfill
  \subfigure[Real]{\includegraphics[width=0.18\textwidth]{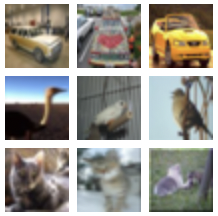}
  \label{subfig:cifarReal}}
  \hfill
  \subfigure[Synthetic]{\includegraphics[width=0.18\textwidth]{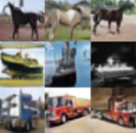}
  \label{subfig:cifarFake}}
  \caption{Member and Non-Member datasets used in our experiments for CelebA at $128$ x $128$ resolution (\ref{subfig:celebAreal}, \ref{subfig:celebFake}) and CIFAR10 (\ref{subfig:cifarReal}, \ref{subfig:cifarFake}) image data at $32$ x $32$ resolution. } 
  \label{datasets}
\end{figure}
%\hadrien{maybe give more details on the model used to remove bad quality sample and to label the synth data} \mishaal{can add; i thought it was added somewhere previously, will check}

\subsection{Language Modeling}

\label{subsec:llmfullexpdetails}
{\bf Target model.}
The target model (a small GPT-2~\cite{radford2019language}) task is causal language modeling (CLM), a self-supervised task in which the model predicts the next word in a given sequence of words (\emph{i.e.}, context), which is done on a subset of WikiText-103 dataset~\cite{merity2016pointer}, a collection of Wikipedia articles. 
While the GPT-2 training dataset is undisclosed, \citet{radford2019language} do state that they did not include any Wikipedia document.

% We partition the WikiText-2 dataset into two equally sized disjoint datasets. 
% One of these datasets serves as the training data for the target model, while the other dataset is used for training the embedding model $e$. 

%Ali: This is not the case anymore in the main result. This is just relevant for the table 5, which I think we should mention the setup while discussing that result

One common standard pre-processing in CLM is to break the tokenized training dataset into chunks of equal sizes to avoid padding of the sequences.
To achieve this, the entire training dataset is split into token sequences of fixed length, specifically 64 tokens per sequence, refer hereafter as ``chunks''. 
Each of these chunks serves as an individual sample within our training dataset. 
Therefore, in the experiments, the membership inference will be performed for a chunk. In other words, we adopt an example-level differential privacy (DP) viewpoint, where an example is considered a sequence of tokens that constitute a row in a batch of samples.
This chunk-based strategy offers additional benefits in synthetic sample generation. 
By maintaining uniformity in the length of these synthetic samples, which also consists of 64 tokens, it is possible to effectively mitigate any distinguishability between synthetic and real samples based solely on their length.
We also consider the possibility of weak correlations between chunks (e.g., being from the same article). To mitigate this, we perform a form of stratified sampling as a pre-processing step. From each document, we include only the first 50 chunks and discard the rest. We select 2,113 documents in total, resulting in 105,650 chunks in our dataset.

We split the dataset into train and test with a 90:10 ratio for the target model (We insist on following the typical training pipelines in machine learning, since we plan to conduct post-hoc audits with \acronym). The test dataset provides real non-member samples, allowing comparison to the O(1) auditor and another version of our approach which uses real non-member data.  The training dataset ($\DT$) contains 95,085 samples, while the test dataset consists of 10,565 samples.

% We include these additional documents to improve the sample complexity of the hypothesis test in the auditing game, thereby increasing the upper bound of the power of the test. 
% In our experiments with the WikiText dataset, the audit set typically involves 10,000 samples. 
% Assuming 5000 of them are members on average (bearing in mind that our test statistic exclusively considers members), theoretically, we can measure up to 7.21 for either $\clb$ or $\cpepslb$. 

We train the target model for 200 epochs in total. 
Figure \ref{fig:gpt-2-target-overtraining-loss} displays how the training and validation cross-entropy loss changes throughout the training, when we audit models from epochs 37 (best generalization), 75, and 150.
To check how well the target model generalizes, we look at the cross-entropy loss (on a validation set), which is the only metric in a causal language model task to report (or the perplexity~\cite{radford2019language} which conveys the same information).

\begin{figure}[!h]
  \centering  
\includegraphics[width=0.8\columnwidth]{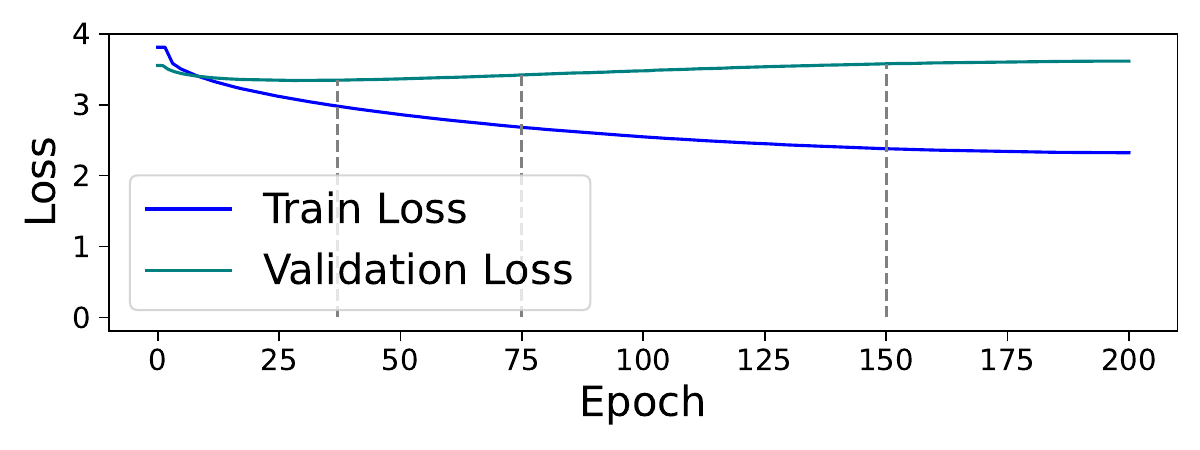}
  \caption{Target model (GPT-2) loss during overtraining on WikiText dataset. We pick GPT-2\_E37, GPT-2\_E75, GPT-2\_E150 as target models to audit corresponding to the gray dashed lines.}
  \label{fig:gpt-2-target-overtraining-loss}
\end{figure}
\vspace{-0.07\baselineskip}

% For the experiment in which we have a target model with canaries, we incorporate an additional 8,500 synthetic samples into the target model and subsequently train a new model. We confirm that the model's utility augmented with synthetic samples (canaries) performs similarly to the target model without canaries. 
% \ali{does it need a figure?} 

The \textbf{compute resources} used for training the target models are 4 V100-32gb GPUs, running on cloud-hosted VMs. 
Using this setup, fine-tuning the target model for 200 epochs required about 13 hours.

{\bf Generator.}
The $\gG$ is a GPT-2 fine-tuned using a CLM task on dataset $\DG$, a subset of $\DT$. 
To create synthetic samples with $\gG$, we use top-$k$ sampling \cite{holtzman2019curious} method in an auto-regressive manner while keeping the generated sequence length fixed at 64 tokens. 
To make $\gG$ generate samples less like the real members it learned from (and hopefully, more like real non-members), we introduce more randomness in the generation process in the following manner.
% We set the top\_p parameter that determines the probability cutoff for the selection of a pool of tokens to sample from to the value of 1.
We choose a top\_k value of 200, controlling the number of available tokens for sampling. 
% after applying top\_p filtering. 
In addition, we balance the value of top\_k, making it not too small to avoid repetitive generated texts but also not too large to maintain the quality.
Finally, we fix the temperature parameter that controls the randomness in the softmax function at one. 
Indeed, a larger value increases the entropy of the distribution over tokens resulting in a more diverse generated text.

However, the quality of synthetic text depends on the prompts used. 
To mitigate this issue, we split the $\DG$ dataset into two parts: $\DG^{train}$ and $\DG^{prompt}$.
$\gG$ is fine-tuned on $\DG^{train}$, while $\DG^{prompt}$ is used for generating prompts. 
Numerically, the size of the $\DG^{train}$ is $35\%$ of the $\DT$ and the size of $\DG^{prompt}$ is $15\%$ of $\DT$.
During generation, we sample from $\DG^{prompt}$, which is a sequence of length 64. 
We feed the prompt into the generator and request the generator to generate a suffix of equivalent length. 
In our experiments, for each prompt, we make 8 synthetic samples, supplying a sufficient number of synthetic non-member samples for different parts of the pipeline of \acronym, including training the helper model, and constructing $\Douttr$, $\Doutte$. 
Overall, we generate 92276 synthetic samples out of which we use 53502 of them for training the helper model and keep the rest for audit purposes.
Note that this generation approach does not inherently favor our audit scheme, and shows how \acronym can leverage existing, public generative models with fine tuning to access high-quality generators at reasonable training cost.

We also perform a sanity check on the quality of the generated data. We visualize the loss distribution of in-distribution members, in-distribution non-members, and synthetic data under both the target and helper models, aiming for synthetic data that behaves similarly to in-distribution non-member data. \cref{fig:gpt_losses} visualizes these distributions for our target and helper models.

The \textbf{compute resources} used for fine-tuning the generator model, and synthetic text generation are 4 V100-32gb GPUs, running on cloud-hosted VMs. 
Using this setup, fine-tuning the GPT2-small model took approximately 1.5 hours, and generating the synthetic samples required about 1 hour.

{\bf Baseline \& MIA.}
For both the baseline and \acronym~classifiers, we employ a GPT-2 based sequence classification model. 
In this setup, GPT-2 extracts features (\emph{i.e.}, hidden vectors) directly from the samples.  
% Furthermore, feature extraction is performed on the samples using the helper model $\helper$ for the baseline classifier and the target model for \acronym. 
We concatenate a vector of features extracted using the helper $\helper$ (for the baseline) or target model $f$ (for the MIA) to this representation. 
This vector of features is the loss sequence produced in a CLM task for the respective model.
These two sets of extracted features are then concatenated before being processed through a feed-forward layer to generate the logit values required for binary classification (distinguishing between members and non-members).

In the main results in \cref{table:audit_values}, for both the baseline and \acronym~classifiers, the training and validation sets consist of 20000 and 2000 samples, respectively, with equal member non-member distribution. 
The test set consists of 10000 samples, in which the actual number of members and non-members that ended up in the test set depends on the Bernoulli samples in the auditing game.
The helper model is fine-tuned for $60$ epochs on synthetic samples and we pick the model that has the lowest validation loss throughout training, generalizing well.

As the GPT-2 component of the classifiers can effectively minimize training loss without achieving strong generalization, regularization is applied to the classifier using weight decay. 
Additionally, the optimization process is broken into two phases. 
In the first phase, we exclusively update the parameters associated with the target (in \acronym) or helper model (in the baseline) to ensure that this classifier has the opportunity to focus on these specific features. 
In the second phase, we optimize the entire model as usual. 
% In addition, we tune the regularization coefficient hyperparameter and the number of epochs for the first phase of the optimization.
We train both the baseline and Membership Inference Attack (MIA) models with 5 different seeds. The randomness is over: 1. The data split into train/validation/test sets for both the baseline and MIA classifiers. Each classifier is retrained for every seed.  2. The coin flips in our auditing game (see \cref{def:auditing-game}). For a fixed seed mentioned above, we employ an additional set of 5 seeds to account for the randomness in the training algorithm (such as initialization of the model, and batch sampling). We then train a classifier for each of these seeds. 
From the 5 models generated for each latter seed, we select the one with the largest $\clb$ value (for the baseline) or $\cpepslb$ (for MIA), determined over the validation set. We calculate the 95\% confidence interval of our measurements using the t-score.

The \textbf{compute resources} used for MIA and Baseline training were 1 GPU with 32gb memory, running on cloud-hosted VMs. The compute time depends on the training set size, for the largest training size, it took about 3 hours to train the baseline or MIA.

The \textbf{compute resources} used for fine-tuning the helper model are 4 V100-32gb GPUs, running on cloud-hosted VMs. Using this setup, it took about 4 hours to fine-tune the helper model.

\subsection{Tabular data}
{\bf Target $f$ and the helper model $\mathbf{\helper.}$}
The Target $f$ and the helper $\helper$ model are both Multi-Layer Perceptron with four hidden layers containing 150, 100, 50, and 10 neurons respectively.
Both models are trained using a learning rate of 0.0005 over 100 training epochs.
In the generalized case and to obtain the embedding, we retain the model's parameters that yield the lowest loss on a validation set, which typically occurs at epoch 10 (MLPE\_10). 
For the overfitted scenario, we keep the model's state after the 100 training epochs (MLP\_E100). The non overfitted model and the overfitted one achieved an accuracy of 86\% and 82\% respectively.

{\bf Generator.}
We use the MST method~\cite{mckenna2021winning} which is a differentialy-private method to generate synthetic data.
However, as we do not need Differential Privacy for data generation we simply set the value for $\varepsilon$ to $1000$.
Synthetic data generators can sometimes produce bad-quality samples. 
Those out-of-distribution samples can affect our audit process (the detection between real and synthetic data being due to bad quality samples rather than privacy leakage). 
To circumvent this issue, we train an additional classifier to distinguish between real data from $D_G$ and additional synthetic data (not used in the audit). 
We use this classifier to remove from the audit data synthetic data synthetic samples predicted as synthetic with high confidence.

{\bf Baseline and MIA.}
To distinguish between real and synthetic data, we use the Gradient Boosting model and conduct a grid search to find the best hyperparameters.

% For both the baseline and \acronym~classifiers, the training and validation set each comprises 1000 samples (500 samples per class) while the test set consists of 2000 samples. 
% More precisely, the test set includes a total of 1000 negatives (\emph{i.e.}, non-members), which sets the lower limit for meaningful FPR at $10^-3$. 
% During training, the validation set is used to select the model that achieves the best FPR. 
% Given our specific interest in the low FPR regime, the FPR values are computed within the interval $[10^{-3}, 10^{-1}]$.

\label{appendix:O1_exp_details}
\subsection{Comparison with Privacy Auditing with One (1) Training Run: Experimental Details}
% - exp details: model, data split for RM; RN; test accuracy, test set, params etc
We implement the black-box auditor version of O(1) approach \cite{steinke2023privacy}. 
This method assigns a membership score to a sample based on its loss value under the target model. 
They also subtract the sample's loss under the initial state (or generally, a randomly initialized model) of the target model, helping to distinguish members from non-members even more. 
In our instantiation of the O(1) approach, we only consider the loss of samples on the final state of the target model. 
Moreover, in their audit, they choose not to guess the membership of every sample. 
This abstention has an advantage over making wrong predictions as it does not increase their baseline.  
Roughly speaking, their baseline is the total number of correct guesses achieved by employing a randomized response $(\epsilon, 0)$ mechanism, for those samples that O(1) auditor opts to predict. 
We incorporate this abstention approach in our implementation by using two thresholds, $t_{+}$ and $t_{-}$. 
More precisely, samples with scores below $t_{+}$ are predicted as members, those above $t_{-}$ as non-members, and the rest are abstained from prediction. 
We check for different combinations of $t_{+}$ and $t_{-}$ and report the highest $\epsilon$ among them. This involves performing multiple tests on the same evidence, each with a separate confidence level of 95\%. To ensure that these tests hold collectively, we apply a union bound and adjust the significance level accordingly. The total number of separate tests depends on the different number of guesses that O(1) can make.
We also set $\delta$ to 0 and use a confidence interval of 0.05 in their test.  
In \acronym, for each hypothesis test (whether for $\clb$ or $\cpepslb$), we stick to a 0.025 confidence interval for each one, adding up to an overall confidence level of 0.05. 
% Furthermore, the audit set of O(1) is the same as the audit set of \acronym.

%\subsection{Comparison with MIA on RM; RN: Experimental details on tabular}
%\label{appendix:RMRN}

%In order to proceed with the comparison of \acronym with a variant of \acronym in which we use real data instead of synthetic data, we need a significant number of real non-member observations. 
%Due to the limited size of the adult dataset ($\approx 30k$ observations), we train the target model $f$ on only half of the dataset and use the remaining half to train and evaluate the MIA.
\section{Results: Additional Experiments and Detailed Discussion}

\subsection{Baseline Strength Evaluation on Text and Tabular datasets}
\label{subsec:baseline_moreeval}

\begin{table} [h]
\centering
 \resizebox{0.5\textwidth}{!}{
\begin{tabular}{|l|l|}
\hline
Baseline model                               & \(c_{lb}\) \\
                    \hline
                    CIFAR-10 Baseline\textsubscript{\(D_{h}^{\text{tr}}=\text{gen, WRN}\)}       & $\mathbf{2.44 \pm 0.19 }$ \\
                    CIFAR-10 Baseline\textsubscript{\(D_{h}^{\text{tr}}=\text{real, WRN}\)}      & $2.21 \pm 0.17$ \\
                   CIFAR-10 Baseline\textsubscript{\(D_{h}^{\text{tr}}=\text{real, resnet101}\)} & $2.01 \pm 0.15$  \\
CIFAR-10 Baseline\textsubscript{\(D_{h}^{\text{tr}}=\text{imgnet}\)}          & $1.12 \pm 0.19$ \\
CIFAR-10 Baseline\textsubscript{no helper}                                    & $1.25 \pm 0.24$ \\

\hline
WikiText Baseline\textsubscript{\(D_{h}^{\text{tr}}=\text{gen}\)}      & \bm{$3.31 \pm 0.15$} \\
WikiText Baseline\textsubscript{\(D_{h}^{\text{tr}}=\text{real}\)}      & $3.26 \pm 0.14$ \\
WikiText Baseline\textsubscript{no helper}                                    & $3.11 \pm 0.15$ \\
\hline
% Adult Baseline\textsubscript{\(D_{h}^{\text{tr}}=\text{gen}\)}      & \textbf{2.34} \\
% Adult Baseline\textsubscript{\(D_{h}^{\text{tr}}=\text{real}\)}      & 2.18 \\
% Adult Baseline\textsubscript{no helper}                                    & 2.01 \\
% \hline
\end{tabular}}
\label{table:baseline_eval}
\vspace{0.5cm}
\caption{Baseline evaluation with different helper model scenarios, where WRN is Wide-ResNet-28-2 helper model architecture (in the case of CIFAR10 baseline).}
% \vspace{-10pt}
\label{table:baseline_eval-detailed}
\end{table}

The basis of our privacy measurement directly depends on how well our baseline classifier distinguishes between real member data samples and synthetic non-member data samples generated by our generative model. 
As mentioned in \cref{subsec:baseline} and \cref{subsec:imgexpdetails}, in order to increase the performance of our baseline $b$, we mimic the role of the target model $f$'s loss in the MIA using a helper model $\helper$, which adds a loss-based feature to $b$. 
To assess the strength of the baseline classifier, we train the baseline's helper model $\helper$ on two datasets, synthetic data, and real non-members. 
We also train a baseline without any helper model. 
This enables us to identify which setup allows the baseline classifier to distinguish the member and non-member data the best in terms of the highest $c_{lb}$ value ($c_{lb}$ is reported for each case in table \ref{table:baseline_eval-detailed}). 
For initial analysis, we keep a sub-set of real data as non-members to train a helper model with them, for conducting this ablation study on the baseline. 
% \ali{Here, it should be more clear what are these various datasets. Also, it should be more clear that only for the sake of this comparison, we keep some real samples as non-members, not anywhere else in the experiments.} \mishaal{done}
% The results in table \ref{table:baseline_eval-detailed} show the comparison of baselines trained with different helper models, and no helper model, in each data modality.
% trained on real non-member data of the same distribution as the respective target ML model, generated data produced by our generative models in each respective data modality setting. 
Our experiments confirm that the baseline with loss values from a helper model trained on synthetic data performs better compared to a helper model trained on real non-member data. 
% Hence, it removes the dependency of using real non-members for the baseline.
This confirms the strength of our framework in terms of not having a dependency of using real non-members. Therefore, we choose this setting for our baseline helper model.
% , owed to giving the highest $c_{lb}$. 
In addition, we also experiment with different helper model architectures in the case of CIFAR10, looking for the best baseline performance among them. 
For the WikiText dataset, the $\clb$ reported in \cref{table:baseline_eval-detailed} differs from that in \cref{table:audit_values} because we had less data available for this experiment. We withheld half of the real samples to train the helper model with real samples in its training set.
% as well as no helper model for all data modalities. 

For each data modality setting, we also train the baseline classifier on an increasing dataset size of real member and synthetic non-member data points to gauge the trend of the baseline performance as the training dataset size increases (see Figure~\ref{fig:baseline_eval_figure} and Figure~\ref{fig:baseline-other-modalities-eval}).
% \ali{We need a conclusion or takeaway sentence here to conclude that a helper trained with synthetic data is sufficient, no need for real non-members.} \mishaal{done}
This experiment 
% works two-fold: it 
allows us to aim for a worst-case empirical lower bound 
% \ali{Isn't it a strong claim to say it's tight?} \mishaal{tru, I removed tight, how does it look now?} 
on how close the generated non-member data is to the real member data distribution, \emph{i.e.} what would be the largest $\clb$ we can reach as we increase the training set size. 
Subsequently, we can compare the $\clb$ reported in table \ref{table:audit_values} (corresponding to the vertical dashed lines in Figure~\ref{fig:baseline_eval_figure} and Figure~\ref{fig:baseline-other-modalities-eval}) to the largest $\clb$ we can achieve in this experiment. 
In the following, we explain why we chose that number of training samples for our baseline in the main results.
% secondly, it allows us to understand 
% it grants us confidence in the strength of the baseline performance at being able to classify synthetic data from real data, 
% allowing anything above that as detected by the MIA, i.e $\{\epsilon+c\}_{lb}$ owed to privacy leakage.

On the WikiText dataset, 
we increase the baseline's training set size up to 50k, while fixing the validation and test sets sizes to 2000 and 10000, respectively). \cref{fig:baseline-other-modalities-eval} shows the changes of $\clb$ while varying the training set size. The vertical dashed line specified the number of training samples we chose for our main auditing results reported in \cref{table:audit_values}. 

Now, let us discuss why we don't choose the largest training set size available for the baseline. The measurement for $\clb$ depends not only on the training set size but also on the test set size.
Therefore, we need to balance the sizes of the train and test sets. We conducted another experiment where we varied the train set size and, instead of fixing the test set size, allocated the remaining available data to the test set (keeping in mind that our data for the audit is limited by the number of real members allocated for audit purposes).\cref{subfig:nlp_vary_train_test} depicts the trend of changes. 
As expected, we observed an initial increase in $\clb$ as the training set size increased (and the number of test data is sufficient to capture that). However, it decreased toward the end because of the smaller test set size.
For a train set size of 40k, we found the best trade-off between train and test set sizes, which justifies our choice of 40k for the training set size in the main results. The same argument applies for other data modalities.

Let us make one last point. It is also important to note that the MIA is expected to show similar behavior if we use a larger train set size or test set size. Hence, the gap between the performance of MIA and baseline (which is our privacy measurement) is not expected to be affected (see \cref{subsec:increasing-train-size-mia}).
% This modification accounts for the mismatch between the $\clb$ at the vertical dashed line in Figure~\ref{fig:baseline-other-modalities-eval} and the reported one in Table~\ref{table:audit_values}. 
% Nevertheless, the maximum average $\clb$ attained on our largest training set for the baseline (represented by the horizontal dashed line in Figure~\ref{fig:baseline-other-modalities-eval}) is 3.66, which is still lower than 3.78, the one reported in Table~\ref{table:audit_values}.

% Due to dataset size constraints in practicality, in terms of how much real member data is available, the train dataset size for our MIA and baseline where the performance of the baseline in terms of $c_{lb}$ is the highest. 
% This accounts for how much data we have in terms of member data for the target model, and the split we account for while training the respective generative models, as well as the real members test data used to evaluate the MIA and baseline. 
% These data constraints lead us to choose the train dataset size as shown by the vertical line in Figure~\ref{fig:baseline_eval_figure} and Figure~\ref{fig:baseline-other-modalities-eval}. 
% Hence the final values reported in Table~\ref{table:baseline_eval-new} are on these corresponding data sizes. 

\begin{figure}[h!]
  \centering
  \subfigure[$c_{lb}$ as the size of the training set increases for baseline (while the test set size is 10k), trained to classify real (WikiText) and synthetic data points. 
  The vertical dashed line indicates the number of samples we used to train the baseline in our main results, reported in Table~\ref{table:audit_values}.]
     {\includegraphics[width=0.45\linewidth]{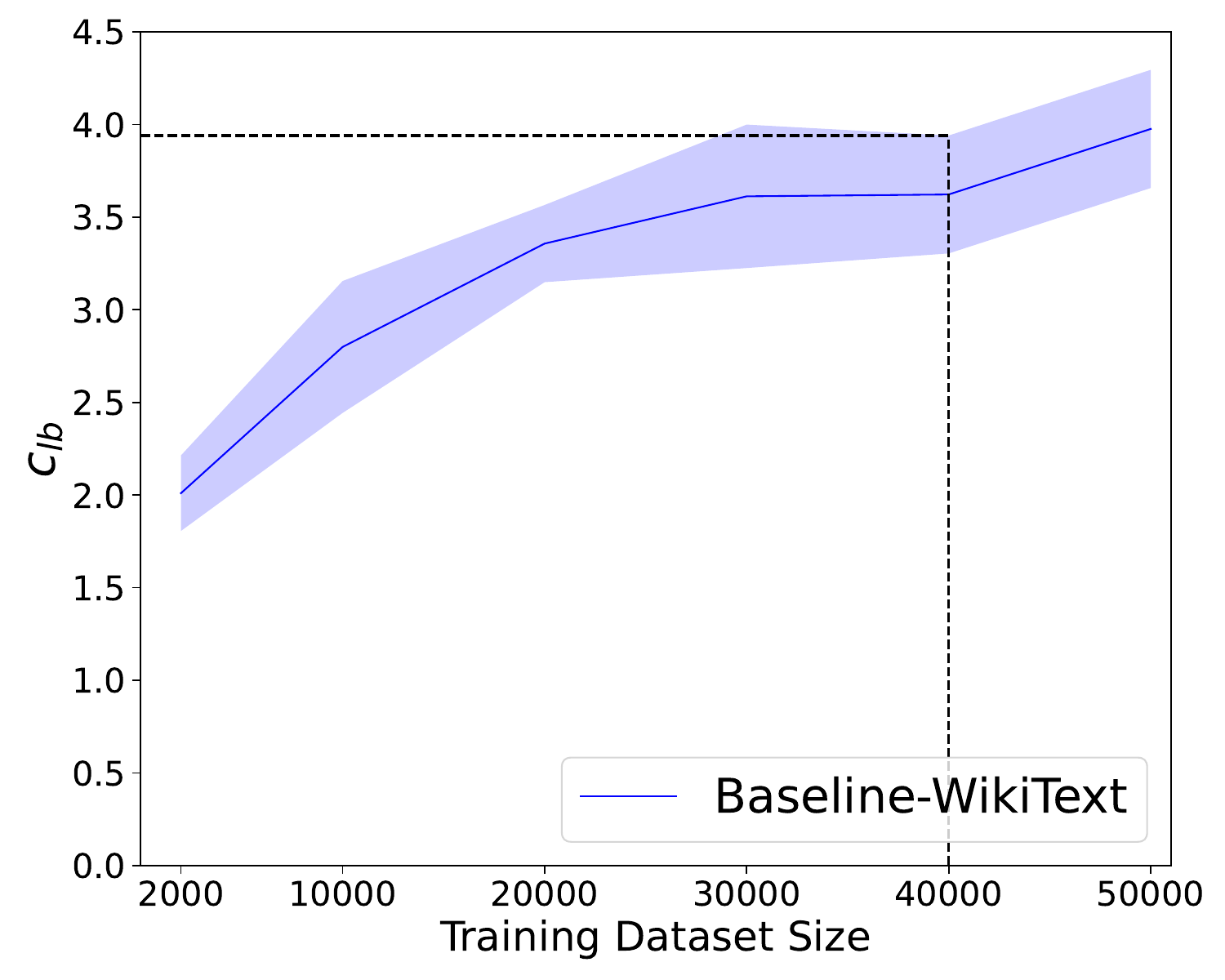}
     \label{subfig:nlp_plateau}
     }
  \hfill
  \subfigure[$c_{lb}$ as the size of the training set increases for baseline, while the test set size is changing at the same time. The number of test samples for a given train set size is specified by the total number of samples we have for the audit ($\Din$), which is limited by the number of real members.]{
  \includegraphics[width=0.45\linewidth]{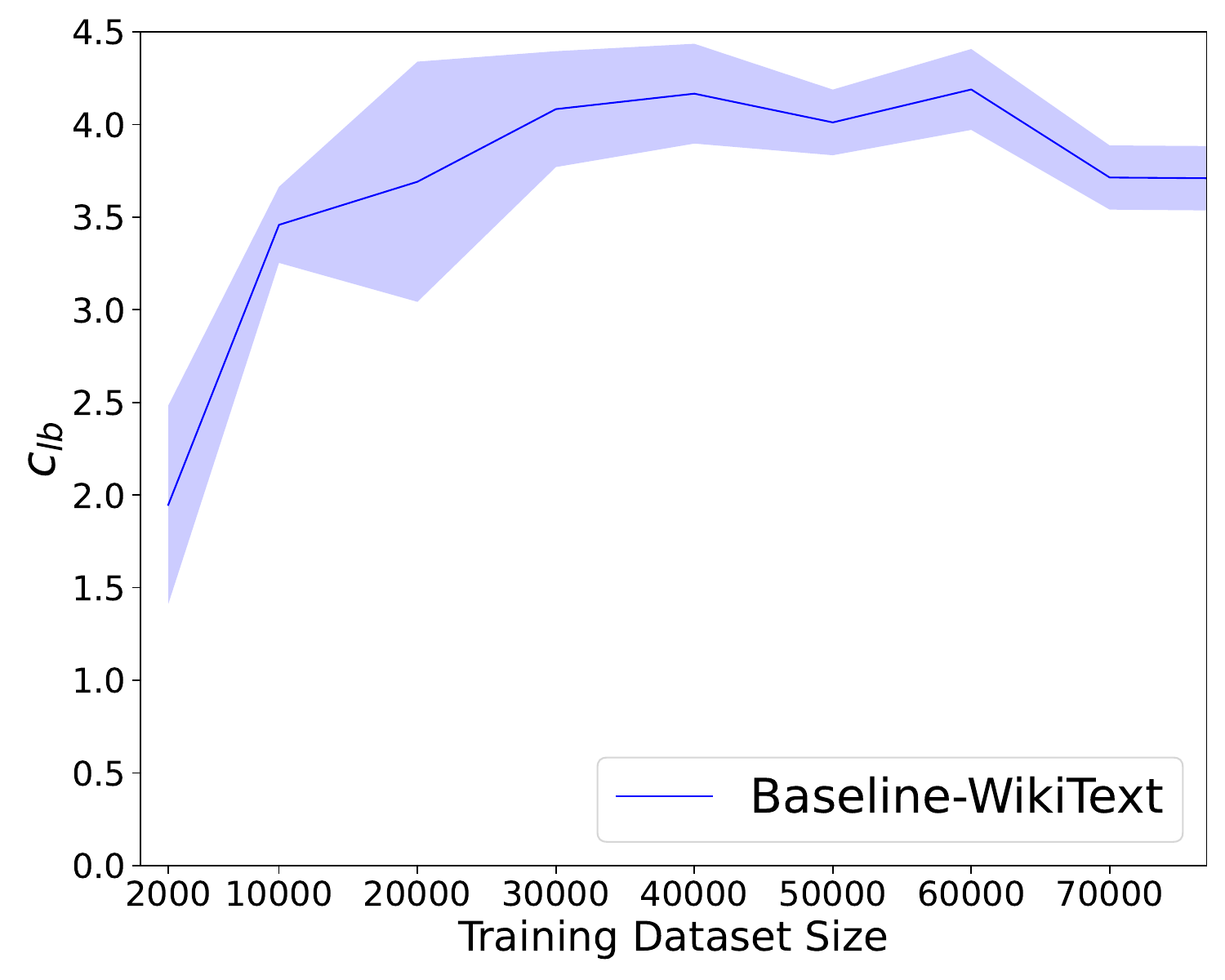}
     \label{subfig:nlp_vary_train_test}
  }

%subfigure[$c_{lb}$ as training data increases for baseline trained to distinguish real and synthetic Adult tabular  data points. 
%  The vertical dashed line indicates the number of samples we used to train the baseline in our main results (30k), reported in table \ref{table:audit_values} while the horizontal dashed line shows the maximum average $\clb$ at 2.48.]{\includegraphics[width=0.45\linewidth]{images/tab_c_lb_trainingsize_full_dashlines.pdf}}
  
 \caption{Baseline Strength Evaluation on WikiText dataset as the train set size increases. In \cref{subfig:nlp_plateau}, the test set size is fixed, while in \cref{subfig:nlp_vary_train_test} it varies as well. See \cref{subsec:baseline_moreeval} for more details.} 
  \label{fig:baseline-other-modalities-eval}
\end{figure}

The experiments are repeated five times, and we provide a 95\% confidence interval. For the image data modality, the randomness lies in the training data provided to the classifiers, while the test set remains fixed throughout. For the text data modality, both the train and test sets change across the five runs (the test set size is fixed to 10k).

\begin{figure*}[ht]
  \centering

    \subfigure[Number of predictions of ResNet101 on CIFAR10]{\includegraphics[width=0.23\textwidth]{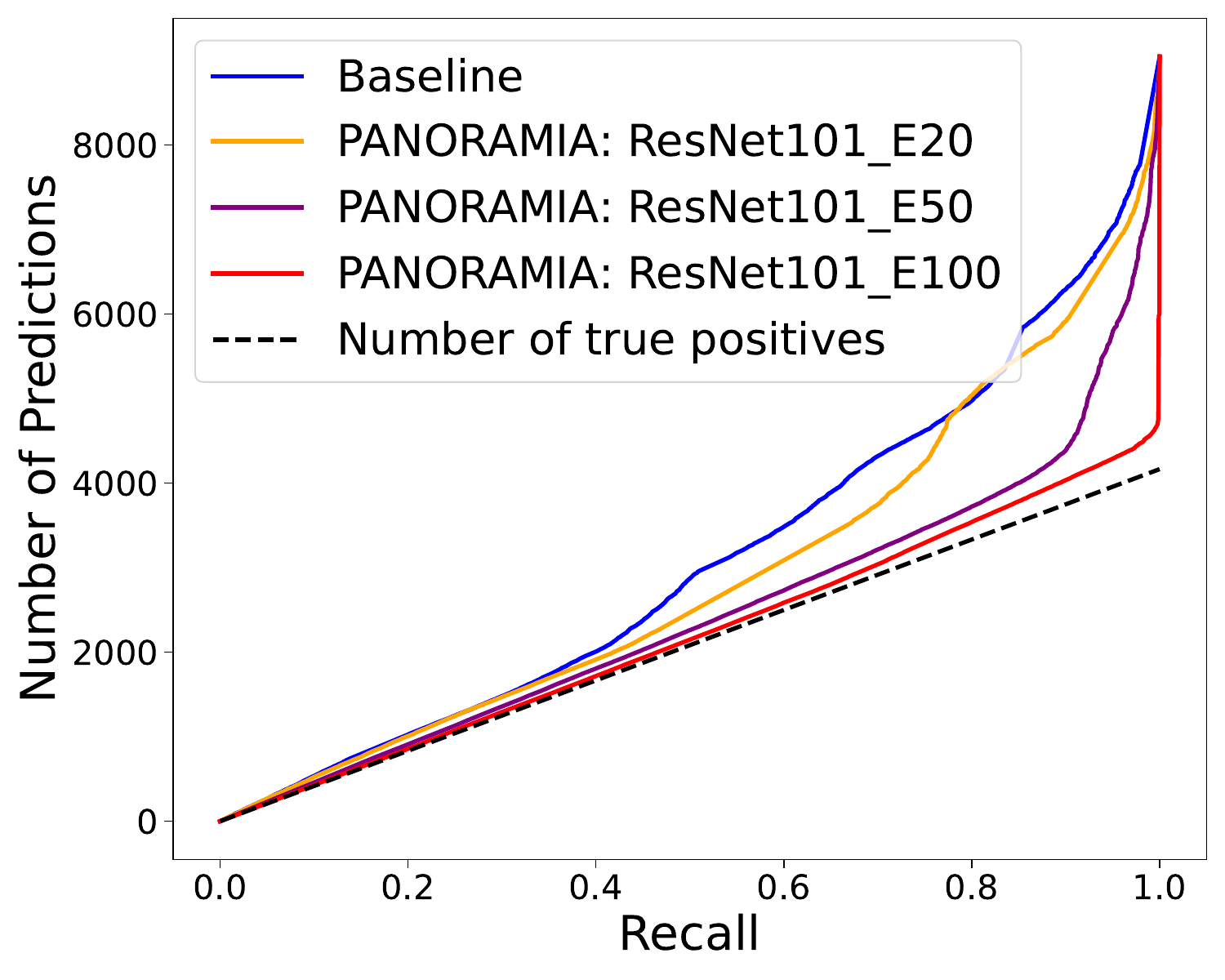}
    \label{fig:image3}
    }
    \subfigure[Number of predictions of Multi-Label CNN on the CelebA]{\includegraphics[width=0.23\textwidth]{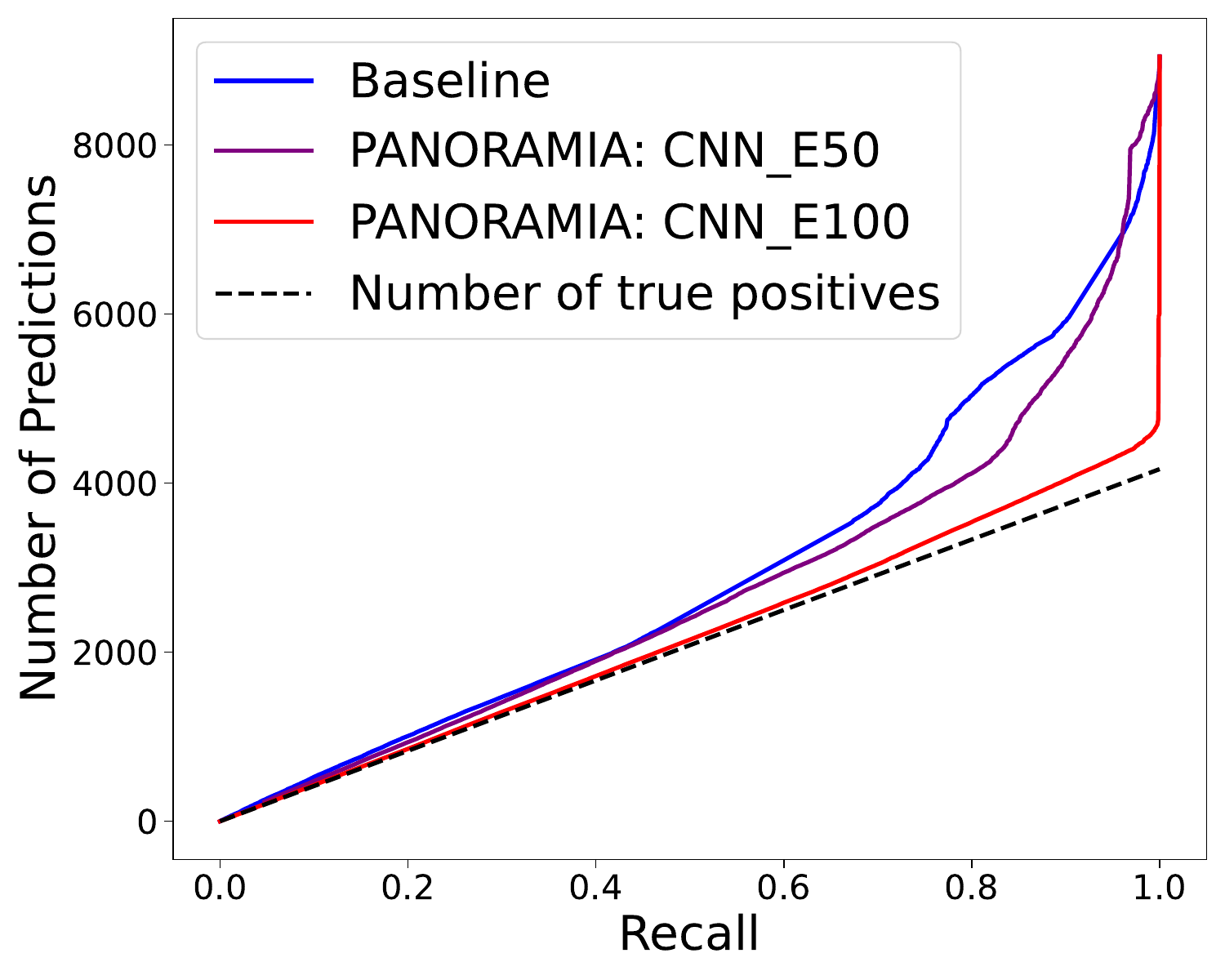}
    \label{fig:image3}
    }
    \subfigure[Number of predictions of GPT-2 model on WikiText]{\includegraphics[width=0.23\textwidth]{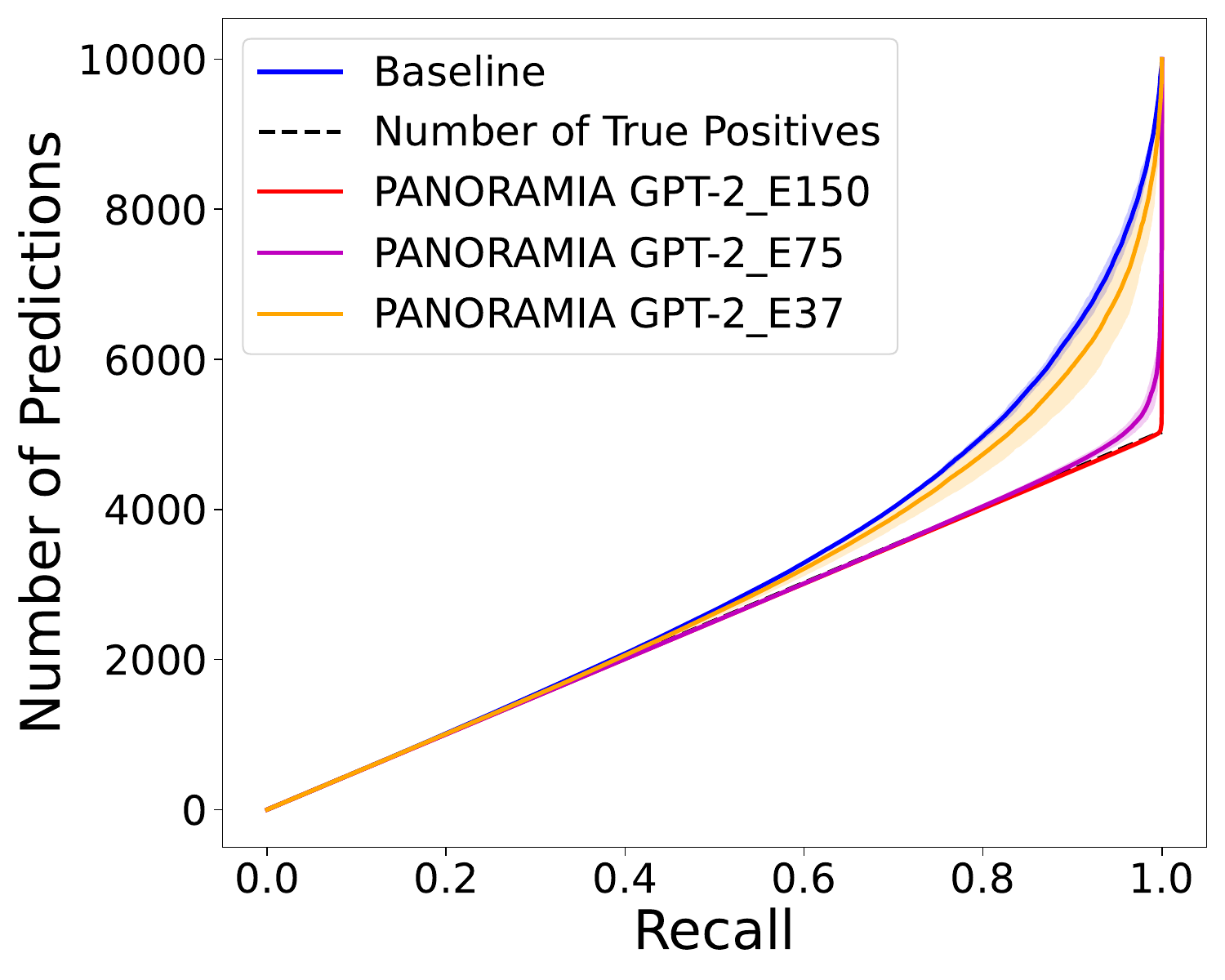}
    \label{fig:gpt2_numpred_recall}
    }
    \subfigure[Number of predictions of an MLP on the Adult dataset.]{\includegraphics[width=0.23\textwidth]{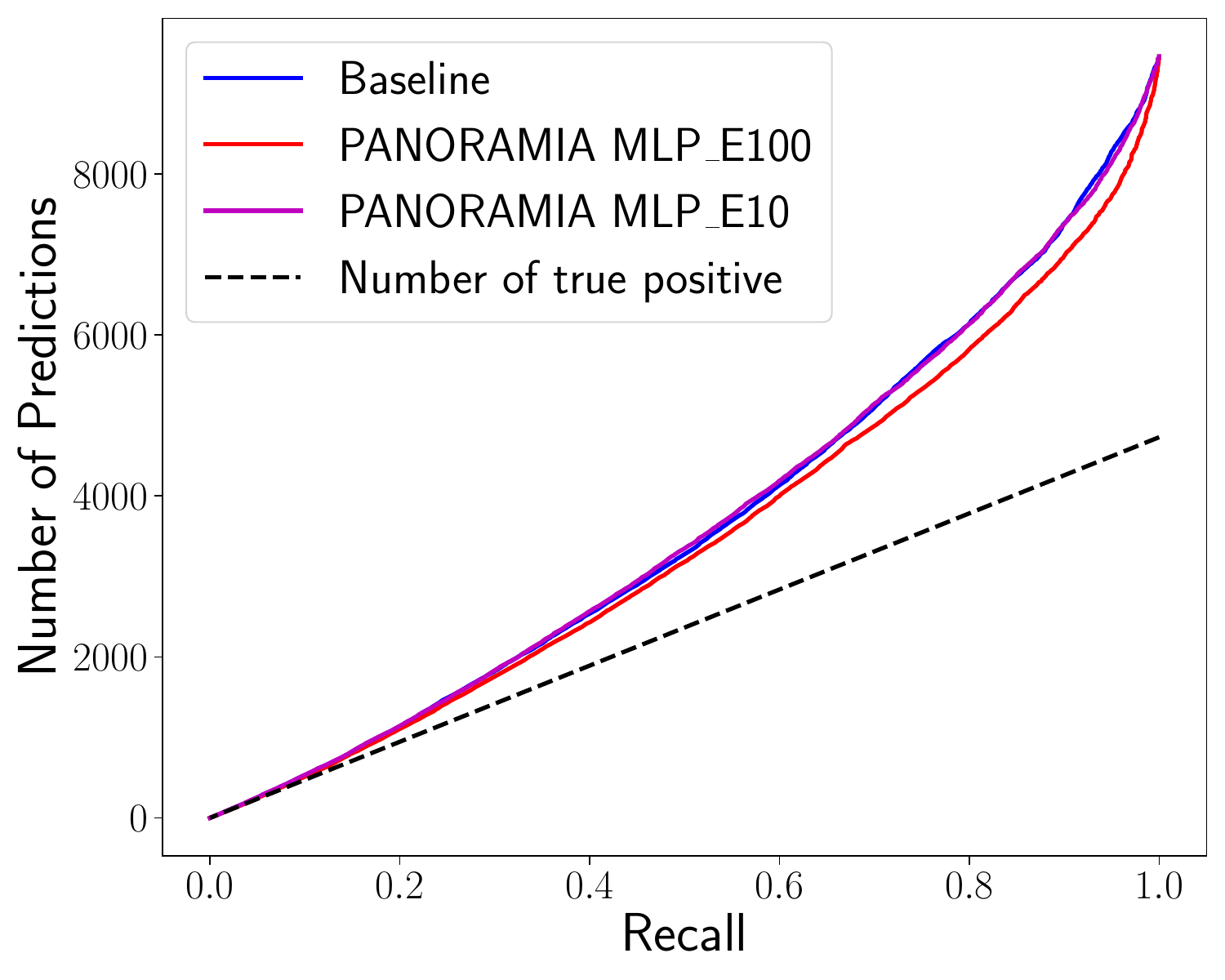}
    \label{fig:image2}
    }
    
  \caption{Comparison of the number of true positives and predictions on different datasets}
  \label{fig:preds_comparison}

\end{figure*}

\subsection{Further Analysis on General Privacy Measurement Results}
\label{subsec:general-results}

The value of $\tilde\epsilon$ for each target model is the gap between its corresponding dashed line and the baseline one in Figure~\ref{fig:eval:comparison_privacy}. 
This allows us to compute values of $\tilde\epsilon$ reported in Table~\ref{table:audit_values}.
%\hadrien{add a few sentences to describe this table when finished}
%For all experiments, we notice that more training epochs on the audited model lead to greater values of $\{c+\epsilon\}_\textrm{lb}$ Our experiment is coherent with the intuition that more training epochs lead to more over-fitting which in turn leads to more privacy leakage. 
It is also interesting to note that the maximum value of $\cpepslb$ typically occurs at low recall. 
Even if we do not use the same metric (precision-recall as opposed to TPR at low FPR) this is coherent with the findings from~\cite{carlini2022membership}. 
We can detect more privacy leakage (higher $\cpepslb$ which leads to higher $\tilde\epsilon$) when making a few confident guesses about membership rather than trying to maximize the number of members found (\emph{i.e.}, recall). Figure \ref{fig:preds_comparison}, decomposes the precision in terms of the number of true positives and the number of predictions for a direct mapping to the propositions~\ref{prop:gen-test} and~\ref{prop:dp-test}. 
%Table \ref{table:audit_values} collects the value of $\cpepslb$, as well as $\tilde\epsilon = \cpepslb - \clb$ measured by \acronym. We can see that $\tilde\epsilon$ reaches up to $0.962$ on CIFAR10 (ResNet101\_E100), $1.02$ on CelebA (Multi-Label CNN-E100), $1.28$ on WikiText-2 (GPT-2\_E96), and $0.34$ on Adult (MLP\_E100). 
These are non-negligible values of privacy leakage, even though the true value is likely much higher.

\subsection{\acronym performance with increasing Test and Train Size $\mathbf{m}$}
\label{subsec:increasing-m-size}
\textbf{\acronym: Increasing Test Set Size:} We compare \acronym with the O(1) audit \citep{steinke2023privacy} (in input space, without canaries), which uses a loss threshold MIA. As shown in \cref{subfig:main-test-inc-resnet100} and \cref{subfig:panoramia_vs_O1_m_dp_increase} in the CIFAR10 image dataset case, both for non-private and Differentially Private (DP) models, \acronym can leverage much more data for its measurement (up to the whole training set size for training and testing the MIA), while in our setting O(1) is limited by the size of the test set (for non-members). We vary test size from $m=500$, up to $35k$. The experiment is repeated $10$ times with resampling of the test set each time to produce the average. Unless otherwise stated, all results are produced at a $95\%$ confidence.

\subsection{\acronym with Increasing Train Dataset Size}
\label{subsec:increasing-train-size-mia}
We also conduct experiments to show the effect of increasing training dataset size on the performance of \acronym (as well as the baseline). In the CIFAR10 image dataset case, \cref{fig:inc-train-size}, we see \acronym performance slowly increase as the train dataset size increases. It reaches its max performance at dataset sizes $4500-5000, 10k, 20k-30k$ and slowly begins to stabilize at the higher dataset sizes. We vary train size from $m=500$, up to $30k$. The experiment is repeated $5$ times with resampling of the train set each time to produce the average. Unless otherwise stated, all results are produced at a $95\%$ confidence.

\begin{figure}[htb]
  \centering
 {\includegraphics[width=0.55\textwidth]{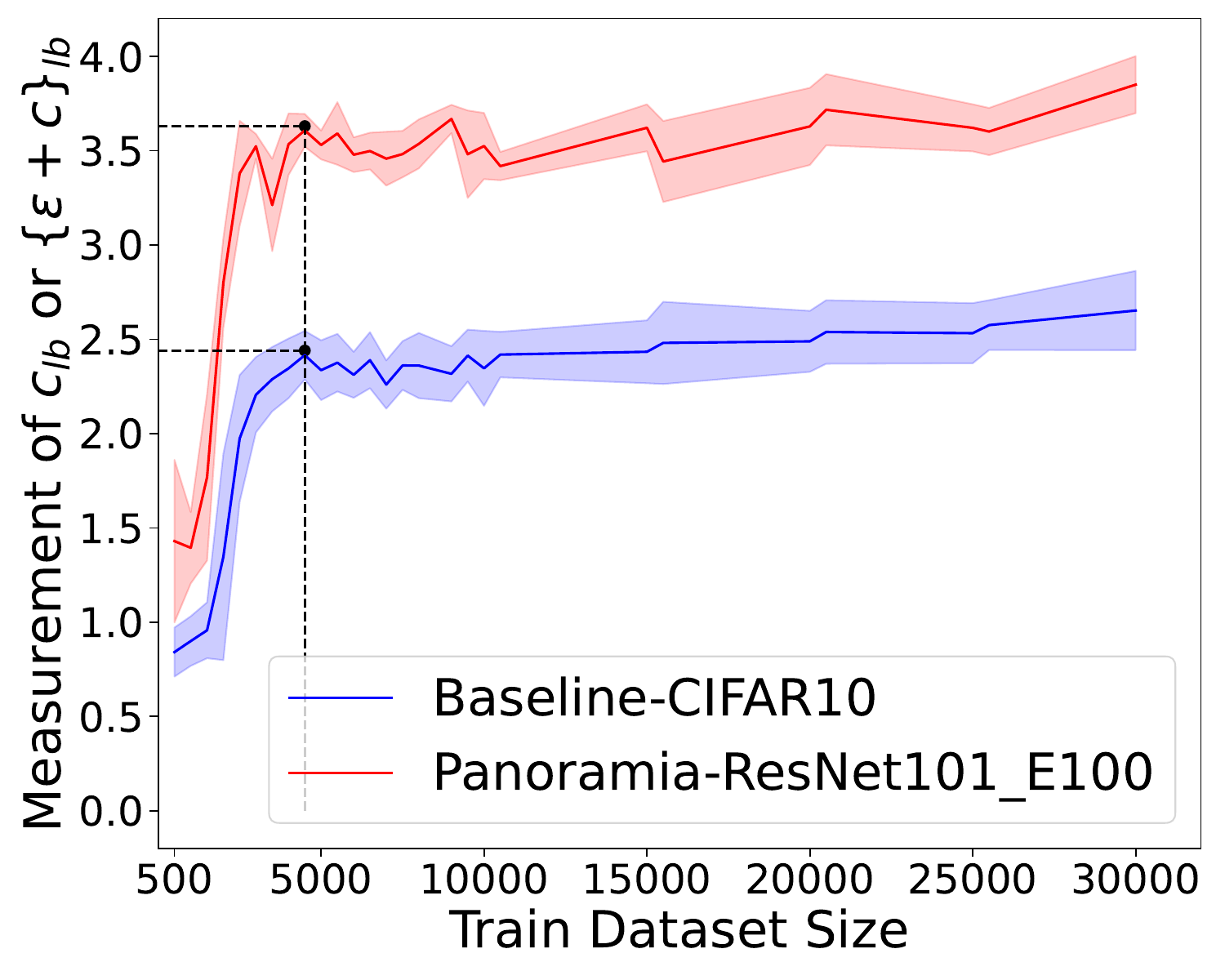} \label{subfig:var_complexities}}

 \caption{Effect on \acronym and baseline with increasing train set size for CIFAR10, ResNet101 target model (trained for 100 epochs) (with union bound correction).
}
 % returned by \acronym on different privacy audit applications. Figure \ref{subfig:resnet18-dp-fig}, shows auditing in a DP setting, whereas figure \ref{subfig:var_complexities} shows audit as model parameter size and depth increases.}
  \label{fig:inc-train-size}
\end{figure}
\subsection{\acronym in cases of no privacy leakage}
\label{subsec:no-priv-leak}
% \mishaal{is this relevant now?} \\
% There are also cases in which \acronym is not able to meaningfully measure privacy loss. \mishaal{it could be just zero, if we verify from O(1) and that there is no privacy leakage to detect}
% \ali{This paragraph doesn't fit well with the new results. Will Edit.} By taking a closer look at the precision-recall curve of both the baseline (blue) and \acronym's MIA for GPT-2\_E23 (orange)  on figure \ref{fig:gpt2_precision_recall}  we notice that the precision-recall curve of the MIA is mostly above the baseline. Moreover, the MIA retains a high degree of precision compared to the baseline whose precision decreases when the recall increases. However, both the baseline and GPT-2\_E25 have similar maximum precision values. As a result, we can see in Figure \ref{fig:gpt2_privacy_recall}  that the maximum value for $\cpepslb$ is close to $\clb$ leading to a small $\tilde \epsilon$ value even if the MIA outperforms the baseline. This is due to our independent selection of recall thresholds, which leads to similar maximum values of $\clb$ and $\cpepslb$.
% Ali: the above is not relevant anymore, at least for text results.

% \mishaal{new NLP results will determine how the storyline of the below para needs to be updated.} \ali{Will edit. It needs to be re-written I think wrt the new results and knowing the O(1) epsilon}

As illustrated in Figure \ref{fig:gpt2_privacy_recall} and \cref{table:audit_values}, there exist instances in which the $\clb$ is as large as $\cpepslb$. 
For instance, for GPT-2\_E37, the upper bound in the confidence interval measured for $\clb$ is larger than the lower bound of the measurement for $\cpepslb$. Although the measured $\tilde \epsilon$ is positive on average, but the lower bound of the interval is negative. For such cases, we don't claim any detection of privacy leakage. 
This implies that, under certain conditions, the baseline could outperform the MIA in \acronym. 
We can see three possible reasons for this behavior.  
(1) The synthetic data might not be of high enough quality. If there is an evident distinction between real and synthetic samples, the baseline cannot be fooled reducing the power of our audit. 
(2) The MIA in \acronym is not powerful enough. 
(3) The target model's privacy leakage is indeed low. 
For GPT-2\_E150 and GPT-2\_E75 in Figure~\ref{fig:gpt2_privacy_recall}, we can claim that given the target model's leakage level, the synthetic data and the MIA can perform the audit to a non-negligible power. 
For GPT-2\_37, we need to examine the possible reasons more closely.
%By fixing the model leakage, we can study (to some extent) how much room there is to improve the fake data quality optimally. 
% Figure \ref{fig:losses} presents the loss distribution for real member samples (members of the target model), real non-members, and synthetic (non-member) samples under the target model and the embedding model. 
Focusing on generated data, in the ideal case we would like the synthetic data to behave just like real non-member data on Figure~\ref{subfig:lossdist-gpt-2E12}. 
If this was the case, $\cpepslb$ would go up (since the MIA uses this loss to separate members and non-members).
Nonetheless, predicting changes in $\clb$ remains challenging, as analogous loss distributions do not guarantee indistinguishability between data points themselves (also, we cannot necessarily predict how the loss distribution changes under the helper model $\helper$ in Figure~\ref{subfig:lossdist-h}).
Hence, reason (1) which corresponds to low generative data quality, seems to be a factor in the under-performance of \acronym.
%Assuming $\clb$ doesn’t increase, for GPT-2\_E25 and GPT-2\_12, we state that the synthetic samples quality is not optimal to perform an audit. 
Moreover, in the case GPT-2\_37,  we observe that real members and non-members have very similar loss distributions.
Hence, factor (2) also seems to be at play, and a stronger MIA that does not rely exclusively on the plot may increase auditing power. 
Finally, the helper model $\helper$ seems to have a more distinct loss distribution when compared to GPT-2\_37 (see \cref{subfig:lossdist-h}). 
% If the MIA had access to both target and helper models, we would expect $\cpepslb$ to at least match $\clb$.

% \ali{it’s kind of mentioning future work in the middle of the analysis} 
% \ali{should I mention that so our mia can be not strong enough, so other mia attack are worth trying -> again, future work in the middle of analysis.} 

%Finally, similar to what we did for the baseline, we plot on Figure \ref{fig:comparison_precision} the maximum theoretical precision given the maximum values of $\cpepslb$.

%Table \ref{table:audit_values} reports the obtained audit value. We obtain $\tilde \epsilon$ values ranging from 0.07 to 2.08. We notice that the nlp modality is where we have both the highest and lowest $\tilde \epsilon$ value. 
%% description of figure 3 (Mishaal)

\subsection{\acronym with Tabular data}
\label{appendix:tabular-results}

\begin{table}[h!]
\centering
\begin{tabular}{|l|l|l|}
\hline
 & \acronym $\tilde{\varepsilon}$ & O(1) $\varepsilon$ \\
\hline
MLP\_E10 & 0 $\pm$ 0.142 & 0.655 $\pm$ 0.014 \\
MLP\_E100 & 0.062 $\pm$ 0.156 & 0.648 $\pm$ 0.015 \\
\hline
\end{tabular}
\vspace{0.3cm}
\caption{Privacy measurement with \acronym and O(1) for models trained on tabular Adult dataset.}
\label{table:results_tab}
\end{table}

As we can see in Table \ref{table:results_tab}, our proposed method \acronym cannot detect any meaningful privacy leakage in the case of a small MLP classifier trained on tabular data.

As shown in Figure \ref{fig:privacy_recall_tab} there is no significant difference between the values$(c+\varepsilon)_{lb}$  $c_{lb}$ returned by $\acronym$. As a result we get $\tilde{\varepsilon} =0$

%First, by taking a closer look at the precision-recall curve of both the baseline (blue) and \acronym’s MIA for MLP\_E100 on figure \ref{fig:fail_precision}, we notice that the precision-recall curve of the MIA is mostly above the baseline.
%However, the baseline has higher maximum precision values at extremely low recall. 
%As a result, we can see in Figure~\ref{fig:fail_privacy} that the maximum value for $(c+\varepsilon)_{lb}$ is lower than $c_{lb}$ leading to failure in measuring privacy leakage.
%This is due to our independent selection of recall thresholds, choosing heuristically the highest value.
\begin{figure} [H]
  \centering  
\subfigure[precision vs recall, for tabular one instance of the tabular data classification audit.]{
\includegraphics[width=0.45\linewidth]{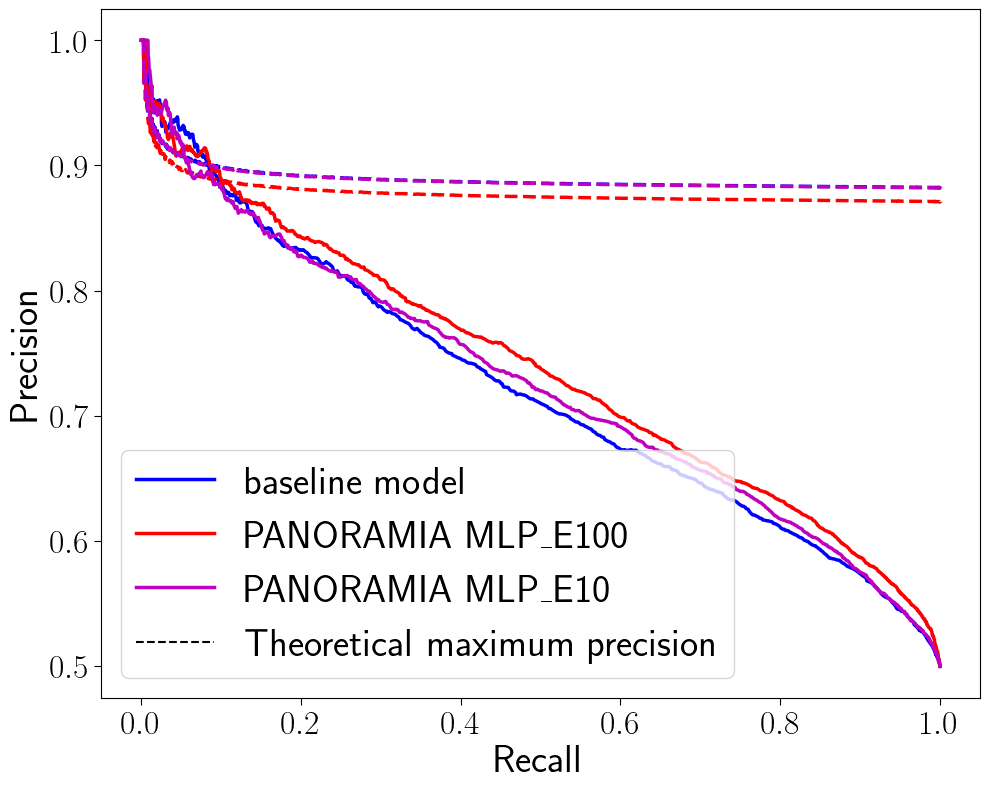}
\label{fig:precision_recall_tab}
}
\subfigure[Average $\{c+\epsilon\}_{\textnormal{lb}}$ (or $\clb$) vs recall with 95\% confidence, for tabular data classification. Randomness comes from resampling of the MIA's train and test set as well as retraining the MIA]{
\includegraphics[width=0.45\linewidth]{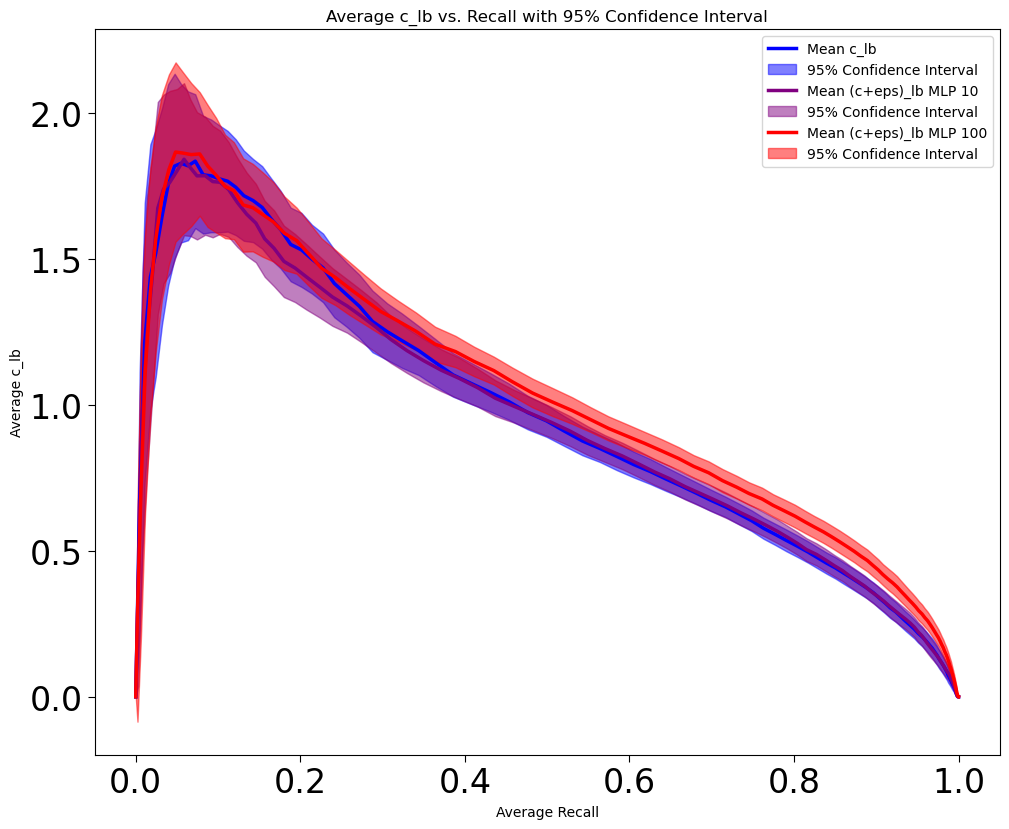}
\label{fig:privacy_recall_tab}
}
\caption{\acronym on tabular Adult dataset, MLP target model.}  
\end{figure}

The high precision at low recall in the baseline as can be seen in Figure \ref{fig:precision_recall_tab} illustrates the fact that some real data can easily be distinguished from the synthetic ones. 
Instead of poor quality synthetic data
%as exposed in~\ref{subsec:no-priv-leak}
, we are dealing with real observation that are actually ``looking to real''. 
Investigating those observations reveals that they have extreme or rare values in one of their features. 
Indeed, the Adult dataset contains some extreme values (\emph{e.g.}, in the capital\_gain column), such extreme values are rarely reproduced by the generative model. 
This leads to the baseline accurately predicting such observations as real with high confidence that translates into a high $c_{lb}$. 
Moreover, the audit values maximum values $(c+\varepsilon)_{lb}$  and $c_{lb}$ being achieved for a relatively small number of guesses corresponding to extreme values, the audits results vary a lot depending on the sampling or not of such extreme values in the audit test set. We believe this is part of the reason why we observe an high variance in the audit results leading to no significant privacy leakage detected by $\acronym$.

%This behavior is the opposite of what we generally expect in MIA in which extreme values are usually more vulnerable to MIA~\cite{carlini2022membership}. 
This is a potential drawback of our auditing scheme, while the presence of extreme or rare values in the dataset usually allows for good auditing results it can have a negative impact on \acronym's success.

\subsection{Privacy Auditing of Overfitted ML Models}
\label{appendix:sect_results_overfitting}
{\bf Methodology.}
Varying the number of training epochs for the target model to induce overfitting is known to be a factor in privacy loss~\cite{yeom2018privacy, LIRA}.
As discussed in Section~\ref{subsec:main-auditing-results}, since these different variants of target models share the same dataset and task, \acronym can compare them in terms of privacy leaking.

To verify if $\acronym$ will indeed attribute a higher value of $\tilde \epsilon$ to more overfitted models, we train our target models for varying numbers of training epochs. 
The final train and test accuracies are reported in Table~\ref{table_metrics}. 
For GPT-2 (as a target model), we report the train and validation loss as a measure of overfitting in Figure~\ref{fig:gpt-2-target-overtraining-loss}.
Figures ~\ref{fig:celeb_losses} and ~\ref{fig:gpt_losses} show how the gap between member data points (\emph{i.e.}, data used to train the target models) and non-member data points (both real as well as generated non-members) increases as the degree of overfitting increases, in terms of loss distributions. 
We study the distribution of losses since these are the features extracted from the target model $f$ or helper model $\helper$, to pass respectively to \acronym and the baseline classifier. 
The fact that the loss distributions of member data become more separable from non-member data for more overfitted models is a sign that the corresponding target model could leak more information about its training data.
%Adopting the definition proposed by~\cite{carlini2019secret}, the term \emph{``overfitting''} corresponds to the moment at which validation loss stops decreasing.
\vspace{-0.07\baselineskip}
We thus run \acronym on each model, hereafter presenting the results obtained.
\begin{figure}[h]
  \centering
  \subfigure[CelebA Multi-Label CNN\_E30.]{\includegraphics[width=0.48\linewidth]{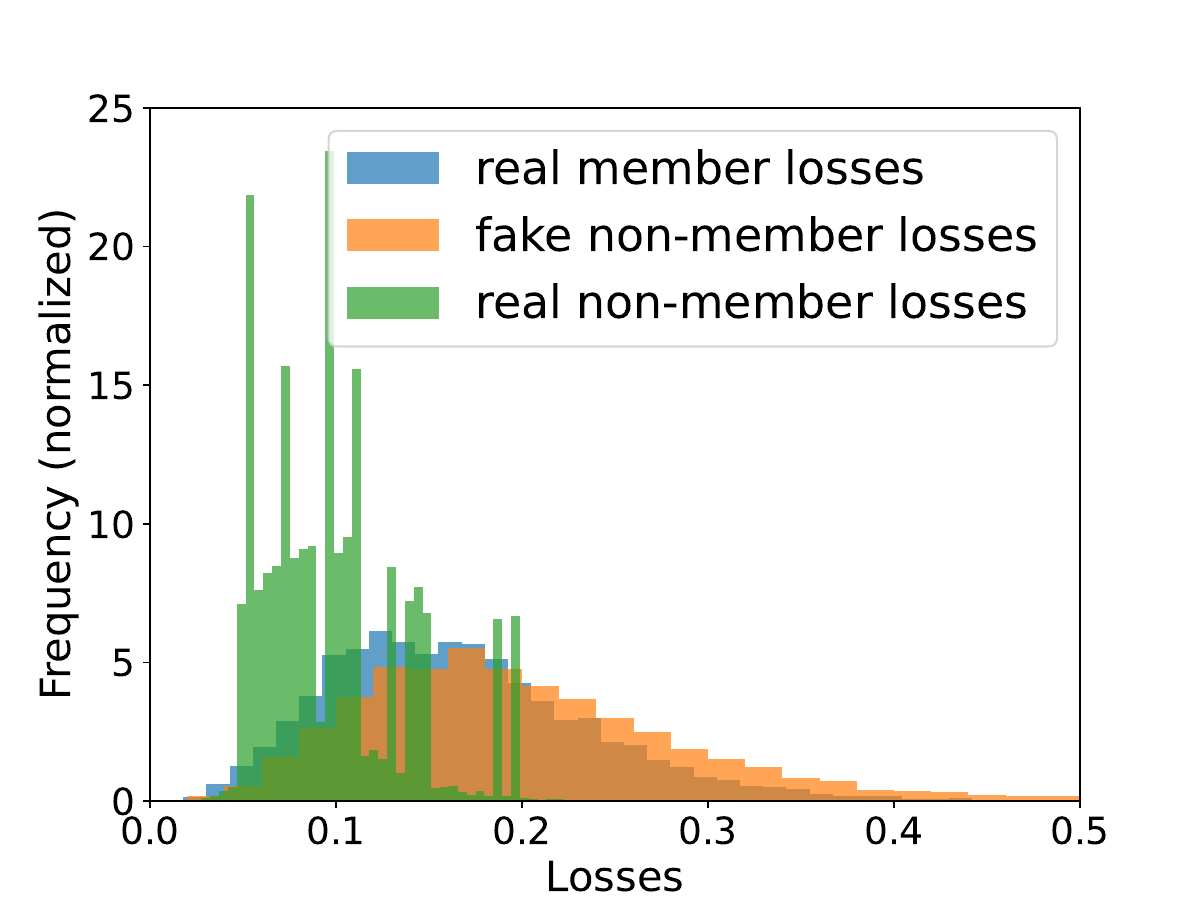}}
  \hfill
  \subfigure[CelebA Multi-Label CNN\_E100; an overfit model.]{\includegraphics[width=0.48\linewidth]{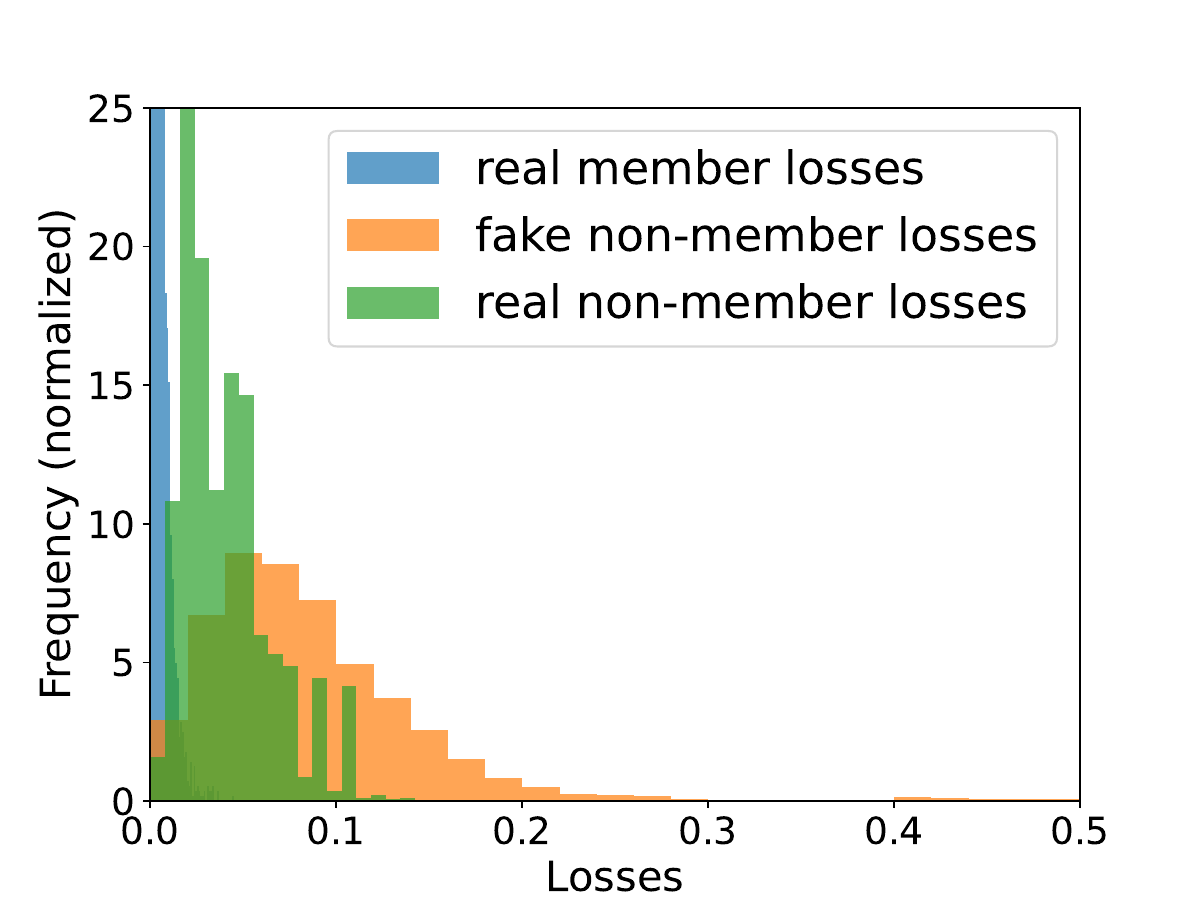}}
  \caption{CelebA Multi-Label CNN Loss Comparisons for a generalized vs overfitted model.}
  \label{fig:celeb_losses}
\end{figure}
\begin{figure}
% \vspace{-22pt}
  \centering
    \subfigure[GPT-2\_E37]{\includegraphics[width=0.45\linewidth]{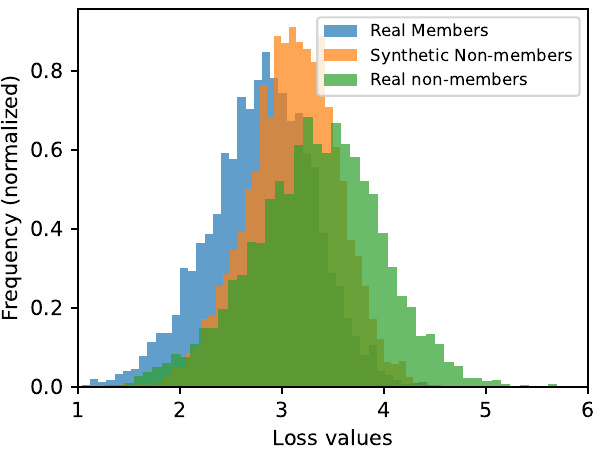}
  \label{subfig:lossdist-gpt-2E12}
  }
  % \hfill
  % \subfigure[GPT-2\_E23]{\includegraphics[width=0.45\linewidth]{images/nlp_loss_separability_target_E23_wt2_paper.pdf}
  % \label{subfig:lossdist-gpt-2E23}
  % }
  \hfill
  \subfigure[GPT-2\_E75]{\includegraphics[width=0.45\linewidth]{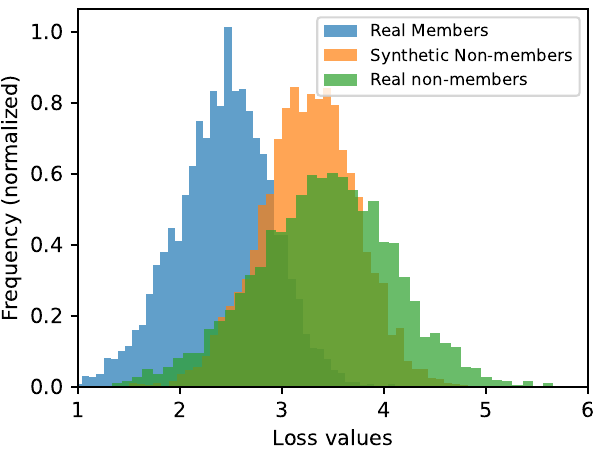}
  \label{subfig:lossdist-gpt-2E46}
  }
  \subfigure[GPT-2\_E150]{\includegraphics[width=0.45\linewidth]{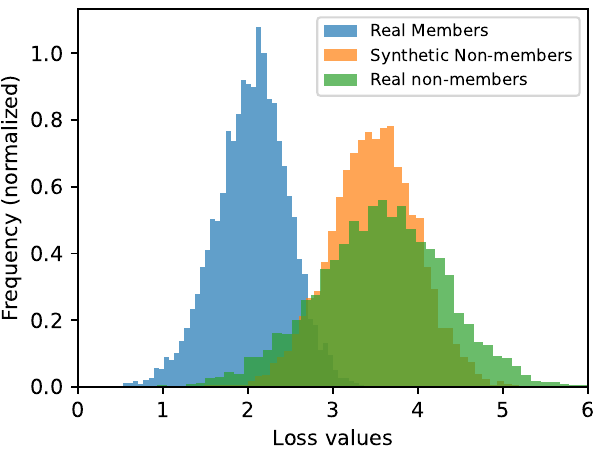}
  \label{subfig:lossdist-gpt-2E92}
  }
  \hfill
  \subfigure[Helper model $\helper$]{\includegraphics[width=0.45\linewidth]{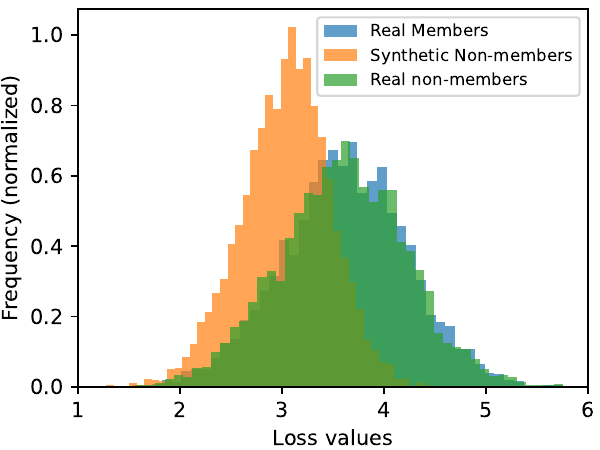}
  \label{subfig:lossdist-h}
  }
  \caption{Comparison of the loss distributions of real members, real non-members, and synthetic non-members under three target models while varying the degree of over-training on the WikiText dataset.  
  Figure~\ref{subfig:lossdist-h} compares the loss distributions under the helper model, the model providing side information to our baseline. 
  We train the helper with some other synthetic samples, which effectively mimic real non-members' loss distributions under the target models. 
  However, they are distinguishable to some extent from real non-members under the helper model, thus increasing our $\clb$.}
  %\caption{Loss Value Comparisons for helper model in the case of GPT-2} 
  
  \label{fig:gpt_losses}
\end{figure}

{\bf Results.} 
In Figure~\ref{fig:eval:comparison_precision}, we observe that more training epochs (\emph{i.e.}, more overfitting) lead to better precision-recall trade-offs and higher maximum precision values. 
Our results are further confirmed by Figure~\ref{fig:preds_comparison} with \acronym being able to capture the number of member data points (positive predictions) better than the baseline $b$.

In Table \ref{table:audit_values}, we further demonstrate that our audit output $\tilde\epsilon$ orders the target models in terms of privacy leakage: higher the degree of overfitting, more memorization and hence a higher $\tilde\epsilon$ returned by \acronym. 
From our experiments, we consistently found that as the number of epochs increased, the value of $\tilde\epsilon$ also increased. 
For instance, in the WikiText dataset, since we observe $\cpepslb^{\text{GPT-2\_E150}} > \cpepslb^{\text{GPT-2\_E75}}  > \cpepslb^{\text{GPT-2\_E37}}$, we would have $\cpepslb^{\text{GPT-2\_E150}} - \clb > \cpepslb^{\text{GPT-2\_E75}} - \clb  > \cpepslb^{\text{GPT-2\_E37}} - \clb$, which leads to $\tilde\epsilon_{\text{GPT-2\_E150}} > \tilde\epsilon_{\text{GPT-2\_E75}} > \tilde\epsilon_{\text{GPT-2\_E37}}$. 
% \mishaal{this is hard to read.}
Our experiment is coherent with the intuition that more training epochs lead to more over-fitting, leading to more privacy leakage measured with a higher value of $\tilde \epsilon$.
\begin{figure}[H]
  \centering
  {\includegraphics[width=0.5\linewidth]{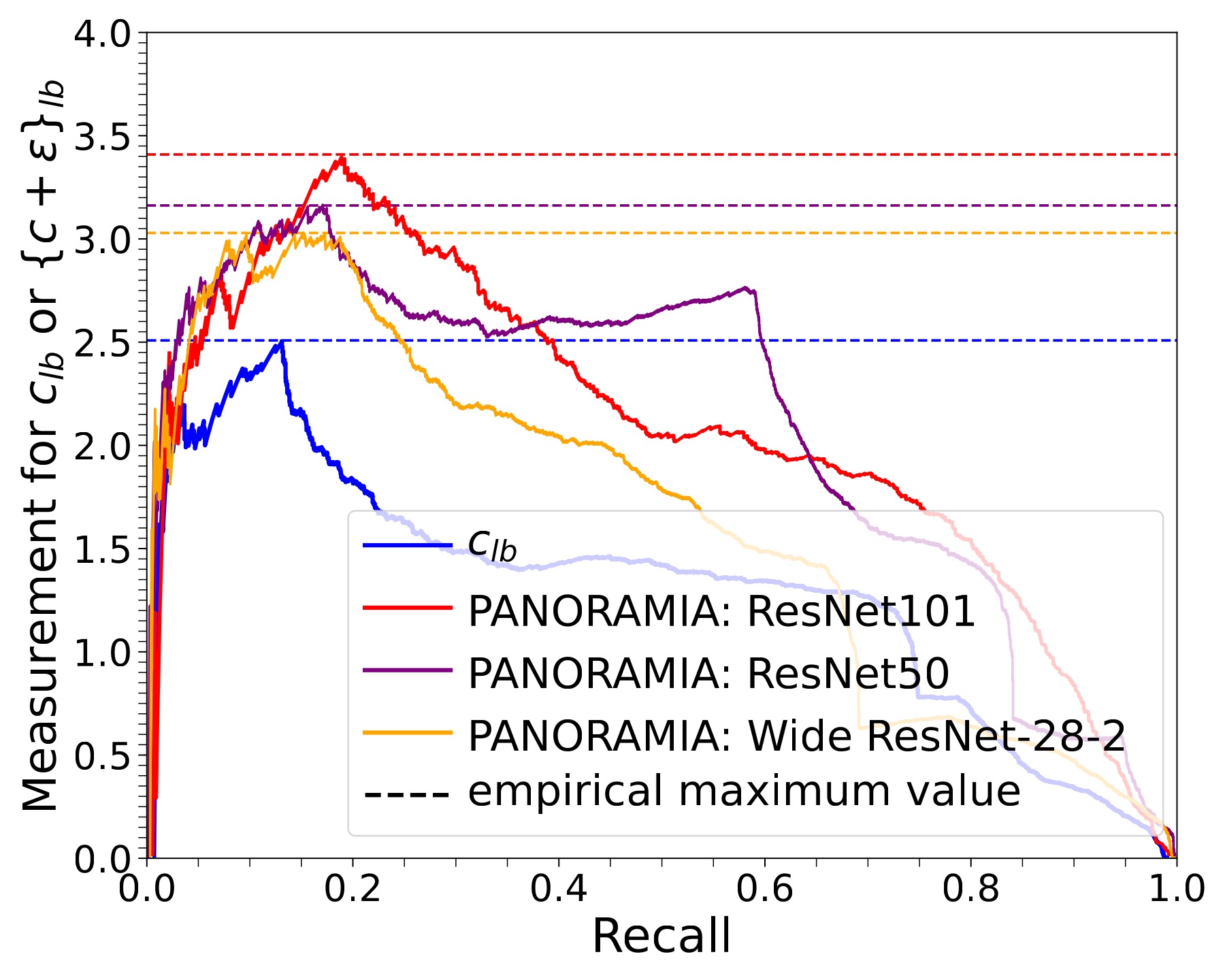} \label{subfig:var_complexities}}
      \hfill
   % returned by \acronym on different privacy audit applications. Figure \ref{subfig:resnet18-dp-fig}, shows auditing in a DP setting, whereas figure \ref{subfig:var_complexities} shows audit as model parameter size and depth increases.}
   \caption{CIFAR-10, $\{c+\epsilon\}_{lb}$ when varying privacy leakage with increasing model parameters over one experiment run.}
  \label{fig:var_complexities}
\end{figure}

\subsection{Models of varying complexity:} 
\citet{carlini2021extracting} have shown that larger models tend to have bigger privacy losses. 
To confirm this, we also audit ML models with varying numbers of parameters, from a $\approx 4M$ parameters Wide ResNet-28-2, to a $25.5M$ parameters ResNet-50, and a $44.5M$ parameters ResNet-101.
\cref{fig:var_complexities} shows that \acronym does detect increasing privacy leakage, with $\Tilde{\epsilon}_{wide-resnet} \leq \Tilde{\epsilon}_{resnet50}\leq \Tilde{\epsilon}_{resnet101}$.
% \ali{Title should be improved, since we had a experiment detail section already}
% }  
\subsection{Privacy Auditing of Differentially-Private ML Models}
\label{appendix:dp_sup_results}

{\bf Methodology.} 
% Ali: Edited to integrate the nlp results. The following comment is the version I edited. Keeping it for the reference.

% We evaluate the performance of \acronym on differentially-private ResNet-18 and Wide-ResNet-16-4 models on the CIFAR10 dataset under different target privacy budgets ($\epsilon$) with $\delta=10^{-5}$ and the non-private ($\epsilon=\infty$) cases. 
% The models are trained using the DP-SGD algorithm~\citep{abadi2016deep} using Opacus \citep{opacus}, which we tune for the highest train accuracy on learning rate $lr$, number of epochs $e$, batch size $bs$ and maximum $\ell_2$ clipping norm (C) for the largest final accuracy. 
% The noise multiplier $\sigma$ is computed given $\epsilon$, number of epochs and batch size. 
% Both \acronym and O(1)~\citep{steinke2023privacy} audits privacy loss with pure $\epsilon$-DP analysis.

For the image data modality, we evaluate the performance of \acronym on differentially-private ResNet-18 and Wide-ResNet-16-4 models on the CIFAR10 dataset under different target privacy budgets ($\epsilon$) with $\delta=10^{-5}$ and the non-private ($\epsilon=\infty$) cases. 
The models are trained using the DP-SGD algorithm~\citep{abadi2016deep} using Opacus \citep{opacus}, which we tune for the highest train accuracy on learning rate $lr$, number of epochs $e$, batch size $bs$ and maximum $\ell_2$ clipping norm (C) for the largest final accuracy. 

% \ml{[this feels like it's about NLP, but the text modality has not been introduced yet, this should move I think (or be removed?).]}
% The models are trained using the DP-SGD algorithm~\cite{abadi2016deep} with Ghost clipping~\cite{li2021large}. We optimize for the lowest validation loss by tuning the learning rate ($lr$), number of epochs ($e$), batch size ($bs$), and clipping norm ($C$). The best validation loss values achieved are $3.79$, $3.7$, and $3.65$ for $\epsilon$ values of 1, 3, and 10, respectively. For the non-private GPT-2-Small model, the best validation loss is $3.341$.

For the text data modality, we evaluate the performance of \acronym on the differentially private GPT-2-Small model, which has been fine-tuned on the WikiText dataset (the preprocessing steps are identical to those in \cref{subsec:llmfullexpdetails}). As with the image modality, we use various target privacy budgets ($\epsilon$) with $\delta=10^{-5}$, as well as non-private settings ($\epsilon=\infty$). We train the models using the DP-SGD algorithm~\cite{abadi2016deep}, incorporating Ghost clipping~\cite{li2021large} for calculating the per-example gradient. We  minimize the validation loss by tuning hyperparameters such as the learning rate ($lr$), number of epochs ($e$), batch size ($bs$), and clipping norm ($C$). The lowest validation loss values obtained are $3.79$, $3.7$, and $3.65$ for $\epsilon$ values of 1, 3, and 10, respectively. For the non-private GPT-2-Small model, the lowest validation loss achieved is $3.341$.

The noise multiplier $\sigma$ is computed given $\epsilon$, number of epochs and batch size. 
We use both \acronym and O(1)~\citep{steinke2023privacy} privacy loss measurements under the pure $\epsilon$-DP analysis.

% Table \ref{} and \ref{} summarizes the experimental details for the ResNet-18 and Wide-ResNet-16-4 models respectively.
% For the ResNet18 model, for $\epsilon=1$, the train and test accuracy are 45.31\% and 44.95\%, $lr=0.554, bs=2048, e=19, \sigma=4.960, C=1.007$; for $\epsilon=10$, the train and test accuracy are 62.33\% and 65.89\%, $lr=0.434, bs=2048, e=90, \sigma=1.178, C=3.640$; for non-private case, the train and test accuracy are 75.99\% and 99.81\%, $lr=1.258, bs=1024, e=50$.

{\bf Results.} 
For the image data modality, tables~\ref{tab:dpResNet18-table} and~\ref{tab:dpWideRes-table} summarizes the auditing results of \acronym on different DP models. 
For ResNet-18, we observe that at $\epsilon=1, 2, 4, 6$ (more private models) \acronym detects no privacy loss, whereas at higher values of $\epsilon=10, 15, 20$ (less private models) and $\epsilon=\infty$ (a non-private model) \acronym detects an increasing level of privacy loss with $\tilde{\epsilon}_{\epsilon=10} < \tilde{\epsilon}_{\epsilon=15} <  \tilde{\epsilon}_{\epsilon=20} < \tilde{\epsilon}_{\epsilon=\infty}$, suggesting a higher value of $\epsilon$ correspond to higher $\tilde \epsilon$. 
We observe a similar pattern with Wide-ResNet-16-4, in which no privacy loss is detected at $\epsilon=1,2$ and higher privacy loss is detected at $\epsilon=10,15,20,\infty$.
We also compare the auditing performance of \acronym with that of O(1)~\citep{steinke2023privacy}, with the conclusion drawn by these two methods being comparable. 
For both ResNet-18 and Wide-ResNet-16-4, O(1) reports values close to 0 (almost a random guess between members and non-members) for $\epsilon<10$ DP models, and higher values for $\epsilon=10,15,20,\infty$ DP models.
The results suggest that \acronym is potentially an effective auditing tool for DP models that has comparable performance with O(1) and can generalize to different model structures.

For the text data modality, \cref{table:audit-nlp-dp-target-models} summarizes the auditing results of \acronym on differentially private (fine-tuned) Large Language Models. In all cases, neither \acronym not O(1)~\cite{steinke2023privacy} do not detect any privacy loss. 

\begin{table}[h]
    \centering
    \resizebox{0.65\textwidth}{!}{
    \begin{tabular}{|l|l|l|l|l|l|}
        \hline
        \textbf{Target model} & \textbf{Audit} & $\mathbf{c_{lb}}$ & $\mathbf{\varepsilon + c_{lb}}$ & $\mathbf{\tilde{\varepsilon}}$ & $\mathbf{\varepsilon}$ \\
        \hline
        ResNet18 $\epsilon=\infty$ & \acronym RM;GN & 2.44 & 3.87 & 1.43 & - \\
                             & O (1) RM;RN & - & - & - & 1.471 \\
        \hline
        ResNet18 $\epsilon=20$& \acronym RM;GN & 2.44 & 3.6331 & 1.19 & - \\
                           & O (1) RM;RN & - & - & - & 1.34 \\
        \hline
        ResNet18 $\epsilon=15$& \acronym RM;GN & 2.44 & 3.57 & 1.13 & - \\
                           & O (1) RM;RN & - & - & - & 1.22 \\
        \hline
        ResNet18 $\epsilon=10$& \acronym RM;GN & 2.44 & 2.70 & 0.26 & - \\
                           & O (1) RM;RN & - & - & - & 0.28 \\
        \hline
        ResNet18 $\epsilon=6$ & \acronym RM;GN & 2.44 & 2.478 & 0.038 & - \\
                          & O (1) RM;RN & - & - & - & 0.049 \\
        \hline
        ResNet18 $\epsilon=4$ & \acronym RM;GN & 2.44 & 1.99 & 0 & - \\
                          & O (1) RM;RN & - & - & - & 0 \\
        \hline
        ResNet18 $\epsilon=2$ & \acronym RM;GN & 2.44 & 1.709 & 0 & - \\
                          & O (1) RM;RN & - & - & - & 0.05 \\
        \hline
        ResNet18 $\epsilon=1$ & \acronym RM;GN & 2.44 & 1.412 & 0 & - \\
                          & O (1) RM;RN & - & - & - & 0 \\
        \hline
    \end{tabular}}
        \vspace{8pt}
    \caption{Privacy audit of ResNet18 under different values of $\epsilon$-Differential Privacy using \acronym and O(1) auditing frameworks, in which $RM$ is for real member, $RN$ for real non-member and $GN$ for generated (synthetic) non-members.}
 \label{tab:dpResNet18-table}
\end{table}

\begin{table}[h]
    \centering
    \resizebox{0.67\textwidth}{!}{
    \begin{tabular}{|l|l|l|l|l|l|}
        \hline
        \textbf{Target model} & \textbf{Audit} & $\mathbf{c_{\textnormal{lb}}}$ & $\mathbf{\{\varepsilon + c\}_{\textnormal{lb}}}$ & $\mathbf{\tilde{\varepsilon}}$ & $\mathbf{\varepsilon}$ \\
        \hline
        WRN-16-4 $\epsilon=\infty$ & \acronym RM;GN & 2.44 & 3.02 & 0.58 & - \\
                             & O (1) RM;RN & - & - & - & 0.6408 \\
        \hline
        WRN-16-4 $\epsilon=20$& \acronym RM;GN & 2.44 & 2.926 & 0.486 & - \\
                           & O (1) RM;RN & - & - & - & 0.5961 \\  
        \hline
        WRN-16-4 $\epsilon=15$& \acronym RM;GN & 2.44 & 2.912 & 0.472 & - \\
                           & O (1) RM;RN & - & - & - & 0.5774 \\
        
        \hline
        WRN-16-4 $\epsilon=10$& \acronym RM;GN & 2.44 & 2.783 & 0.343 & - \\
                           & O (1) RM;RN & - & - & - & 0.171 \\
        \hline
        WRN-16-4 $\epsilon=2$ & \acronym RM;GN & 2.44 & 2.36 & 0 & - \\
                          & O (1) RM;RN & - & - & - & 0 \\
        \hline
        WRN-16-4 $\epsilon=1$ & \acronym RM;GN & 2.44 & 1.84& 0 & - \\
                          & O (1) RM;RN & - & - & - & 0 \\
        \hline
    \end{tabular}}
        \vspace{8pt}

    \caption{Privacy audit of Wide ResNet16-4 under different values of $\epsilon$-Differential Privacy (DP) using \acronym and O(1) auditing frameworks, where $RM$ is for real member, $RN$ for real non-member and $GN$ for generated (synthetic) non-members. 
    }
    \label{tab:dpWideRes-table}
\end{table}

\begin{table}[h]
\centering
\resizebox{0.67\textwidth}{!}{
\begin{tabular}{|l|l|l|l|l|}
\hline
\textbf{Target model} $f$ & Audit & \textbf{Generator} $G$   & $\mathbf{\tilde\epsilon}$ & $\mathbf{\epsilon}$\\
\hline

$\infty$-DP GPT2-Small (GPT2\_E37) & \acronym & GPT2-Small  & $0.22 \pm 0.37$ & - \\
 & O (1) RM;RN & -  & - & $2.82 \pm 0.31$ \\

\hline
%  & \acronym & $3.78$ & $3.96$ & $0.18$ &-\\
%  GPT2\_E23 & \acronym{RM;RN}& $0$ & $1.15$  & - &$1.15$\\
%  & O (1) RM;RN    & -   & -     & - &$2.54$\\
% \hline
% Ali: Commennted this row for saving space
 % & \acronym  & $\infty$-DP GPT-2-Small  & $0$ & - \\
 % & \acronym  & $3$-DP GPT-2-Small  & $0.059 \pm 0.15$ & - \\
$1$-DP GPT2-Small & \acronym  & GPT-2-Small  & $0$ & - \\
 % & \acronym  & $3$-DP GPT-2-Large  & $0.019 \pm 0.18$ & - \\
 & O (1) RM;RN   & - & - &0 \\

\hline
 % & \acronym  & $\infty$-DP GPT-2-Small  & $0$  & - \\
 % & \acronym  & $3$-DP GPT-2-Small  & $0.116 \pm 0.30$ & - \\
$3$-DP GPT2-Small & \acronym  & GPT-2-Small  & $0$ & - \\
 % & \acronym  & $3$-DP GPT-2-Large  & $0.095 \pm 0.27$ & - \\
 & O (1) RM;RN   & - & - &0 \\
\hline
 % & \acronym  & $\infty$-DP GPT-2-Small  & $0$ & - \\
 % & \acronym  & $3$-DP GPT-2-Small  & $0.124 \pm 0.22$ & - \\
$10$-DP GPT2-Small & \acronym  &  GPT-2-Small  & $0$ & - \\
 % & \acronym  & $3$-DP GPT-2-Large  & $0.048 \pm 0.09$ & - \\
 & O (1) RM;RN   & - & - &0 \\

\hline
\end{tabular}
}
\vspace{8pt}
\caption{Privacy audits on GPT-2-Small target models with varying privacy guarantees.}
\label{table:audit-nlp-dp-target-models}
\end{table}

\paragraph{Discussion.}
In the image data modality, we observe that the auditing outcome ($\tilde \epsilon$ values for \acronym and $\epsilon$ for O(1)) can be different for DP models with the same $\epsilon$ values (Table~\ref{tab:audit_different_same_eps}). 
We hypothesize that the auditing results may relate more to level of overfitting than the target $\epsilon$ values in trained DP models. 
The difference between train and test accuracies could be a possible indicator that has a stronger relationship with the auditing outcome. 
We also observe that O(1) shows a faster increase in $\epsilon$ for DP models with higher targeted $\epsilon$ values.
We believe it depends on the actual ratio of the correct and total number of predicted samples, since O(1) considers both true positives and true negatives while \acronym considers true positives only. 
We leave these questions for future work.

In the text data modality, the membership inference attacks (MIAs) relying on loss thresholds in privacy audits, without retraining or access to canaries, may not be sufficiently strong to effectively classify members and non-members.

\begin{table}[htb]
\centering
\resizebox{0.9\textwidth}{!}{%
\begin{tabular}{|c|c|c|c|c|c|c|c|c|}
\hline
\textbf{Target model} & \textbf{Audit} & $\mathbf{c_{lb}}$ & $\mathbf{\varepsilon + c_{lb}}$ & $\mathbf{\tilde{\varepsilon}}$ & $\mathbf{\varepsilon}$ & \textbf{Train Acc} & \textbf{Test Acc} & \textbf{Diff(Train-Test Acc)} \\ \hline
\multirow{4}{*}{ResNet18-eps-20} & PANORAMIA RM;GN & 2.508 & 3.63 & 1.06 & - & 71.82 & 67.12 & 4.70 \\
 & O (1) RM;RN & - & - & - & 1.22 & - & - & - \\
 & \begin{tabular}[c]{@{}c@{}}PANORAMIA \\ RM;GN\end{tabular} & 2.508 & 2.28 & 0 & - & 71.78 & 68.08 & 3.70 \\
 & O (1) RM;RN & - & - & - & 0.09 & - & - & - \\ \hline
\multirow{4}{*}{ResNet18-eps-15} & PANORAMIA RM;GN & 2.508 & 3.63 & 1.13 & - & 69.01 & 65.7 & 3.31 \\
 & O (1) RM;RN & - & - & - & 1.34 & - & - & - \\
 & PANORAMIA RM;GN & 2.508 & 1.61 & 0 &  & 66.68 & 69.30 & 2.62 \\
 & O (1) RM;RN & - & - & - & 0.08 & - & - & - \\ \hline
\end{tabular}%
}
        \vspace{5pt}

\caption{DP models with the same $\epsilon$ values can have different auditing outcomes (the above numbers are reported here without union-bound correction).}
\label{tab:audit_different_same_eps}
\end{table}

\section{Generator Closeness $\gamma$-relaxation}
\label{appendix:delta_relaxation}

\begin{definition}[$(c, \gamma)$-closeness]
\label{def:c_delta_closeness}
Let the $\gamma$-approximate maximum divergence between generative model $\gG$ and data distribution $\gD$ be,
\[
    D_{\infty}^{\gamma}(\gD \| \gG) \coloneqq \max_{S \subseteq Supp(\gD): \sP(\gD \in S) \geq \gamma} \bigg[ \ln \frac{\sP(\gD \in S) - \gamma}{\sP(\gG \in S)}\bigg].
\]
For all $c>0$ we say $\gG$ is $(c, \gamma)$-close to $\gD$ if $D_{\infty}^{\gamma}(\gD \| \gG) \leq c$, i.e. 
\[
    % \forall x \in \gX, e^{-c}(\sP_\gD[x] - \delta) \leq \sP_\gG[x].
    \forall x \in \gX, \sP_\gD[x] \leq e^{c}\sP_\gG[x] + \gamma \; \textnormal{or} \; e^{-c}(\sP_\gD[x] - \gamma) \leq \sP_\gG[x].
\]
\end{definition}

\subsection{Pure DP case}
\label{appendix:delta_relaxation_pure}

First we prove a similar result as Lemma 5.5 \citep{steinke2023privacy} but with the one-sided inequality of $(c, \gamma)$-closeness and the case when we additionally observe output from a pure-DP mechanism.

\begin{lemma}
\label{lemma:decomposition_pure}
Let $\gG, \gD$ be probability distributions over $\gY$, fix $c, \gamma \geq 0$. Suppose for all measurable $S \subset \gY$ we have $e^{-c}(\gD(S) - \gamma) \leq \gG(S)$. Then $\exists \gamma' \in [0, \gamma]$ and $\gD', \gD'', \gG', \gG''$ such that the following properties are satisfied. 
\begin{enumerate}[label={(\arabic*)}]
    \item $\gD, \gG$ can be expressed as a convex combination $\gD = (1-\gamma')\gD' + \gamma' \gD'', \gG = (1-\gamma')\gG' + \gamma' \gG''$;
    \item for all measurable $S \subset \gY$, we have $e^{-c}\gD'(S) \leq \gG'(S)$;
    \item there exists measurable $S \subset \gY$ such that $\gD''(S)=1$, $\forall S' \subset S, \gD(S') \geq \gG(S')$.
\end{enumerate}
Let $M$ denote an arbitrary $\epsilon$-DP mechanism which takes input $Y$ sampled from distribution $\gG$ or $\gD$ and outputs $f$. Let $\mathcal{P}, \mathcal{Q}$ be probability distributions over $\mathcal{Z}$ which are distributions of $M(Y)$ when $Y$ is sampled from $\gD$ and $\gG$ respectively. Fix $\epsilon \geq 0$, suppose for all measurable $T \subset \mathcal{Z}$ we have $\mathcal{P}(T) \leq e^{\epsilon}\mathcal{Q}(T)$ and $\mathcal{Q}(T) \leq e^{\epsilon}\mathcal{P}(T)$. When additionally observing the output from a DP mechanism we have the following properties. Let $(\gD, \mathcal{P}), (\gG, \mathcal{Q})$ be the joint distribution over $(\gY, \mathcal{Z})$ which we assume $\gD \indep \mathcal{P}$, $\gG \indep \mathcal{Q}$.
\begin{enumerate}[resume*]
    \item $(\gD, \mathcal{P}), (\gG, \mathcal{Q})$ can be expressed as a convex combination $\gD \mathcal{P} = ((1-\gamma')\gD' + \gamma' \gD'')\mathcal{P}, \gG \mathcal{Q} = ((1-\gamma')\gG' + \gamma' \gG'')\mathcal{Q}$;
    \item for all measurable $S \subset \gY$, $T \subset \mathcal{Z}$, we have $e^{-c}e^{-\epsilon}\gD'(S)\mathcal{P}(T) \leq \gG'(S)\mathcal{Q}(T)$;
    \item there exists measurable $S \subset \gY$, $T \subset \mathcal{Z}$ such that $\gD''(S)\mathcal{P}(T)=1$, $\forall S' \subset S, T' \subset T, \gD(S')\mathcal{P}(T) \geq \gG(S')\mathcal{P}(T)$.
\end{enumerate}
\end{lemma}
\begin{proof}
    For the edge cases of $\gamma=0$, the results hold with $\gamma'=0, \gD'=\gD, \gG'=\gG$; similarly when $\gamma=1$, the results hold with $\gamma'=1, \gD''=\gD, \gG''=\gG, \gD'=\gG'$.
    Let $c' \in [0, c]$, $\gamma_1, \gamma_2 \in (0, 1)$, define distribution $\gG', \gG'', \gD', \gD''$ as follows. For all points $y \in \gY$,
    \begin{align*}
        \gD'(y) &= \frac{\min\{\gD(y), e^{c'}\gG(y)\}}{1-\gamma_1}, \\
        \gD''(y) &= \frac{\gD(y)-(1-\gamma_1)\gD'(y)}{\gamma_1} = \frac{\max\{ 0, \gD(y)- e^{c'}\gG(y)\}}{\gamma_1}, \\
        \gG'(y) &= \frac{\gG(y)}{1-\gamma_2},\\
        \gG''(y) &= \frac{\gG(y)-(1-\gamma_2)\gG'(y)}{\gamma_2} = 0.
    \end{align*}
    By construction $(1-\gamma_1)\gD' + \gamma_1 \gD'' = \gD$, $(1-\gamma_2)\gG' + \gamma_2 \gG'' = \gG$, so the first property is satisfied, and $\gD''(y)$ is supported on $S = \{y\in \gY: \gD(y) > e^{c'}\gG(y)\}$ so the third property is implied. 
    If $0< \gamma_1 =\gamma_2 =\gamma' \leq \gamma$, for all $y \in \gY$ we have,
    \[
        \frac{\gG'(y)}{\gD'(y)} =\frac{\gG(y)}{\min\{\gD(y), e^{c'}\gG(y)\}} \geq e^{-c'} \geq e^{-c},
    \]
    as required for the second property.
    Following the same as in Lemma 5.5 \citep{steinke2023privacy}, we can ensure $0 < \gamma_1 = \gamma_2 \leq \gamma$ by appropriately setting $\epsilon_1, \epsilon_2 \in [0, \epsilon]$.
    We can use the same decomposition of $\gD, \gG$ to prove the second part of the Lemma. If assuming $\gD  \indep \mathcal{P}$ and $\gG \indep \mathcal{Q}$, the joint distributions can be decomposed as $f_{\gD, \mathcal{P}} = f_{\mathcal{P}} \cdot f_{\gD}$ and $f_{\gG, \mathcal{Q}} = f_{\mathcal{Q}} \cdot f_{\gG}$.
    Therefore by construction, $\gD \mathcal{P} = ((1-\gamma')\gD' + \gamma' \gD'')\mathcal{P}, \gG \mathcal{Q} = ((1-\gamma')\gG' + \gamma' \gG'')\mathcal{Q}$. $\gD''\mathcal{P}$ is supported on $(S,T)=\{ y\in \gY, M(y) \in \mathcal{Z}: \gD(y)\mathcal{P}(M(y)) > e^{c'}\gG(y)\mathcal{P}(M(y)) \}$. For all $y \in \gY$, $M(y) \in \mathcal{Z}$ we have,
    \[
        \frac{\gG'(y)}{\gD'(y)} \cdot \frac{\mathcal{Q}(M(y))}{\mathcal{P}(M(y))} = \frac{\gG(y)}{\min\{\gD(y), e^{c'}\gG(y)\}} \cdot \frac{\mathcal{Q}(M(y))}{\mathcal{P}(M(y))} \geq e^{-c'}e^{-\epsilon} \geq e^{-c}e^{-\epsilon},
    \]
    where $\frac{\mathcal{Q}(M(y))}{\mathcal{P}(M(y))} \geq e^{-\epsilon}$ since $M$ is an $\epsilon$-DP mechanism.
\end{proof}

Similar to Lemma 5.6 \citep{steinke2023privacy}, next we prove a Bayesian version of Lemma \ref{lemma:decomposition_pure} which is used to prove the main result afterwards.

\begin{lemma}
\label{lemma:baysian_decomposition_pure}
Let $\gG, \gD$ be probability distributions over $\gY$, fix $c, \gamma \geq 0$, suppose for all measurable $S \subset \gY$ we have $e^{-c}(\gD(S) - \gamma) \leq \gG(S)$. Then there exists a randomized function $E_{\gD, \gG}: \gY \rightarrow \{0, 1\}$ with the following properties. Suppose $X \sim \mathrm{Bernoulli}(\frac{1}{2})$, if $X=1$ sample $Y \sim \gD$, and if $X=0$ sample $Y \sim \gG$. Then for all $y \in \gY$ we have,
\begin{gather*}
    \underset{\substack{X \sim \textrm{Bernoulli}(\frac{1}{2})\\Y \leftarrow X\gD+(1-X)\gG}}{\sP}[X=1, E_{\gD, \gG}(Y)=1|Y=y] \leq \frac{e^c}{1+e^c}, \\
    \E_{Y \sim \gD}[E_{\gD, \gG}(Y)] \geq 1-\gamma, \; \E_{Y \sim \gG}[E_{\gD, \gG}(Y)] \geq 1-\gamma.
\end{gather*}
Let $M$ denote an arbitrary $\epsilon$-DP mechanism which takes input $Y$ sampled from distribution $\gG$ or $\gD$ and outputs $f$. Let $\mathcal{P}, \mathcal{Q}$ be probability distributions over $\mathcal{Z}$ which are distributions of $M(Y)$ when $Y$ is sampled from $\gD$ and $\gG$ respectively. Fix $\epsilon \geq 0$, suppose for all measurable $T \subset \mathcal{Z}$ we have $\mathcal{P}(T) \leq e^{\epsilon}\mathcal{Q}(T)$ and $\mathcal{Q}(T) \leq e^{\epsilon}\mathcal{P}(T)$. Let $(\gD, \mathcal{P}), (\gG, \mathcal{Q})$ denote the joint distributions over $(\gY, \mathcal{Z})$ which we assume $\gD \indep \mathcal{P}$, $\gG \indep \mathcal{Q}$. Then there exists a randomized function $E_{(\gD, \mathcal{P}), (\gG, \mathcal{Q})}: (\gY, \mathcal{Z}) \rightarrow \{0, 1\}$ such that for all $y \in \gY$, $M(y) \in \mathcal{Z}$ we have,
\begin{gather*}
    \underset{\substack{X \sim \textrm{Bernoulli}(\frac{1}{2})\\Y \leftarrow X\gD+(1-X)\gG \\ M(Y)\leftarrow M(X\gD+(1-X)\gG)}}{\sP}[X=1, E_{(\gD, \mathcal{P}), (\gG, \mathcal{Q})}(Y, M(Y))=1|Y=y, M(Y)=f] \leq \frac{e^{c+\epsilon}}{1+e^{c+\epsilon}}, \\
    \E_{Y \sim \gD, M(Y) \sim \mathcal{P}}[E_{(\gD, \mathcal{P}), (\gG, \mathcal{Q})}(Y, M(Y))] \geq 1-\gamma, \\
    \E_{Y \sim \gG, M(Y) \sim \mathcal{Q}}[E_{(\gD, \mathcal{P}), (\gG, \mathcal{Q})}(Y, M(Y))] \geq 1-\gamma.
\end{gather*}
\end{lemma}
\begin{proof}
    We apply the decomposition from Lemma \ref{lemma:decomposition_pure} and denote $\gD', \gD'', \gG', \gG'', \gamma'$ as is. We define $E_{\gD, \gG}: \gY \rightarrow \{0, 1\}$ by,
    \[
        \sP[E_{\gD, \gG}(y)=1] = (1-\gamma')\frac{\gD'(y)}{\gD(y)} = 1-\frac{\gamma'\gD''(y)}{\gD(y)}.
    \]
    For any $y \in \gY$ we have,
    \begin{align*}
        &\sP[X=1, E_{\gD, \gG}(Y)=1 | Y=y] \\
        &= \sP[X=1 | Y=y, E_{\gD, \gG}(Y)=1] \sP[E_{\gD, \gG}(Y)=1 | Y=y] \\
        &= \sP[X=1| Y=y] \sP[E_{\gD, \gG}(y)=1] \\
        &= \frac{\sP[Y=y|X=1]\sP[X=1]}{\sP[Y=y]}\sP[E_{\gD, \gG}(y)=1] \\
        &= \frac{\gD(y)}{\gD(y)+\gG(y)}\sP[E_{\gD, \gG}(y)=1] \\
        % &= \frac{(1-\delta')\gD'(y)+\delta'\gD''(y)}{(1-\delta')\gD'(y)+\delta'\gD''(y)+(1-\delta')\gG'(y)+\delta'\gG''(y)}\sP[E_{\gG, \gD}(Y)=1] \\
        &= \frac{(1-\gamma')\gD'(y)+\gamma'\gD''(y)}{(1-\gamma')\gD'(y)+\gamma'\gD''(y)+(1-\gamma')\gG'(y)}\sP[E_{\gG, \gD}(y)=1] \\
        &= \frac{1+\frac{\gamma'\gD''(y)}{(1-\gamma')\gD'(y)}}{1+\frac{\gamma'\gD''(y)}{(1-\gamma')\gD'(y)}+\frac{\gG'(y)}{\gD'(y)}}\sP[E_{\gG, \gD}(y)=1]\\
        &\leq \frac{1+\frac{\gamma'\gD''(y)}{(1-\gamma')\gD'(y)}}{1+0+e^{-c}}\sP[E_{\gG, \gD}(y)=1]\\
        &= \frac{1}{1+e^{-c}}\cdot \bigg( \frac{(1-\gamma')\gD'(y)+\gamma'\gD''(y)}{(1-\gamma')\gD'(y)}\bigg) \cdot \sP[E_{\gG, \gD}(y)=1] \\
        &= \frac{1}{1+e^{-c}}\cdot \bigg( \frac{\gD(y)}{(1-\gamma')\gD'(y)}\bigg) \cdot \sP[E_{\gG, \gD}(y)=1] \\
        &= \frac{1}{1+e^{-c}} = \frac{e^c}{1+e^c}.
    \end{align*}
    The first equality uses that $\sP[A,B|C]=\sP[A|B,C]\sP[B|C]$. The second equality holds since we assume the internal randomness of $E$ is independent of everything else. The third equality uses the Bayes' rule. The fourth equality uses the definition of the randomized function $E$. The fifth equality uses the decomposition of Lemma \ref{lemma:decomposition_pure}. After rearranging and use the properties from Lemma \ref{lemma:decomposition_pure} we obtain the first inequality, then the rest follows from simplifications. Furthermore we have the following,
    \begin{align*}
        \E_{Y \sim \gD}[E_{\gD, \gG}] &= \int_{\gY} \gD(y) \sP[E_{\gD, \gG}=1]dy \\
        &= \int_{\gY} (1-\gamma')\gD'(y)dy = 1-\gamma' \geq 1-\gamma,\\
        \E_{Y \sim \gG}[E_{\gD, \gG}] &= 1-\gamma'\E_{Y \sim \gG}\bigg[\frac{\gD''(y)}{\gD(y)} \bigg] \\
        &= 1-\gamma' \int_{\gY}\frac{\gG(y)}{\gD(y)}\cdot \gD''(y) dy \\
        &\geq 1-\gamma' \int_{\gY}\gD''(y) dy \\
        &= 1-\gamma' \geq 1-\gamma.
    \end{align*}
    %%%%%
    We follow a similar procedure to prove the second part of the Lemma. We define $E_{(\gD, \mathcal{P}), (\gG, \mathcal{Q})}: (\gY, \mathcal{Z}) \rightarrow \{0, 1\}$ by,
    \begin{align*}
        % \sP[E_{(\gD, \mathcal{P}), (\gG, \mathcal{Q})}(y, M(y))=1] &= \frac{(1-\delta')\gD'(y)\mathcal{P}(M(y))}{\gD(y)\mathcal{P}(M(y))} \\ &= 1-\frac{\delta'\gD''(y)\mathcal{P}(M(y))}{\gD(y)\mathcal{P}(M(y))}.
        \sP[E_{(\gD, \mathcal{P}), (\gG, \mathcal{Q})}(y, M(y))=1] = \sP[E_{\gD, \gG}(y)=1] 
        %= (1-\delta')\frac{\gD'(y)}{\gD(y)} = 1-\frac{\delta'\gD''(y)}{\gD(y)}.
    \end{align*}
    With similar reasoning, for any $y \in \gY$, $M(y) \in \mathcal{Z}$ we have,
    \begin{align*}
        &\sP[X=1, E_{(\gD, \mathcal{P}), (\gG, \mathcal{Q})}(Y, M(Y))=1 | Y=y, M(Y)=f] \\
        % &= \sP[X=1| Y=y, M(Y)=f]] \sP[E_{(\gD, \mathcal{P}), (\gG, \mathcal{Q})}(Y, M(Y))=1] \\
        &= \frac{\sP[Y=y, M(Y)=f|X=1]\sP[X=1]}{\sP[Y=y, M(Y)=f]}\sP[E_{(\gD, \mathcal{P}), (\gG, \mathcal{Q})}(y, M(y))=1] \\
        &= \frac{\gD(y) \mathcal{P}(M(y))}{\gD(y)\mathcal{P}(M(y))+\gG(y)\mathcal{Q}(M(y))}\sP[E_{(\gD, \mathcal{P}), (\gG, \mathcal{Q})}(y, M(y))=1] \\
        &= \frac{((1-\gamma')\gD'(y)+\gamma'\gD''(y))\mathcal{P}(M(y))}{((1-\gamma')\gD'(y)+\gamma'\gD''(y))\mathcal{P}(M(y))+(1-\gamma')\gG'(y)\mathcal{Q}(M(y))}\sP[E_{(\gD, \mathcal{P}), (\gG, \mathcal{Q})}(y, M(y))=1] \\
        &= \frac{1+\frac{\gamma'\gD''(y)\mathcal{P}(M(y))}{(1-\gamma')\gD'(y)\mathcal{P}(M(y))}}{1+\frac{\gamma'\gD''(y)\mathcal{P}(M(y))}{(1-\gamma')\gD'(y)\mathcal{P}(M(y))}+\frac{\gG'(y)\mathcal{Q}(M(y))}{\gD'(y)\mathcal{P}(M(y))}}\sP[E_{(\gD, \mathcal{P}), (\gG, \mathcal{Q})}(y, M(y))=1] \\
        &\leq \frac{1+\frac{\gamma'\gD''(y)}{(1-\gamma')\gD'(y)}}{1+0+e^{-c}e^{-\epsilon}}\sP[E_{(\gD, \mathcal{P}), (\gG, \mathcal{Q})}(y, M(y))=1] \\
        % &= \frac{1}{1+e^{-c}e^{-\epsilon}} \cdot \bigg( \frac{(1-\delta')\gD'(y)\mathcal{P}(M(y))+\delta'\gD''(y)\mathcal{P}(M(y))}{(1-\delta')\gD'(y)\mathcal{P}(M(y))} \bigg) \cdot \sP[E_{(\gD, \mathcal{P}), (\gG, \mathcal{Q})}(y, M(y))=1] \\
        % &= \frac{1}{1+e^{-c}e^{-\epsilon}} \cdot \bigg( \frac{\gD(y)\mathcal{P}(M(y))}{(1-\delta')\gD'(y)\mathcal{P}(M(y))} \bigg) \cdot \sP[E_{(\gD, \mathcal{P}), (\gG, \mathcal{Q})}(y, M(y))=1] \\
        % &= \frac{1}{1+e^{-c}e^{-\epsilon}} \cdot \bigg( \frac{(1-\delta')\gD'(y)+\delta'\gD''(y)}{(1-\delta')\gD'(y)} \bigg) \cdot \sP[E_{(\gD, \mathcal{P}), (\gG, \mathcal{Q})}(y, M(y))=1] \\
        &= \frac{1}{1+e^{-c}e^{-\epsilon}} \cdot \bigg( \frac{\gD(y)}{(1-\gamma')\gD'(y)} \bigg) \cdot \sP[E_{(\gD, \mathcal{P}), (\gG, \mathcal{Q})}(y, M(y))=1] \\
        &= \frac{1}{1+e^{-c}e^{-\epsilon}} = \frac{e^{c+\epsilon}}{1+e^{c+\epsilon}}.
    \end{align*}
    Since we define the success probability with the same expression, the expected success of the randomized function directly follows,
    \begin{align*}
        \E_{Y \sim \gD, M(Y) \sim \mathcal{P}}[E_{(\gD, \mathcal{P}), (\gG, \mathcal{Q})}(Y, M(Y))] =  \E_{Y \sim \gD}[E_{\gD, \gG}(Y)] \geq 1-\gamma, \\
        \E_{Y \sim \gG, M(Y) \sim \mathcal{Q}}[E_{(\gD, \mathcal{P}), (\gG, \mathcal{Q})}(Y, M(Y))]= \E_{Y \sim \gG}[E_{\gD, \gG}(Y)] \geq 1-\gamma.
    \end{align*}
\end{proof}

Next we use Lemma \ref{lemma:baysian_decomposition_pure} to prove the following Proposition which is used to form statistical tests on hypothesis $\mathcal{H}$.

\begin{proposition}
\label{prop:stat_test_pure}
    Let $\gG$ be $(c,\gamma)$-close, $f$ be output of an $\epsilon$-DP mechanism, $T^b \triangleq B(S, X)$ be the guess from the baseline, $T^a \triangleq A(S, X, f)$ be the guess from the membership audit. Let $T_b, T_a \in [0, 1]^m$ be bounded.
    Then, for all $v \in \sR$ and all $t$ in the support of $T$:
    \begin{gather*}
    \sP_{S, X, T^b}\Big[ \sum_{i=1}^m T^b_i \cdot S_i \geq v \ | \ T^b = t^b \Big] \leq \underset{\substack{S' \sim \textrm{Bernoulli}(\frac{e^c}{1+e^c})^m, F}}\sP\Big[F(t^b)+ \sum_{i=1}^m t^b_i \cdot S'_i \geq v \Big], \\
    \sP_{S, X, T^a}\Big[ \sum_{i=1}^m T^a_i \cdot S_i \geq v \ | \ T^a = t^a \Big] \leq \underset{\substack{S' \sim \textrm{Bernoulli}(\frac{e^{c+\epsilon}}{1+e^{c+\epsilon}})^m, F}}\sP\Big[F(t^a)+ \sum_{i=1}^m t^a_i \cdot S'_i \geq v \Big],
    \end{gather*}
    where $F$ is independent from $S'$, $F(T^b), F(T^a)$ is supported on $\{0,1,\ldots,m\}$ and $\E_{T^b, F}[F(T^b)] = \E_{T^a, F}[F(T^a)] \leq 2m\gamma$.    
\end{proposition}
\begin{proof}
    Let $B(s_{\leq i})$ denote the distribution on $[0, 1]^m$ obtained by conditioning $B$ on past $S_{< i}=s_{< i}$ and $X_{\leq i}=x_{\leq i}$. By Lemma \ref{lemma:baysian_decomposition_pure}, for all $t^b \in [0,1]^m$ we have
    \begin{gather*}
        \sP[S_i=1, E_{B(s_{<i, 1}), B(s_{< i, 0})}(T^b)=1 | T^b = t^b, S_{<i}=s_{<i}, X_{\leq i}=x_{\leq i}] \leq \frac{e^c}{1+e^c}, \\
        \E[E_{B(s_{\leq i, 1}), B(s_{\leq i, 0})}(T^b)=1 | S_{<i}=s_{<i}, S_i = 0] \geq 1-\gamma, \\
        \E[E_{B(s_{\leq i, 1}), B(s_{\leq i, 0})}(T^b)=1 | S_{<i}=s_{<i}, S_i = 1] \geq 1-\gamma.
    \end{gather*}
    Using the law of total probability we get,
    \begin{align*}
        &\sP[S_i=1, E_{B(s_{<i, 1}), B(s_{< i, 0})}(T^b)=1 | T^b = t^b, S_{<i}=s_{<i}] \\
        &= \sum_{x\leq i} \sP[S_i=1, E_{B(s_{<i, 1}), B(s_{< i, 0})}(T^b)=1 | T^b = t^b, S_{<i}=s_{<i}, X_{\leq i}=x_{\leq i}] \\
        & \;\;\;\;\;\; \sP[X_{\leq i}=x_{\leq i} | T^b = t^b, S_{<i}=s_{<i}] \\
        &\leq \frac{e^c}{1+e^c} \sum_{x\leq i}\sP[X_{\leq i}=x_{\leq i} | T^b = t^b, S_{<i}=s_{<i}] = \frac{e^c}{1+e^c}.
    \end{align*}
    Applying the union bound on $S_i$ we can bound the expectation of success events by,
    \[
        \E[E_{B(s_{\leq i, 1}), B(s_{\leq i, 0})}(T^b)| S_{<i}=s_{<i}, S_i=s_i] \geq 1-2\gamma.
    \]
    By induction, for $k \in [m]$ define $\Tilde{W}_{k}(s,t^b) \coloneqq \sum_{i=1}^{k} t_i^b \cdot s_i \cdot E_{B(s_{<i, 1}), B(s_{<i, 0})}(t^b)$, then $\Tilde{W}_{k}(S,t^b)$ is stochastically dominated by $\check{W}_k(t^b) \coloneqq \sum_{i=1}^{k} t_i^b \cdot \check{S}_i(t)$ where $\check{S}_i(t) \sim \mathrm{Bernoulli}(\frac{e^c}{1+e^c})$. Let the failing events $F$ be, 
    \begin{gather*}
        F(s, t^b) \coloneqq \sum_{i=1}^{m} \mathds{1}\{ E_{B(s_{\leq i, 1}), B(s_{\leq i, 0})}(t^b)=0 \},
    \end{gather*}
    then we have $W_{m}(s,t^b) \coloneqq \sum_{i=1}^{m} t_i^b  \cdot s_i = \Tilde{W}_{m}(s,t^b) + F(s,t^b)$ is stochastically dominated by $\check{W}_{m}(T^b) + F(S,T^b)$, where the expectation of the failing events is bounded by,
    \begin{gather*}
        \E[F(S, T^b)] = \sum_{i=1}^{m} \sP[E_{B(s_{\leq i, 1}), B(s_{\leq i, 0})}(t^b)=0] \leq m(2\gamma).
    \end{gather*}
    Similar to Lemma 5.5 in \citet{steinke2023privacy} since $\E[F(S, T^b)]$ is independent of $S$ we could have $F(T^b)=F(S, T^b)$ for $S$ drawn from an appropriate distribution. The proof of the second part of the Proposition is essentially the same except with the following result from Lemma \ref{lemma:baysian_decomposition_pure},
    \begin{gather*}
        \sP[S_i=1, E_{A(s_{<i, 1}), A(s_{< i, 0})}(T^a)=1 | T^a = t^a, S_{<i}=s_{<i}, X_{\leq i}=x_{\leq i}] \leq \frac{e^{c+\epsilon}}{1+e^{c+\epsilon}}, \\
        \E[E_{A(s_{\leq i, 1}), A(s_{\leq i, 0})}(T^a)=1 | S_{<i}=s_{<i}, S_i = 0] \geq 1-\gamma, \\
        \E[E_{A(s_{\leq i, 1}), A(s_{\leq i, 0})}(T^a)=1 | S_{<i}=s_{<i}, S_i = 1] \geq 1-\gamma.
    \end{gather*}
\end{proof}

Similar to Theorem 5.2 \citep{steinke2023privacy} we could compute for the optimal distribution $F(t^b)$ and $F(t^a)$ by formulating it into a linear programming problem. Similar to Corollary \ref{corollary:hyp-test}, to test Hypothesis $\mathcal{H}$ together we would apply Union bound on the two parts of the hypothesis as in Proposition \ref{prop:stat_test_pure}. 

\subsection{Approximate DP case}
\label{appendix:delta_relaxation_approx}

In this section we consider the case where we have a $(c, \gamma)$-close generator and audit for an approximate DP mechanism. We modify the results to account for the failing events from these two sources of randomness. We start by proving a similar version of Lemma \ref{lemma:baysian_decomposition_pure} where we use both the decomposition of $\gG, \gD$ from Lemma \ref{lemma:decomposition_pure} and the decomposition of $\mathcal{P}, \mathcal{Q}$ from Lemma 5.5 \citep{steinke2023privacy}.

\begin{lemma}
\label{lemma:baysian_decomposition_approx}
Let $\gG, \gD$ be probability distributions over $\gY$, fix $c, \gamma \geq 0$, suppose for all measurable $S \subset \gY$ we have $e^{-c}(\gD(S) - \gamma) \leq \gG(S)$. Let $M$ denote an arbitrary $(\epsilon, \delta)$-DP mechanism which takes input $Y$ sampled from distribution $\gG$ or $\gD$ and outputs $f$. Let $\mathcal{P}, \mathcal{Q}$ be probability distributions over $\mathcal{Z}$ which are distributions of $M(Y)$ when $Y$ is sampled from $\gD$ and $\gG$ respectively. Fix $\epsilon, \delta \geq 0$, suppose for all measurable $T \subset \mathcal{Z}$ we have $\mathcal{P}(T) \leq e^{\epsilon}\mathcal{Q}(T) + \delta$ and $\mathcal{Q}(T) \leq e^{\epsilon}\mathcal{P}(T) + \delta$. Let $(\gD, \mathcal{P}), (\gG, \mathcal{Q})$ denote the joint distributions over $(\gY, \mathcal{Z})$ which we assume $\gD \indep \mathcal{P}$, $\gG \indep \mathcal{Q}$. Then there exists a randomized function $E_{(\gD, \mathcal{P}), (\gG, \mathcal{Q})}: (\gY, \mathcal{Z}) \rightarrow \{0, 1\}$ such that for all $y \in \gY$, $M(y) \in \mathcal{Z}$ we have,
\begin{gather*}
    \underset{\substack{X \sim \textrm{Bernoulli}(\frac{1}{2})\\Y \leftarrow X\gD+(1-X)\gG \\ M(Y)\leftarrow M(X\gD+(1-X)\gG)}}{\sP}[X=1, E_{(\gD, \mathcal{P}), (\gG, \mathcal{Q})}(Y, M(Y))=1|Y=y, M(Y)=f] \leq \frac{e^{c+\epsilon}}{1+e^{c+\epsilon}}, \\
    \E_{Y \sim \gD, M(Y) \sim \mathcal{P}}[E_{(\gD, \mathcal{P}), (\gG, \mathcal{Q})}(Y, M(Y))] \geq (1-\gamma)(1-\delta), \\
    \E_{Y \sim \gG, M(Y) \sim \mathcal{Q}}[E_{(\gD, \mathcal{P}), (\gG, \mathcal{Q})}(Y, M(Y))] \geq (1-\gamma)(1-\delta).
\end{gather*}
\end{lemma}
\begin{proof}
    For ease of notation we shorthand $F=M(Y)$, $E=E_{(\gD, \mathcal{P}), (\gG, \mathcal{Q})}$, $\gD'=\gD'(y)$, $\gG'=\gG'(y)$, $\gD''=\gD''(y)$, $\gG''=\gG''(y)$, $\mathcal{P}'=\mathcal{P}'(f)$, $\mathcal{Q}'=\mathcal{Q}'(f)$, $\mathcal{P}''=\mathcal{P}''(f)$, $\mathcal{Q}''=\mathcal{Q}''(f)$ below.
    We define $E: (\gY, \mathcal{Z}) \rightarrow \{0, 1\}$ by,
    \[
        \sP[E(y,f)=1] = \frac{(1-\gamma')(1-\delta')\gD'\mathcal{P}'}{\gD \mathcal{P}}.
    \]
    For any $y \in \gY$, $M(y) \in \mathcal{Z}$ we have,
    \begin{align*}
        &\sP[X=1, E(Y, F)=1 | Y=y, F=f] \\
        &= \frac{\sP[Y=y, F=f|X=1]\sP[X=1]}{\sP[Y=y, F=f]}\sP[E(y, f)=1] \\
        &= \frac{\gD(y) \mathcal{P}(f)}{\gD(y)\mathcal{P}(f)+\gG(y)\mathcal{Q}(f)}\sP[E(y, f)=1] \\
        &= \frac{((1-\gamma')\gD'+\gamma'\gD'')((1-\delta')\mathcal{P}'+\delta'\mathcal{P}'')}{((1-\gamma')\gD'+\gamma'\gD'')((1-\delta')\mathcal{P}'+\delta'\mathcal{P}'')+(1-\gamma')\gG'((1-\delta')\mathcal{Q}'+\delta'\mathcal{Q}'')}\sP[E(y, f)=1] \\
        &= \frac{1 + \frac{\delta'\gD'\mathcal{P}''}{(1-\delta')\gD'\mathcal{P}'}+\frac{\gamma'\gD''\mathcal{P}'}{(1-\gamma')\gD'\mathcal{P}'} + \frac{\gamma'\delta'\gD''\mathcal{P}''}{(1-\gamma')(1-\delta')\gD'\mathcal{P}'}}{1 + \frac{\delta'\gD'\mathcal{P}''}{(1-\delta')\gD'\mathcal{P}'}+\frac{\gamma'\gD''\mathcal{P}'}{(1-\gamma')\gD'\mathcal{P}'} + \frac{\gamma'\delta'\gD''\mathcal{P}''}{(1-\gamma')(1-\delta')\gD'\mathcal{P}'} + \frac{\gG'\mathcal{Q}'}{\gD'\mathcal{P}'} + \frac{\delta'\gG'\mathcal{Q}''}{(1-\delta')\gD'\mathcal{P}'}}\sP[E(y, f)=1] \\
        &\leq \frac{1 + \frac{\delta'\gD'\mathcal{P}''}{(1-\delta')\gD'\mathcal{P}'}+\frac{\gamma'\gD''\mathcal{P}'}{(1-\gamma')\gD'\mathcal{P}'} + \frac{\gamma'\delta'\gD''\mathcal{P}''}{(1-\gamma')(1-\delta')\gD'\mathcal{P}'}}{1+e^{-c}e^{-\epsilon}} \sP[E(y, f)=1] \\
        &= \frac{1}{1+e^{-c}e^{-\epsilon}} \bigg( \frac{(1-\delta')(1-\gamma')\gD'\mathcal{P}' + \delta'(1-\gamma')\gD'\mathcal{P}'' + \gamma'(1-\delta')\gD''\mathcal{P}' + \gamma'\delta'\gD''\mathcal{P}''}{(1-\gamma')(1-\delta')\gD'\mathcal{P}'} \bigg) \\
        & \;\;\; \cdot \sP[E(y, f)=1] \\
        &= \frac{1}{1+e^{-c}e^{-\epsilon}} \bigg( \frac{\gD \mathcal{P}}{(1-\gamma')(1-\delta')\gD'\mathcal{P}'} \bigg)\sP[E(y, f)=1] \\
        &= \frac{1}{1+e^{-c}e^{-\epsilon}} = \frac{e^{c+\epsilon}}{1+e^{c+\epsilon}}.
    \end{align*}
    Using the assumption that $\gD \indep \mathcal{P}$, $\gG \indep \mathcal{Q}$, we account for the expectation of success events as follows,
    \begin{align*}
        \E_{Y \sim \gD, M(Y) \sim \mathcal{P}}[E(Y, F)] &= \iint \gD(y)\mathcal{P}(f) \sP[E=1]dydf \\
        &= \iint \gD(y)\mathcal{P}(f) \frac{(1-\gamma')(1-\delta')\gD'(y)\mathcal{P}'(f)}{\gD(y) \mathcal{P}(f)} dydf \\
        &= (1-\gamma')(1-\delta') \int \gD'(y) dy \int \mathcal{P}'(f) df \\
        &= (1-\gamma')(1-\delta') \geq (1-\gamma)(1-\delta).
    \end{align*}
    Note that we can rewrite the success probability as,
    \[
        \sP[E(y,f)=1] = 1 - \frac{\gamma'(1-\delta')\gD''\mathcal{P}'}{\gD \mathcal{P}} - \frac{(1-\gamma')\delta'\gD'\mathcal{P}''}{\gD \mathcal{P}} - \frac{\gamma'\delta'\gD'' \mathcal{P}''}{\gD \mathcal{P}}.
    \]
    Then we have,
    \begin{align*}
        &\E_{Y \sim \gG, M(Y) \sim \mathcal{Q}}[E(Y, F)] \\
        &= 1 - \gamma'(1-\delta')\E\bigg[\frac{\gD''(y)\mathcal{P}'(f)}{\gD(y)\mathcal{P}(f)} \bigg] - (1-\gamma')\delta' \E \bigg[\frac{\gD'(y)\mathcal{P}''(f)}{\gD(y) \mathcal{P}(f)} \bigg] - \gamma'\delta' \E \bigg[\frac{\gD''(y) \mathcal{P}''(f)}{\gD(y) \mathcal{P}(f)}\bigg] \\
        &= 1 - \gamma'(1-\delta')\int \gG(y) \mathcal{Q}(f) \bigg[\frac{\gD''(y)\mathcal{P}'(f)}{\gD(y)\mathcal{P}(f)} \bigg] dydf \\ & \;\;\; - (1-\gamma')\delta' \int \gG(y) \mathcal{Q}(f) \bigg[\frac{\gD'(y)\mathcal{P}''(f)}{\gD(y) \mathcal{P}(f)} \bigg]dydf \\ & \;\;\; - \gamma'\delta' \int \gG(y) \mathcal{Q}(f) \bigg[\frac{\gD''(y) \mathcal{P}''(f)}{\gD(y) \mathcal{P}(f)}\bigg]dydf \\
        &\geq 1- \gamma'(1-\delta')\int \gD''(y)\mathcal{P}'(f) dydf - (1-\gamma')\delta' \int \gD'(y)\mathcal{P}''(f)dydf \\ & \;\;\; - \gamma'\delta' \int \gD''(y) \mathcal{P}''(f)dydf \\
        &= 1- \gamma'(1-\delta') - (1-\gamma')\delta' - \gamma'\delta' \\
        &= 1-\gamma'-\delta'+\gamma'\delta' \\
        &= (1-\gamma')(1-\delta') \geq (1-\gamma)(1-\delta).
    \end{align*}
\end{proof}

Next we use Lemma \ref{lemma:baysian_decomposition_approx} to prove the following Proposition to form the statistical tests on hypothesis $\mathcal{H}$ for the $(\epsilon, \delta)$-DP version.
\begin{proposition}
\label{prop:stat_test_approx}
Let $\gG$ be $(c,\gamma)$-close, $f$ be output of an $(\epsilon, \delta)$-DP mechanism, $T^b \triangleq B(S, X)$ be the guess from the baseline, $T^a \triangleq A(S, X, f)$ be the guess from the membership audit. Let $T_b, T_a \in [0, 1]^m$ be bounded. Then, for all $v \in \sR$ and all $t$ in the support of $T$:
\begin{gather*}
    \sP_{S, X, T^b}\Big[ \sum_{i=1}^m T^b_i \cdot S_i \geq v \ | \ T^b = t^b \Big] \leq \underset{\substack{S' \sim \textrm{Bernoulli}(\frac{e^c}{1+e^c})^m, F}}\sP\Big[F(t^b)+ \sum_{i=1}^m t^b_i \cdot S'_i \geq v \Big], \\
    \sP_{S, X, T^a}\Big[ \sum_{i=1}^m T^a_i \cdot S_i \geq v \ | \ T^a = t^a \Big] \leq \underset{\substack{S' \sim \textrm{Bernoulli}(\frac{e^{c+\epsilon}}{1+e^{c+\epsilon}})^m, F}}\sP\Big[F(t^a)+ \sum_{i=1}^m t^a_i \cdot S'_i \geq v \Big],
\end{gather*}
where $F$ is independent from $S'$, $F(T^b)$, $F(T^a)$ is supported on $\{0,1,\ldots,m\}$ and $\E_{T^b, F}[F(T^b)] \leq 2m\delta$, $\E_{T^a, F}[F(T^a)] \leq 2m(\gamma+\delta-\gamma\delta)$.
\end{proposition}
\begin{proof}
    The first part of the result (baseline test) exactly follows the proof of Proposition \ref{prop:stat_test_pure}. For the approximate DP case, the proof of the auditor test follows similar steps except for the bound on the expectation of the failing events,
    \begin{align*}
        \E[F(S, T^b)] = \sum_{i=1}^{m} \sP[E_{B(s_{\leq i, 1}), B(s_{\leq i, 0})}(t^b)=0] &\leq m(2(1-(1-\gamma)(1-\delta))) \\&= 2m(\gamma + \delta-\gamma\delta).
    \end{align*}
\end{proof}

Similar to the pure DP case we could optimize for the optimal distribution of $F(t^b)$ and $F(t^a)$ by solving a linear program (Theorem 5.2 in \citet{steinke2023privacy}) and apply Union bound as in Corollary \ref{corollary:hyp-test} to test for Hypothesis $\mathcal{H}$.

% \subsection{Implementation}
% \includegraphics[width=15cm]{images/tmp_implementation.png}

% \begin{table}[h]
%     \centering
%     \resizebox{0.45\textwidth}{!}{
%     \begin{tabular}{ccccc}
%     \hline
%     \textbf{Model} & \textbf{Delta} & \textbf{c\_\{lb\}} & \textbf{eps+c} & \textbf{$\tilde{eps}$} \\ \hline
%     \multirow{3}{*}{\begin{tabular}[c]{@{}c@{}}ResNet18\\ -eps-2\end{tabular}} & 1e-5 &  &  &  \\
%      & 1e-4 &  &  &  \\
%      & 1e-3 &  &  &  \\ \hline
%     \multirow{3}{*}{ResNet18-eps-10} & 1e-5 &  &  &  \\
%      & 1e-4 &  &  &  \\
%      & 1e-3 &  &  &  \\ \hline
%     \end{tabular}
%     \caption{Privacy audit of ResNet18 under different values of $\epsilon$-Differential Privacy using \acronym with different level of relaxations, in which $RM$ is for real member, $RN$ for real non-member and $GN$ for generated (synthetic) non-members.}
%     \label{tab:dpResNet18-relax}
% \end{table}

\subsection{Evaluations}
We empirically evaluate the auditing performance of \acronym under different relaxed level of the generator for both the pure and approximate DP cases. We observe that the overall performance is similar to that of no relaxation. We observe smaller $\tilde{\varepsilon}$ (which is analogous to a looser lower bound) as we increase the relaxation level, and in the more extreme cases where we allow many failing events (e.g. when relaxation $>$ 1e-3) we would fail to detect meaningful privacy leakage as $\tilde{\varepsilon}=0$.

\begin{table}[H]
\centering
\resizebox{0.65\textwidth}{!}{%
\begin{tabular}{ccccc}
\hline \hline
\textbf{Model} & \textbf{Generator Relaxation ($\gamma$)} & $\mathbf{c_{\textnormal{lb}}}$ & $\mathbf{\{\varepsilon + c\}_{\textnormal{lb}}}$ & $\mathbf{\tilde{\varepsilon}}$ \\ \hline
\multirow{4}{*}{\begin{tabular}[c]{@{}c@{}}ResNet18 $\epsilon=2$\end{tabular}} & 0 (no relaxation) & 2.508 & 2.066 & 0 \\
 & 1e-5 & 2.507 & 2.065 & 0 \\
 & 1e-4 & 2.499 & 2.059 & 0 \\
 & 1e-3 & 2.365 & 1.989 & 0 \\ \hline
\multirow{4}{*}{ResNet18 $\epsilon=10$} & 0 (no relaxation) & 2.508 & 2.833 & 0.325 \\
 & 1e-5 & 2.507 & 2.801 & 0.294 \\
 & 1e-4 & 2.499 & 2.570 & 0.071 \\
 & 1e-3 & 2.365 & 1.280 & 0 \\ \hline
\multirow{4}{*}{ResNet18 $\epsilon=15$} & 0 (no relaxation) & 2.508 & 3.661 & 1.153 \\
 & 1e-5 & 2.507 & 3.658 & 1.151 \\
 & 1e-4 & 2.499 & 3.628 & 1.129 \\
 & 1e-3 & 2.365 & 3.312 & 0.947 \\ \hline \hline
% \multirow{4}{*}{ResNet18 $\epsilon=20$} & 0 (no relaxation) &  & 2.277 &  \\
%  & 1e-5 &  & 2.276 &  \\
%  & 1e-4 &  & 2.269 &  \\
%  & 1e-3 &  & 2.189 &  \\ \hline
% \multirow{4}{*}{ResNet18 $\epsilon=\infty$} & 0 (no relaxation) & 2.508 & 3.628 & 1.120 \\
%  & 1e-5 & 2.507 & 3.626 &  \\
%  & 1e-4 & 2.499 & 3.601 &  \\
%  & 1e-3 & 2.365 & 2.988 & \\ \hline
\multirow{3}{*}{\begin{tabular}[c]{@{}c@{}}ResNet18 $\epsilon=2, \delta=$1e-5\end{tabular}}
 & 1e-5 & 2.035 & 2.064 & 0.029 \\
 & 1e-4 & 2.029 & 2.059 & 0.030 \\
 & 1e-3 & 1.904 & 1.988 & 0.084 \\ \hline
\multirow{3}{*}{ResNet18 $\epsilon=10, \delta=$1e-5}
 & 1e-5 & 2.035 & 2.765 & 0.730 \\
 & 1e-4 & 2.029 & 2.541 & 0.512 \\
 & 1e-3 & 1.904 & 1.271 & 0 \\ \hline
\multirow{3}{*}{ResNet18 $\epsilon=15, \delta=$1e-5}
 & 1e-5 & 2.035 & 3.654 & 1.619 \\
 & 1e-4 & 2.029 & 3.625 & 1.596 \\
 & 1e-3 & 1.904 & 3.308 & 1.404 \\ \hline \hline
\end{tabular}%
}
\vspace{0.2cm}
\caption{Privacy audit of ResNet18 under different values of $\epsilon$-DP and $(\epsilon, \delta)$-DP, using \acronym with different level of generator relaxations ($\gamma$), testing over range of recall from 0 to 0.5 without applying a union bound.}
\label{tab:dpResNet18-relax}
\end{table}
\section{Extended Related Work}
\label{appendix:rel_work}

Most related privacy auditing works measure the privacy of an ML model by lower-bounding its privacy loss. 
This usually requires altering the training pipeline of the ML model, either by injecting canaries that act as outliers 
\cite{carlini2019secret} or by using data poisoning attack mechanisms to search for worst-case memorization~\cite{jagielski2020auditing, inst}. 
MIAs are also increasingly used in privacy auditing, to estimate the degree of memorization of member data by an ML algorithm by resampling the target algorithm $\mathcal{M}$ to bound $\frac{P(M|in)}{P(M|out)}$ \cite{Jayaraman2019EvaluatingDP}. 
The auditing procedure usually involves searching for optimal neighboring datasets $D, D'$ and sampling the DP outputs $\mathcal{M}(D), \mathcal{M}(D')$, to get a Monte Carlo estimate of $\epsilon$. 
This approach raises important challenges. 
First, existing search methods for neighboring inputs, involving enumeration or symbolic search, are impossible to scale to large datasets, making it difficult to find optimal dataset pairs. 
In addition, Monte Carlo estimation requires up to thousands of costly model retrainings to bound $\epsilon$ with high confidence.
Consequently, existing approaches for auditing ML models predominantly require the re-training of ML models for every (batch of) audit queries, which is computationally expensive in large-scale systems~\cite{jagielski2020auditing, zanellabéguelin2022bayesian, lu2023general}. 
This makes privacy auditing computationally expensive and gives an estimate by averaging over models, which might not reflect the true guarantee of a specific pipeline deployed in practice. 

Nonetheless, improvements to auditing have been made in a variety of directions. 
For example,~\citet{nasr2023tight} and~\citet{maddock2023canife} have taken advantage of the iterative nature of DP-SGD, auditing individual steps to understand the privacy of the end-to-end algorithm. 
% Advances have also occurred with respect to the statistical techniques used for estimating the $\epsilon$ parameter by using Log-Katz confidence intervals~\cite{lu2023general}, Bayesian techniques~\cite{zanellabéguelin2022bayesian} or auditing algorithms with different privacy definitions~\cite{nasr2023tight}. 
~\citet{andrew2023oneshot} leverage the fact that analyzing MIAs for non-member data does not require re-running the algorithm. 
Instead, it is possible to re-sample the non-member data point: if the data points are i.i.d. from an asymptotically Gaussian distribution with mean zero and variance $1/d$, this enables a closed-form analysis of the non-member case.

Recently, \citet{steinke2023privacy} proposed a novel scheme for auditing differential privacy with $O(1)$ training rounds. 
This approach enables privacy audits using multiple training examples from the same model training, if examples are included in training independently (which requires control over the training phase, and altering the target model).
%The theoretical analysis shows that true positive rates (TPR) and false positive rates (FPR) estimates and provide proof of the validity of this heuristic showing that standard MIAs, which attack multiple examples per training run, can be used for auditing analysis. Note that prior work using these attacks requires making an independence assumption. As a consequence, auditing can take advantage of progress in the membership inference field \cite{LIRA, wen2023canary}.
They demonstrate the effectiveness of this new approach on DP-SGD, in which they achieve meaningful empirical privacy lower bounds by training only one model (the strongest results are achieved by including canaries in the training set though, which is not possible in our setting), whereas standard methods would require training hundreds of models.
Our work builds closely on the theory from \citet{steinke2023privacy}, but introduces a baseline model to account for distribution shifts.
The key difference that enables our work to account for member/non-member distribution shifts lies in how we create the audit set in our privacy game (\cref{def:auditing-game}). In \citet{steinke2023privacy}, the audit set is fixed, and data points are randomly assigned to member or non-member by a Bernoulli random variable $S$. Members are actually used in training the target model $f$, while non-members are not (so assignment happens before training). In our framework, we take a set of known members (after the fact), and pair each point with a non-member(generated i.i.d. from the generator distribution). We then flip $S$ to sample which data point of each pair will be shown to the ``auditor'' (MIA/baseline) for testing, thereby creating the test task of our privacy measurement. When we replace our generated data with in-distribution independent non-members (\acronym RM;RN), we exactly enforce the same independence as \citet{steinke2023privacy} (and as a result have $c=0$, and $\tilde{\epsilon}$ is a lower-bound on $\epsilon$), except that we ``waste'' some member data by drawing our auditing game after the fact. The difference in analysis accounts for distribution shifts when our non-members are generated.

The field of MIA without direct connections to privacy leakage measurement has also seen a lot of recent activity, with new proposals to improve the strength of MIAs in various settings \citep{carlini2022membership,nasr2019comprehensive,suri2024parameters,zarifzadeh2023low}.
While it is not the focus of this paper, an interesting avenue for future work is to port ideas from these new MIAs to improve \acronym's MIA and baseline. Such transfer of MIA progress could lead to more powerful measurements with \acronym. 
More closely related to our proposal, a recent set of works on \cite{das2024blind,duan2024membership,meeus2024inherent} MIAs on foundation models has observed that current evaluations are limited due a lack of availability of non-member data. Indeed, foundation models typical include all known data at the time of their training, and there is no well known public set of in-distribution non-member data. Evaluation tasks typical rely on more recent data, which is known to consist in non-members, but suffers from distribution shifts. \citet{das2024blind,duan2024membership,meeus2024inherent} all show that such a shift invalidates MIA evaluations by showing that ``blind attacks'' (i.e., our baselines) perform better than proposed MIAs.
One could apply our proposed framework to such a setting: when the baseline in our framework performs better than the MIA, \acronym returns a privacy loss measurement of zero instead of misleading measurements of MIA performance.
We thus believe that \acronym is an important step towards a theory for rigorous evaluations of membership inference for MIA models.

\newpage

\end{document}